\pdfoutput=1
\documentclass{lmcs}


\usepackage[english]{babel}
\usepackage[utf8]{inputenc}
\usepackage[T1]{fontenc}

\usepackage{tikz}
\usetikzlibrary{cd}
\usetikzlibrary{decorations.pathmorphing}

\usepackage{xcolor}
\usepackage{multirow}
\usepackage{marginnote}

\usepackage{mathpartir}

\usepackage{fontawesome}

\usepackage{stmaryrd}
\usepackage{amsmath}
\usepackage{amsthm}
\usepackage{amssymb}
\usepackage{mathtools}
\usepackage{environ}
\usepackage{bm}
\usepackage{bbm}

\usepackage[nameinlink]{cleveref}

\newcommand{\renewtheorem}[1]{%
  \expandafter\let\csname #1\endcsname\relax
  \expandafter\let\csname c@#1\endcsname\relax
  \expandafter\let\csname end#1\endcsname\relax
  \newtheorem{#1}%
}

\theoremstyle{plain}

\renewtheorem{thm}{Theorem}[section]
\renewtheorem{cor}[thm]{Corollary}
\renewtheorem{lem}[thm]{Lemma}
\renewtheorem{slem}[thm]{Sublemma}
\renewtheorem{prop}[thm]{Proposition}
\renewtheorem{asm}[thm]{Assumption}

\theoremstyle{definition}

\renewtheorem{rem}[thm]{Remark}
\renewtheorem{rems}[thm]{Remarks}
\renewtheorem{exa}[thm]{Example}
\renewtheorem{exas}[thm]{Examples}
\renewtheorem{defi}[thm]{Definition}
\renewtheorem{conv}[thm]{Convention}
\renewtheorem{conj}[thm]{Conjecture}
\renewtheorem{prob}[thm]{Problem}
\renewtheorem{oprob}[thm]{Open Problem}
\renewtheorem{oprobs}[thm]{Open Problems}
\renewtheorem{algo}[thm]{Algorithm}
\renewtheorem{obs}[thm]{Observation}
\renewtheorem{desc}[thm]{Description}
\renewtheorem{fact}[thm]{Fact}
\renewtheorem{qu}[thm]{Question}
\renewtheorem{oqu}[thm]{Open Question}
\renewtheorem{pty}[thm]{Property}
\renewtheorem{clm}[thm]{Claim}
\renewtheorem{nota}[thm]{Notation}
\renewtheorem{com}[thm]{Comment}
\renewtheorem{coms}[thm]{Comments}

\theoremstyle{defC}
\renewtheorem{defiC}[thm]{Definition}

\theoremstyle{thmC}
\renewtheorem{thmC}[thm]{Theorem}
\renewtheorem{propC}[thm]{Proposition}
\renewtheorem{lemC}[thm]{Lemma}

\theoremstyle{plain}\newtheorem*{thm*}{Theorem}
\theoremstyle{definition}\newtheorem{con}[thm]{Construction}
\theoremstyle{definition}\newtheorem*{exas*}{Examples}

\DeclareRobustCommand{\defiEnd}{%
  \leavevmode\unskip\penalty9999 \hbox{}\nobreak\hfill
  \quad\hbox{$\lrcorner$}%
}

\DeclareFontFamily{U}{mathc}{}
\DeclareFontShape{U}{mathc}{m}{it}%
{<->s*[1.03] mathc10}{}
\DeclareMathAlphabet{\mathcal}{U}{mathc}{m}{it}

\DeclareFontFamily{U}{min}{}
\DeclareFontShape{U}{min}{m}{n}{<-> udmj30}{}

\newcommand{\defemph}[1]{\textbf{#1}}

\makeatletter
\newsavebox{\@brx}
\newcommand{\llangle}[1][]{\savebox{\@brx}{\(\m@th{#1\langle}\)}%
  \mathopen{\copy\@brx\kern-0.5\wd\@brx\usebox{\@brx}}}
\newcommand{\rrangle}[1][]{\savebox{\@brx}{\(\m@th{#1\rangle}\)}%
  \mathclose{\copy\@brx\kern-0.5\wd\@brx\usebox{\@brx}}}
\makeatother

\DeclareMathOperator{\colima}{colim}
\newcommand{\colim}{\mathop{\colima}}

\newcommand{\cxlim}{\operatorname{cxlIm}}

\newcommand{\coim}{\operatorname{coim}}
\newcommand{\cxl}{\operatorname{cxl}}

\newcommand{\Ra}{\Rightarrow}

\newcommand{\La}{\Leftarrow}

\newcommand{\rat}{\rightarrowtail}

\newcommand{\abs}[1]{{\left| #1 \right|}}

\newcommand{\Ob}{\mathsf{Ob}}

\newcommand{\id}{\mathsf{id}}
\newcommand{\op}{\mathsf{op}}

\newcommand{\type}{\mathsf{type}}

\newcommand{\Ty}{\mathsf{Ty}}
\newcommand{\ty}{\mathsf{ty}}
\newcommand{\Tm}{\mathsf{Tm}}
\newcommand{\tm}{\mathsf{tm}}

\newcommand{\Nat}{\mathbb{N}}

\newcommand{\Id}{\mathsf{Id}}
\newcommand{\J}{\mathsf{J}}
\newcommand{\refl}{\mathsf{refl}}

\newcommand{\Singl}{\mathsf{Singl}}

\newcommand{\Eq}{\mathsf{Eq}}

\newcommand{\app}{\mathsf{app}}
\newcommand{\lam}{\mathsf{lam}}
\newcommand{\happly}{\mathsf{happly}}
\newcommand{\funext}{\mathsf{funext}}

\newcommand{\Unit}{\mathbf{1}}

\newcommand{\Init}{\mathbf{0}}

\newcommand{\Quot}{\mathbf{Q}}
\newcommand{\quot}{\mathbf{q}}

\newcommand{\UU}{\mathcal{U}}

\newcommand{\El}{\mathsf{El}}

\newcommand{\yo}{\text{\usefont{U}{min}{m}{n}\symbol{'210}}}

\newcommand{\cof}{\mathsf{cof}}

\newcommand{\Th}{\mathbb{T}}

\newcommand{\MC}{\mathbb{C}}
\newcommand{\MD}{\mathbb{D}}
\newcommand{\ME}{\mathbb{E}}

\newcommand{\SProp}{\mathsf{Prop}}
\newcommand{\SSet}{\mathsf{Set}}

\newcommand{\SPsh}{\mathsf{Psh}}
\newcommand{\SRepPsh}{\mathsf{RepPsh}}

\newcommand{\CSet}{\mathbf{Set}}

\newcommand{\CPsh}{\mathbf{Psh}}

\newcommand{\CCwf}{\mathbf{CwF}}
\newcommand{\CA}{\mathcal{A}}
\newcommand{\CB}{\mathcal{B}}
\newcommand{\CC}{\mathcal{C}}
\newcommand{\CD}{\mathcal{D}}
\newcommand{\CE}{\mathcal{E}}

\newcommand{\CL}{\mathcal{L}}
\newcommand{\CM}{\mathcal{M}}

\newcommand{\CQ}{\mathcal{Q}}
\newcommand{\CR}{\mathcal{R}}

\newcommand{\CMod}{\mathbf{Mod}}

\author{Rafaël Bocquet}
\address{Department of Programming Languages and Compilers, Eötvös Loránd University, Budapest, Hungary}
\email{bocquet@inf.elte.hu}
\urladdr{\url{https://rafaelbocquet.gitlab.io/}}
\thanks{The author was supported by the European Union, co-financed by the European Social Fund (EFOP-3.6.3-VEKOP-16-2017-00002).}

\title{Coherence of strict equalities in dependent type theories}

\begin{document}

\begin{abstract}
  We study the coherence and conservativity of extensions of dependent type theories by additional strict equalities.
  By considering notions of congruences and quotients of models of type theory, we reconstruct Hofmann's proof of the conservativity of Extensional Type Theory over Intensional Type Theory.
  We generalize these methods to type theories without the Uniqueness of Identity Proofs principle, such as variants of Homotopy Type Theory, by introducing a notion of higher congruence over models of type theory.
  Our definition of higher congruence is inspired by Brunerie's type-theoretic definition of weak $\infty$-groupoid.
  For a large class of type theories, we reduce the problem of the conservativity of equational extensions to more tractable acyclicity conditions.
\end{abstract}

\maketitle

\section{Introduction}

Equality and computation are central components of type theories.
The computational content of a type theory is presented by elimination rules (often called $\beta$-rules), and perhaps uniqueness rules (usually called $\eta$-rules) or more exotic rules (such as the $\nu$-rules considered in \cite{NewEqsNeutrals}).
This computational content is typically explained by the means of a normalization algorithm.
In presence of identity types, there is a distinction between two notions of equality between terms of a type theory.
Internally, the identity types provide the notion of \emph{internal equality}, also often called \emph{propositional equality}, or sometimes \emph{typal equality} to emphasize that it does not have to be truncated.
Externally, we can also compare terms up to \emph{strict equality}, which is the proof-irrelevant equality of our metatheory.
Strict equality is also often called \emph{definitional} or \emph{judgemental} equality.

When working internally to a type theory, it is desirable to have as many strict equalities as possible.
Indeed, equalities that hold strictly are equalities that can be implicitly and silently coerced over.
On the other hand, equalities that are only internal require explicit transports and coercions, which quickly clutter the terms of the theory.
Conversely, the trade-off is that type theories with additional strict equalities have fewer models, and their semantics are therefore more complicated.

Hofmann proved in \cite{HofmannCons} a conservativity theorem, showing that all equalities in a type theory can conservatively be made strict, in the presence of enough extensionality principles in the base type theory.
The most important of these extensionality principles is the Uniqueness of Identity Proofs (UIP) principle, which states that any two proofs of a weak equality are themselves weakly equal.
A more syntactic proof was later given by Oury~\cite{OuryConservativity} for the calculus of constructions.
Oury's proof had some issues, mainly due to a presentation of the syntax of type theory with too few annotations.
An improvement of Oury's proof and a presentation of this result as a constructive and effective syntactic translation has been given recently by Winterhalter et al~\cite{ElimRefl}.

Since Hofmann's proof of conservativity, there has been a lot of interest going into the study of type theories with non-trivial higher dimensional content inconsitent with UIP~\cite{GroupoidModel}, and their semantics in homotopy theoretic~\cite{HomotopyTheoreticModels, SimplicialModel} and $\infty$-categorical structures~\cite{LangLexInftyCats}.
For type theories without UIP, strict equalities are even more important, because they are automatically coherent.
Thus having more strict equalities means that we escape not only ``transport hell'', but also ``higher-dimensional transport and coherence hell''.
Conversely, it is in practice much harder to justify strict equalities in many homotopy theoretic models.
Some authors have even considered weak variants of the basic computation rules of type theories.
For instance, weak identity types, whose computation rule only holds up to internal equality, have been introduced by \cite{PropIdTypes}, under the name of propositional identity types.
The path types of cubical type theories~\cite{CTT} also only satisfy weakly the computation rule of identity types.
Other type structures can be weakened similarly, and we can even consider type theories whose computation rules are all expressed by internal equalities instead of strict equalities.
At the level of types and universes, instead of assuming that each type former is strictly classified by some code in the universe, we can ask for them to be classified only up to type equivalence.
These weak Tarski universes have been introduced in \cite{WeakTarski}.
The fact that homotopy type theory with strict univalent universes, rather than weak universes, can be interpreted in every $(\infty, 1)$-topos has only been established recently~\cite{InftyToposesUniverses}.

In this setting, we can wonder how type theories with varying amounts of strict equalities can be compared.
More precisely, we wish to know how to establish coherence and strictification theorems, that would allow us, when working internally to a model of a weak type theory, to pretend that it satisfies more strict equalities than it actually does, by replacing it by an equivalent stricter model.
The question of the conservativity of strong identity types over weak identity types has been asked at the TYPES 2017 conference~\cite{WeakJTypes}, motivated by the fact that the path types in cubical type theory only satisfy the elimination principle of weak identity types.
This was also the original motivation for the present paper.

We give some examples of weakenings and extensions of homotopy type theory that ought to be equivalent to standard homotopy type theory.
We believe that the coherence theorems presented in this paper brings the proofs of these equivalences within reach.
\begin{exas}\label{exas:applications} \hfill
  \begin{itemize}
    \item Weakening the $\beta$ and $\eta$ computation rules of identity types, $\Sigma$-types, inductive types, etc, gives a weaker variant of HoTT.
    \item We can add strict equalities that make the addition on natural numbers into a strictly associative and commutative operation.
      That is, while the inductive definition of $(- + -) : \Nat \to \Nat \to \Nat$ only satisfies the strict equalities $0 + y = y$ and $(\mathsf{suc}\ x) + y = \mathsf{suc}\ (x + y)$, we would add the strict equalities $x + 0 = 0$, $x + (\mathsf{suc}\ y) = \mathsf{suc}\ (x + y)$, $x + y = y + x$, $(x + y) + z = x + (y + z)$, etc.
    \item Similarly, we could make the composition of equalities into a strictly associative operation, optionally with strict inverses.
    \item We can extend the theory with a universe of strict proposition~\cite{SProp} that is equivalent to the universe of homotopy propositions.
    \item Similarly, we can extend the theory with universes of strict categories, strict rings, etc, that satisfy strictly the equations of the theories of categories, rings, etc.
    \item We can extend the theory with a universe of ``strictly'' pointed types $\mathsf{SPtType}$, equivalent to the universe of pointed types $\mathsf{PtType} \triangleq (A : \UU) \times A$, with a smash product operation $(- \wedge -) : \mathsf{SPtType} \to \mathsf{SPtType} \to \mathsf{SPtType}$ with more strict equalities than the smash product of $\mathsf{PtType}$.
      This would provide an alternative interpretation of Brunerie's rewriting based method to prove that the smash product is a symmetric monoidal product on pointed types~\cite{BrunerieSmash}.
  \end{itemize}
\end{exas}

Some progress has been made by Isaev in \cite{IsaevMorita}.
In that paper, Isaev defines the notion of Morita equivalence between type theories, and gives some characterizations of that notion.
A first conservativity result in the absence of UIP is also proven, showing that type theories with weak or strong unit types are Morita equivalent.

The constructions by Isaev~\cite{IsaevMS} and Kapulkin and Lumsdaine~\cite{HoThTT,HoInvCwA}, of Quillen model or semi-model structures over the categories of models of type theories, are also extremely relevant for our work.
In particular, as remarked in \cite{HoThTT}, Hofmann's conservativity theorem proves exactly that the morphism $\Init_{\mathsf{ITT}} \to \Init_{\mathsf{ETT}}$ between the initial models of intensional type theory and extensional type theory is a trivial fibration of their semi-model structure.
The weak equivalences of the same semi-model structure correspond to a weaker notion of conservativity than trivial fibrations.
Isaev's definition of Morita equivalence relies on that notion of weak equivalence.

This paper builds on top of the aforementioned work.
While Isaev considers the notion of Morita equivalence for arbitrary morphisms between type theories, we restrict our attention to the equational extensions of a weak type theory $\Th_{w}$ to a strong type theory $\Th_{s}$, by a family of equations $\Th_{e}$, which should hold weakly in $\Th_{w}$ and strictly in $\Th_{s}$.
We then establish sufficient conditions for the theories $\Th_{w}$ and $\Th_{s}$ to be Morita equivalent.

The situation can be compared to other well-known coherence theorems, such as Mac Lane's coherence theorem for monoidal categories~\cite{MacLaneCoh}.
They can often be stated in multiple different ways.
For example, here are two related ways to state the coherence theorem for monoidal categories.
\begin{enumerate}
  \item\label{itm:coh_moncat_1} Every (weak) monoidal category is monoidally equivalent to a strict monoidal category.
  \item\label{itm:coh_moncat_2} In a freely generated monoidal category, every diagram made up of associators and unitors commutes.
\end{enumerate}

The statement (\ref{itm:coh_moncat_1}) is generally the one that we want to use: it allows us to work with any weak monoidal category as if it was strict.
The statement (\ref{itm:coh_moncat_2}) is however perhaps easier to prove, because free monoidal categories can be seen as syntactic objects, that are relatively easy to describe explicitly and understand.
See \cite{joyal1991geometry} for a proof of the statement (\ref{itm:coh_moncat_1}) that relies on the statement (\ref{itm:coh_moncat_2}).
In the case of monoidal category, it is actually possible to prove the statement (\ref{itm:coh_moncat_1}) more directly using representation theorems similar to the Yoneda lemma.
This kind of approach does not seem suitable for the coherence theorems that we are interested in.

The main result of this paper is a coherence theorem for type theories that is analogous to the fact the statement (\ref{itm:coh_moncat_1}) can be deduced from the statement (\ref{itm:coh_moncat_2}).
It states that to establish the conservativity of the extension of a weak type theory $\Th_{w}$ to a strong type theory $\Th_{s}$ by a family of equations $\Th_{e}$, it suffices to check, for every cellular model $\CC$ of $\Th_{w}$ (a cellular model is a model that is freely generated by some types and terms), that the higher congruence on $\CC$ freely generated by the equations of $\Th_{e}$ exists and is acyclic.
The acyclicity condition encodes the same idea as the fact that every diagram made up of associators and unitors commutes in a freely generated monoidal category.

The main problem lies in the details of the definition of the notion of higher congruence.
An ordinary congruence over a model $\CC$ consists of equivalence relations on the families of types and terms of $\CC$, that should be preserved by all type-theoretic operations.
Equivalently, a congruence can be seen as an extension of the set-valued model $\CC$ to a model valued in setoids.
A higher congruence over $\CC$ should instead be an extension of $\CC$ to a model valued in weak $\infty$-groupoids, or spaces.

Defining higher congruences requires choosing a model of weak $\infty$-groupoids among many.
Our solution is to base our definition of higher congruences on a reformulation of Brunerie's type-theoretic definition of weak $\infty$-groupoid~\cite[Appendix A]{BrunerieThesis}.
We note that Brunerie $\infty$-groupoids are known to be equivalent to the other models of spaces, thanks to work by Henry~\cite{HenryHoHyp}.
Using other models of weak $\infty$-groupoids, such as simplicial or cubical Kan complexes, could also potentially work, but using a type-theoretic definition seems to make the shapes of different objects involved (models of the base type theories and higher congruences) match up.
Concretely, we will define a new type theory $\Th_{w,2}$ extending the weak type theory $\Th_{w}$, and define a higher congruence over a model $\CC$ of $\Th_{w}$ to be a model $\CD$ of $\Th_{w,2}$ whose underlying model of $\Th_{w}$ is equivalent to $\CC$.
The higher congruences freely generated by some equations then become the initial models of some type theories, which are syntactic objects that can be handled using standard type-theoretic techniques (such as parametricity, logical relations, etc).

We don't give any application of our coherence theorem in this paper, which is instead focused on proving general results that hold for a wide class of type theories.
Another paper is in preparation with proofs of acyclicity for type theories with weak variants of the standard type-theoretic structures (identity types, $\Pi$-types, $\Sigma$-types, inductive types, universe, ...).
We also hope to include some of the other examples of \cref{exas:applications}.
The proof of acyclicity relies on ideas from a paper of Lasson~\cite{Lasson14}, which proves the canonicity of the weak $\infty$-groupoid laws definable in Brunerie's type theory.
In our setting that result can be reinterpreted as a proof of the fact that the higher congruence that is freely generated by the empty family of equations is acyclic.
For higher congruences generated by non-empty families of equations, Lasson's construction can be combined with a normalization proof.
As a general heuristic, we expect acyclicity, and hence coherence and conservativity, to hold whenever the strong type theory admits a well-behaved normalization algorithm.

\subsection*{Outline of the paper}

In \cref{sec:background}, we introduce our notations and conventions, and review the notion of Category with Families (CwFs) and associated definitions (contextual CwFs, cumulative CwFs and families of telescopes).
We make extensive use of the internal language of presheaf categories to describe our constructions.
We also give a formal definition of type theory signature, extending the notion of QIIT-signature of \cite{QIITs}, although we only use it informally in the rest of the paper.

In \cref{sec:weak_id_types}, we define the structures of weak identity types and weak $\Pi$-types, and derive some basic tools that are necessary to work with them.
In particular, we don't assume that our type theories include $\Sigma$-types, but we prove results that allow us to work as if we had $\Sigma$-types.

In \cref{sec:homotopy_theory} we recall, and adapt to our setting, the classes of maps of the semi-model structure introduced in \cite{HoThTT}.
We also study further the trivial fibrations, which are defined by some surjectivity conditions, and show that they correspond (up to (contextual) isomorphism) to the quotients of a well-behaved class of congruences, which we call fibrant congruences.

In \cref{sec:equivalence_models} we define the notion of equational extension of a theory by a family of internal equalities and recall the notion of Morita equivalence between type theories of \cite{IsaevMorita}.

In \cref{sec:conservativity_strict} we use the notion of fibrant congruence and its relationship with trivial fibrations to obtain characterizations of Morita equivalences of type theories for strict type theories, i.e. type theories that satisfy the UIP principle.
We obtain the following variant of Hofmann's conservativity theorem.
\begin{thm*}[Simplified statement of \cref{thm:conservativity_strict}]
  Let $\Th_{w}$ be a type theory with a cumulative hierarchy of universes and weak identity types satisfying the UIP principle.
  Let $\Th_{s}$ be the extension of $\Th_{w}$ with the equality reflection rule.

  If either of the following two conditions is verified, then $\Th_{w}$ and $\Th_{s}$ are Morita equivalent.
  \begin{enumerate}
    \item The theory $\Th_{w}$ includes $\Pi$-types with a strict $\beta$-rule.
    \item The category $\CMod_{w}^{\cxl}$ of contextual models of $\Th_{w}$, equipped with the classes of weak equivalences, fibrations and cofibrations defined in \cref{sec:homotopy_theory}, is a semi-model category.
      \defiEnd
  \end{enumerate}
\end{thm*}

In \cref{sec:conservativity_nonstrict} we introduce our notion of type-theoretic higher congruence, and use it to obtain characterizations of Morita equivalences for equational extensions of type theories.

Given any type theory $\Th_{w}$, we define a type theory $\Th_{w,2}$ extending $\Th_{w}$, such that ordinary models of $\Th_{w,2}$ correspond to models of $\Th_{w}$ valued in $\infty$-groupoids.

Given any equational extension $\Th_{e}$ over $\Th_{w}$, we will define a further extension $\Th_{w,2,e}$ of $\Th_{w,2}$.
The left adjoint $L : \CMod_{w} \to \CMod_{w,2,e}$ of the adjunction between the categories of models of $\Th_{w}$ and $\Th_{w,2,e}$ can be seen as a functor associating to every model $\CC$ of $\Th_{w}$ the higher congruence over $\CC$ freely generated by the equations of $\Th_{e}$.

Using these notions, we prove the following coherence theorem.
\begin{thm*}[Simplified statement of \cref{thm:conservativity_nonstrict}]
  Let $\Th_{w}$ be a type theory with a cumulative hierarchy of universes and weak identity types, and let $\Th_{e}$ be a family of internal equalities of $\Th_{w}$.
 
  Let $\Th_{s}$ be the extension of $\Th_{w}$ obtained by making the internal equalities of $\Th_{e}$ strict.

  If for every cellular (i.e. freely generated) model $\CC$ of $\Th_{w}$, the morphism $\CC \to L_{w,2,e}\ \CC$ is a weak equivalence and $L_{w,2,e}\ \CC$ is acyclic, then $\Th_{w}$ and $\Th_{s}$ are Morita equivalent.
  \defiEnd
\end{thm*}

In \cref{sec:first_weq}, we show that the first condition of \cref{thm:conservativity_nonstrict} holds as soon as the theory $\Th_{w}$ has $\Pi$-types with a strict $\beta$-rule.
We obtain the following coherence theorem, which is the main theorem of this paper.
\begin{thm*}[Simplified statement of \cref{thm:conservativity_nonstrict_pi}]
  Let $\Th_{w}$ be a type theory with a cumulative hierarchy of universes and weak identity types, and let $\Th_{e}$ be a family of internal equalities of $\Th_{w}$.
  We assume that $\Th_{w}$ also includes $\Pi$-types with a strict $\beta$-rule.
  We also assume that $\Th_{e}$ includes the computation rules of the weak identity types of $\Th_{w}$.
  Let $\Th_{s}$ be the extension of $\Th_{w}$ obtained by making the internal equalities of $\Th_{e}$ strict.

  If for every cellular (i.e. freely generated) model $\CC$ of $\Th_{w}$, $L_{w,2,e}\ \CC$ is acyclic, then $\Th_{w}$ and $\Th_{s}$ are Morita equivalent.
  \defiEnd
\end{thm*}

\subsection*{Agda formalization}
Some of our constructions are expressed in the type-theoretic internal language of presheaf categories, and have been formalized in Agda. The Agda development can be found in the files attached to the arXiv version of the paper or at \url{https://rafaelbocquet.gitlab.io/Agda/CoherenceStrict/}.

\section{Background}\label{sec:background}

We recall in this section the semantics of type theories in categories with families (CwFs), and introduce the tools and notations that we will use in this paper.
We will make use in particular of the internal type-theoretic language of the presheaf category $\CPsh\ \CC$ as a tool to define and work with type-theoretic structures over a CwF $\CC$.

\subsection{Metatheory and basic notations}
\begin{enumerate}
  \item We assume that a sufficiently large hierarchy of universes is available in the ambient metatheory.
    The sets of dependent function are written $(a : A) \to B(a)$, and dependent functions are introduced by $(a : A) \mapsto b(a)$.
    We often use braces $\{\}$ to indicate implicit arguments.
    For instance, given a function $f : \{a : A\} \to B\ a \to C\ a$, and elements $a : A$ and $b : B\ a$, we will just write $f\ b$ for the application of $f$ to $a$ and $b$.
    We write $f\ \{a\}\ b$ or $f_{a}\ b$ when we want to make the argument $a$ explicit.

    The sets of dependent pairs are written $(a : A) \times B(a)$, and dependent pairs are introduced by $(a,b)$.
  \item We assume the axiom of choice.
    As currently formulated, our results do not hold constructively, as we will note in \cref{rem:split_equiv}.
    We believe that they could be reformulated so as to hold constructively.
    Alternatively, it should be possible to bypass the non-constructive parts in the case of type theories with decidable equality and decidable type checking.
  \item We denote the set of objects of a small category $\CC$ by $\abs{\CC}$, and the set of morphisms between $x,y : \abs{\CC}$ by $\CC(x \to y)$.
    We may also quantify over objects of a category using either $(x : \CC)$ or $(x : \CC^{\op})$ instead of $(x : \abs{\CC})$; in that case, it is often understood that all constructions depending on $x$ are covariantly or contravariantly natural (or functorial) in $x$.
  \item We denote the composition of morphisms $f : \CC(A \to B)$ and $g : \CC(B \to C)$ by either $(f \cdot g) : \CC(A \to C)$ or $(g \circ f) : \CC(A \to C)$.
    The diagrammatic composition order $(f \cdot g)$ is preferred in presence of commutative diagrams, contravariant actions on the left (e.g. $f^{\star}\ (g^{\star}\ X) = (f \cdot g)^{\star}\ X$) and covariant actions on the right.
    The standard composition order $(g \circ f)$ is usually used in presence of contravariant actions on the right (e.g. $X[f][g] = X[f \circ g]$) and covariant actions on the left (e.g. $(f \circ g)(x) = f(g(x))$).
  \item Given a category $\CC$, we denote the slice category over an object $x : \abs{\CC}$ by $(\CC / X)$, and the coslice category under an object $x : \abs{\CC}$ by $(C \backslash x)$.
\end{enumerate}

\subsection{Internal language of presheaf categories}
Fix a base category $\CC$.
We recall how the presheaf category $\CPsh\ \CC$ is given the structure of a model of extensional type theory.
We refer the reader to \cite{SyntaxAndSemantics,LiftingUniverses} for a more detailed presentation of this structure.

\begin{enumerate}
  \item The objects of $\CPsh\ \CC$ are presheaves over $\CC$, i.e. functors from $\CC^{\op}$ to $\CSet$, and the morphisms are natural transformations.
    Given a presheaf $X : \CPsh\ \CC$, we denote by $\abs{X}_{\Gamma} : \SSet$ its component at an object $\Gamma : \abs{\CC}$, and write $x[f] : \abs{X}_{\Gamma}$ for the restriction of an element $x : \abs{X}_{\Delta}$ by a morphism $f : \CC(\Gamma \to \Delta)$.
    We may also occasionally write $\abs{X}_{f} : \abs{X}_{\Delta} \to \abs{X}_{\Gamma}$ for the restriction operation.
  \item A type of $\CPsh\ \CC$ over a presheaf $X : \CC$ is a dependent presheaf over $X$, or equivalently a presheaf over the category of elements $(\CC/X)$.
    Given a dependent presheaf $Y$ over $X$, we may denote its component at an object $\Gamma : \abs{\CC}$ by $\abs{Y}_{\Gamma} : (x : \abs{X}_{\Gamma}) \to \SSet$.
    We may also introduce a dependent presheaf $Y$ over $X$ by writing
    \[ Y : \{\Gamma : \CC^{\op}\} \to \abs{X}_{\Gamma} \to \SSet. \]
    A term of type $Y$ over a presheaf $X : \CC$ is a dependent natural transformation from $X$ to $Y$, or equivalently a global element of $Y$ when seen as a presheaf over $(\CC/X)$.
    We may introduce a dependent natural transformation $y$ from $X$ to $Y$ by writing
    \[ y : \{\Gamma : \CC^{\op}\} \to (x : \abs{X}_{\Gamma}) \to \abs{Y}_{\Gamma}\ x. \]
  \item The Yoneda embedding is written $\yo : \CC \to \CPsh\ \CC$.
    The presheaf represented by an object $\Gamma : \abs{\CC}$ is written $\yo_{\Gamma} : \CPsh\ \CC$.
  \item The presheaf universe $\SPsh_{\CC,i}$ is the classifier of $i$-small dependent presheaves.
    Given any presheaf $X$, a global element of $\SPsh_{\CC,i}$ over $X$ is a dependent presheaf over $X$.
    Its definition can be computed using the Yoneda lemma: for every object $\Gamma : \CC$, $\abs{\SPsh_{\CC,i}}_{\Gamma}$ has to be (isomorphic to) the set of $i$-small dependent presheaves over the representable presheaf $\yo_{\Gamma}$.

    We will ignore most size issues and omit universe levels in this paper, and write just $\SPsh_{\CC}$ for the presheaf universe at any universe level.
  \item Most type structures available externally, such as $\Pi$-types, $\Sigma$-types, quotients and indexed inductive types are also available in the presheaf model $\CPsh\ \CC$.
    The presheaf universes $\SPsh_{\CC}$ are closed under those.
    We use the same notations internally to $\CPsh\ \CC$ as in the external metatheory, e.g. $(a : A) \to B\ a$ for $\Pi$-types, etc.

    Even though we assume the axiom of choice in our external metatheory, it may not hold internally to $\CPsh\ \CC$.
  \item The presheaf model $\CPsh\ \CC$ supports extensional equality types; we can reason about equality internally in the same way as we do externally.
\end{enumerate}
We will also need the notion of locally representable dependent presheaf, which is the semantic counterpart of the syntactic notion of context extension.
\begin{defi}
  A dependent presheaf $Y : \CPsh\ (\CC/X)$ is said to be \defemph{locally representable} if for every $\Gamma : \abs{\CC}$ and $x : \abs{X}_{\Gamma}$, the presheaf $Y_{x} : \CPsh\ (\CC/\Gamma)$ defined by \[ (\CC/\Gamma)^{\op} \ni (\rho : \CC(\Delta \to \Gamma)) \mapsto \abs{Y}_{\Delta}\ x[\rho] \in \CSet \] is representable.

  This condition can be unfolded into one of the following equivalent definitions.
  \begin{enumerate}
    \item For every $\Gamma : \abs{\CC}$ and $x : \abs{X}_{\Gamma}$, we have an extended object $\Gamma \rhd x$, a projection map $\mathbf{p}_{x} : \CC(\Gamma \rhd x \to \Gamma)$ and a generic element $\mathbf{q}_{x} : \abs{Y}_{\Gamma \rhd x}\ x[\mathbf{p}_{x}]$, satisfying a universal property: for every $\Delta : \abs{\CC}$, $\rho : \CC(\Delta \to \Gamma)$ and $y : \abs{Y}_{\Delta}\ x[\rho]$, there is a unique map $\langle \rho, y \rangle : \CC(\Delta \to (\Gamma \rhd x))$ such that $\langle \rho, y \rangle \cdot \mathbf{p}_{x} = \rho$ and $\mathbf{q}_{x}[\langle \rho,y \rangle] = y$.

      The reader familiar with the notion of category with families will have recognized the combinators used in Dybjer's original definition of CwF~\cite{InternalTT}.
    \item For every $\Gamma : \abs{\CC}$ and $x : \abs{X}_{\Gamma}$, the category of elements $(\CC/\Gamma/Y_{x})$ has a terminal object $(\Gamma \rhd x, \mathbf{p}_{x}, \mathbf{q}_{x})$.
    \item The natural transformation $\pi_{1} : \Sigma\ X\ Y \to X$ is a representable natural transformation in $\CPsh\ \CC$, i.e. for every representable presheaf $\yo\ \Gamma$ of $\CC$ and natural transformation $x : \yo\ \Gamma \to X$, there is a pullback square
      \[ \begin{tikzcd}
          \yo\ (\Gamma \rhd x) \ar[r] \ar[d] \ar[rd, phantom, very near start, "\lrcorner"] & \Sigma\ X\ Y \ar[d, "\pi_{1}"] \\
          \yo\ \Gamma \ar[r, "x"] & X
        \end{tikzcd} \]
      where the pullback $\yo\ (\Gamma \rhd x)$ is representable.
      See also Awodey's definition of natural model of type theory~\cite{NaturalModels}.
    \item We work in the internal language of $\CPsh\ \CC$ and write $\Ob_{\CC}$ for the constant presheaf of objects of $\Ob_{\CC}$ and $\yo : \Ob_{\CC} \to \SPsh_{\CC}$ for the internalization of the Yoneda embedding. We have global elements $X : \SPsh_{\CC}$ and $Y : X \to \SPsh_{\CC}$.

      For every global object $\Gamma : \Ob_{\CC}$ and global element $x : \yo\ \Gamma \to X$, the following type is inhabited
      \[ ((\Gamma \rhd x) : \Ob_{\CC}) \times (\yo\ (\Gamma \rhd x) \simeq ((\gamma : \yo\ \Gamma) \times Y\ (x\ \gamma))). \]
      The notion of \emph{global} element could be internalized using the flat modality of crisp type theory, as done in \cite{LicataOPS18}.
  \end{enumerate}
  Local representability is a structure on dependent presheaves, although it is categorically irrelevant.
  There are universes $\SRepPsh_{\CC,i}$ of $i$-small locally representable presheaf families, defined analogously to the presheaf universes.
  The universes $(\SRepPsh_{\CC,i})$ of representable presheaf families are closed under dependent pairs.
  \defiEnd
\end{defi}
From now on, we will say that a dependent presheaf is representable to mean that it is locally representable.
Because we usually consider dependent presheaves, for which there is no non-local notion of representability, this should be unambiguous.
Also note that in presence of finite products, local representability is equivalent to representability for non-dependent presheaves.

\subsection{Categories with (representable) families}

We use \defemph{categories with families} (CwFs) as our models of type theory.

\begin{defi}
  Internally to a presheaf model $\CPsh\ \CC$, an \defemph{internal representable family} is a pair $(\Ty,\Tm)$ with $\Ty : \SPsh_{\CC}$ and $\Tm : \Ty \to \SRepPsh_{\CC}$.

  A \defemph{category with families} (CwF) is a category $\CC$ equipped with a distinguished terminal object (written $\diamond$) and a global representable family $(\Ty_{\CC},\Tm_{\CC})$, consisting of a presheaf $\Ty_{\CC}$ and a locally representable dependent presheaf $\Tm_{\CC}$ over $\Ty_{\CC}$.
  \defiEnd
\end{defi}

One can check that unfolding this definition gives a notion of CwF that is isomorphic (or at least equivalent, depending on the precise definitions of the universes $\SPsh_{\CC}$ and $\SRepPsh_{\CC}$) to the standard definition.

A presentation of a type theory consists of algebraic operations and equations over CwFs.
This definition can be made formal using a notion of type theory signature extending the notion of QIIT-signature from \cite{QIITs}.
\begin{defi}\label{def:th_sig}
  The type theory of type theory signatures is defined to include the following structures.
  \begin{enumerate}
    \item A universe $\UU$ of sorts.
      \begin{mathpar}
        \inferrule{ }{\UU\ \type}

        \inferrule{A : \UU}{\El\ A\ \type}
      \end{mathpar}
    \item A subuniverse $\overline{\UU}$ of representable sorts (sorts for which context extension is permitted).
      \begin{mathpar}
        \inferrule{ }{\overline{\UU}\ \type}

        \inferrule{A : \overline{\UU}}{A : \UU}
      \end{mathpar}
    \item Dependent function types with arities in $\UU$.
      \begin{mathpar}
        \inferrule{A : \UU \\ [a : \El\ A]\ B(a)\ \type}{\Pi\ A\ B\ \type}

        \inferrule{[a : \El\ A]\ b(a) : B(a)}{\lam\ b : \Pi\ A\ B}

        \inferrule{f : \Pi\ A\ B \\ a : \El\ A}{\app\ f\ a : B(a)}
      \end{mathpar}
      They are used to specify the arguments of the operations and equations in a signature.
      The facts that these $\Pi$-types take arities in $\UU$ is similar to the strict positivity restriction of inductive types.
    \item Dependent function types in $\UU$ with arities in $\overline{\UU}$.
      \begin{mathpar}
        \inferrule{A : \overline{\UU} \\ [a : \El\ A]\ B(a) : \UU}{\overline{\Pi}\ A\ B : \UU}

        \inferrule{[a : \El\ A]\ b(a) : \El\ B(a)}{\overline{\lam}\ b : \El\ (\overline{\Pi}\ A\ B)}

        \inferrule{f : \El\ (\overline{\Pi}\ A\ B) \\ a : \El\ A}{\overline{\app}\ f\ a : \El\ B(a)}
      \end{mathpar}
      They are used to encode the fact that the arguments of the operations and equations in a signature can live in extended contexts.
    \item Extensional equality types.
      \begin{mathpar}
        \inferrule{A\ \type \\ x,y : A}{\Eq_{A}\ x\ y\ \type}

        \inferrule{x : A}{\refl\ x : \Eq_{A}\ x\ x}

        \inferrule{p : \Eq_{A}\ x\ y}{x = y}

        \inferrule{p : \Eq_{A}\ x\ x}{p = \refl\ x}
      \end{mathpar}
    \item A unit type and dependent pair types.
      \begin{mathpar}
        \inferrule{ }{\Unit\ \type}

        \inferrule{A\ \type \\ [a : A]\ B(a)\ \type}{\Sigma\ A\ B\ \type}
      \end{mathpar}
  \end{enumerate}
  A \defemph{type theory signature} is a closed type in the syntax of the theory of type theory signatures.
  An \defemph{extension} of a type theory signature $\Th$ is a dependent type $\Th'$ over $\Th$.
  It gives rise to an extended signature $(x : \Th) \times \Th'\ x$.
  \defiEnd
\end{defi}

Any type theory signature can be interpreted in any presheaf model $\CPsh_{\CC}$; the universe $\UU$ of sorts is interpreted as the presheaf universe $\SPsh_{\CC}$, the universe $\overline{\UU}$ of representable sorts is interpreted as the universe $\SRepPsh_{\CC}$ of locally representable presheaf families, and the other components are interpreted by the standard $\Pi$-types, $\Sigma$-types, and equality types of the presheaf model.
For example, the signature of CwFs is \[ \Th_{\mathsf{CwF}} \triangleq (\Ty : \UU) \times (\Tm : \Ty \to \overline{\UU}). \]
The interpretation of $\Th_{\mathsf{CwF}}$ in a presheaf model gives exactly the presheaf of internal representable families.
The sort $\Tm$ of terms is representable, while the sort $\Ty$ of types is not.
Thus type-theoretic structures extending the signature of CwFs are allowed to contain operations with higher-order arguments (i.e. binders), but the higher-order arguments can themselves only depend on $\Tm$.

Other type-theoretic structures can be described by extensions of the signature of CwFs.
For example, the structure of $\Pi$-types with strict $\beta$ and $\eta$ rules consists of operations $\Pi$, $\lam$ and $\app$ and equations $\app_{\beta}$ and $\Pi_{\eta}$, specified by the following signature extending $\Th_{\mathsf{CwF}}$.
\begin{alignat*}{2}
  & \Pi && : (A : \Ty) (B : \overline{\Tm\ A} \to \Ty) \to \Ty \\
  & \lam && : \{A,B\} (b : \overline{(a : \Tm\ A)} \to \Tm\ (B\ a)) \to \Tm\ (\Pi\ A\ B) \\
  & \app && : \{A,B\} (f : \Tm\ (\Pi\ A\ B)) (a : \Tm\ A) \to \Tm\ (B\ a) \\
  & \app_{\beta} && : \{A,B,b,a\} \to \app\ (\lam\ b)\ a = (b\ a) \\
  & \Pi_{\eta} && : \{A,B,f\} \to \lam\ (\app\ f) = f
\end{alignat*}

Again, one can check that unfolding the interpretation of this signature in presheaf categories gives a definition that is equivalent to the usual external definition of $\Pi$-type structures over CwFs.

Note that this definition of type-theoretic structure allows us to work with the syntax of type theories in the internal language of presheaf categories using higher-order abstract syntax (HOAS).
Indeed, presheaf models have been used to justify HOAS~\cite{SyntaxAndSemantics,HofmannHOAS}.

We included the definition of type theory signature for completeness, and to justify the quantifications on all type-theories appearing in this paper.
However, being fully formal with it would require us to develop its theory further, which we believe to be outside of the scope of this paper.
Thus we will only use this notion informally.
It may seem to invalidate our claim that our theorems are valid for arbitrary type theories.
However, it will be quite clear that all of the constructions that we perform in this paper and that depend on the actual signature are uniform in the type-theoretic operations of the signature.
For instance, when defining the interpretation of the $\Pi$ type former in a model, the interpretation will only depend on the shape of the type former $\Pi : (A : \Ty) (B : \overline{\Tm\ A} \to \Ty) \to \Ty$, but never on the presence of any other operation in the signature, and would work just as well for any other type former $X : (A : \Ty) (B : \overline{\Tm\ A} \to \Ty) \to \Ty$.

All of the arguments presented in this paper can alternatively be checked independently for any concrete type theory.

We however refer the reader to two similar general definitions of type theories.
Capriotti's rule framework~\cite{CappriottiRF,CapriottiThesis} is similar to our definition, without the strict positivity restriction on $\Pi$-types, and without the $\Pi$-types with arities in $\overline{\UU}$.
This implies that the type formers can include unrestricted higher-order arguments, which may fail to have well-defined categories of models or initial models.
Uemura's representable map categories~\cite{GeneralFrameworkTT} can encode almost the same type theories as our definition.
The main difference is that our definition generalizes QIIT-signatures and generalized algebraic theories (i.e. algebraic theories with dependent sorts), whereas Uemura's definition generalizes essentially algebraic theories (i.e. algebraic theories with partial operations).
This does not change the class of presentable type theories, but the additional structure of generalized algebraic theories is crucial for the present paper.
Another minor difference is that we only consider finite signatures.
The semantics of representable map categories are given by functorial semantics, whereas the semantics for our notion of signature is more directly defined by induction on the signatures, as in~\cite{QIITs}.

We will only consider type theories that extend the signature of CwFs (and cumulative CwFs, which will be introduced later) by new operations and equations only, that is type theories whose only sorts are the non-representable sort of types and the representable sort of terms.

From the generalized algebraic presentations of CwFs and type-theoretic structures, we obtain a $1$-category $\CMod_{\Th}$ of models from any signature $\Th$.
The objects of $\CMod_{\Th}$ are CwFs equipped with the additional type-theoretic structures of $\Th$.
The morphisms are functors, with additional actions on types and terms, strictly preserving the chosen terminal object and the representing objects for the context extensions and the type-theoretic operations of $\Th$.
Given a morphism $F : \CMod_{\Th}(\CC \to \CD)$, we will denote its actions on types and terms by $F : \{\Gamma : \CC^{\op}\} \to \abs{\Ty_{\CC}}_{\Gamma} \to \abs{\Ty_{\CD}}_{F\ \Gamma}$ and $F : \{\Gamma : \CC^{\op}\}\{A : \abs{\Ty_{\CC}}_{\Gamma}\} \to \abs{\Tm_{\CC}}_{\Gamma}\ A \to \abs{\Tm_{\CD}}_{F\ \Gamma}\ (F\ A)$.

We also automatically obtain the existence of an initial object $\Init_{\Th}$ of $\CMod_{\Th}$.
We adopt the algebraic point of view on the syntax of type theory: we only work with the abstract characterization of the syntax as the components of an initial object, and don't try to give any more explicit construction of this initial object.

More generally, $\CMod_{\Th}$ is a finitely locally presentable category, and is in particular complete and cocomplete.
We also automatically obtain that freely generated models exist.
We write $\mathsf{Free}(\dots)$ for freely generated models.
For example, $\mathsf{Free}(\bm{\Gamma} \vdash )$ is the model freely generated by a single object $\bm{\Gamma}$, $\mathsf{Free}(\bm{\Gamma} \vdash \bm{A} : \Ty)$ is the model freely generated by an object $\bm{\Gamma}$ and a type $\bm{A}$ over $\bm{\Gamma}$, and $\mathsf{Free}(\bm{\Gamma} \vdash \bm{a} : \Tm\ \bm{A})$ is the model freely generated by an object $\bm{A}$, a type $\bm{A}$ over $\bm{\Gamma}$ and a term $\bm{a}$ of type $\bm{A}$.
We use bold symbols ($\bm{\Gamma}$, $\bm{A}$, $\bm{a}$, etc) to distinguish the generators of a freely generated model.
These models satisfy some universal properties.
For instance, the morphisms $\mathsf{Free}(\bm{\Gamma} \vdash ) \to \CC$ are in natural bijection with the objects of $\CC$.

We denote the $1$-category of CwFs without any additional structure by $\CCwf$.

For some purposes, it may have been preferable or more elegant to work with $2$-categories of models, weak morphisms (i.e. morphisms that only preserve the terminal object and the context extensions up to isomorphism) and natural transformations.
We will however need to consider additional structures on the categories of models that are better developed in the $1$-categorical setting, such as (both orthogonal and weak) factorization systems and semi-model structures.

\subsection{Contextual models}\label{ssec:contextual_models}

An important class of CwFs are the \emph{contextual} CwFs, whose objects and morphisms are really given by lists of types and terms.
Indeed, from some point of view, in the language of type theory, we never explicitly talk about the objects and morphisms of a model, but only about types and terms that live in the same contextual slice of a given model.
Thus only the contextual models matter.
However, a direct definition of contextual models is complicated (their generalized algebraic presentation is infinite), and many intermediate constructions go through non-contextual models.
It is thus more convenient to define contextuality as a property of general models.
Fortunately, they can nicely be described by the means of an orthogonal factorization system on $\CCwf$\footnote{The author learnt of this definition of contextuality from Christian Sattler.}.

We first recall the definition of orthogonal factorization systems, originally introduced in \cite{FreydKellyOFS}.
\begin{defi}
  An \defemph{orthogonal factorization system} on a category $\CC$ consists of two classes of maps $\CL$ and $\CR$ satisfying the following two properties:
  \begin{itemize}
    \item Every map $f : \CC(X \to Y)$ can be factored as $f = l \cdot r$, where $l \in \CL$ and $r \in \CR$.
    \item Every map in $\CL$ is left orthgonal to every map in $\CR$; this means that for every $l \in \CL$ and $r \in \CR$ and commutative square
      \[ \begin{tikzcd}
          A \ar[d, "l"] \ar[r, "f"] & X \ar[d, "r"] \\
          B \ar[r, "g"] & Y \rlap{\, }
        \end{tikzcd} \]
      there exist an unique map $j : \CC(B \to X)$ such that $j \cdot r = g$ and $l \cdot j = f$.
      \defiEnd
  \end{itemize}
\end{defi}

\begin{defi}
  A morphism $F : \CCwf(\CC \to \CD)$ is said to be a \defemph{contextual isomorphism} if its actions on types and terms are bijective.
  \defiEnd
\end{defi}

Let $I$ be the set of maps of $\CCwf$ consisting of $I^{\Ty} : \mathsf{Free}(\bm{\Gamma} \vdash) \to \mathsf{Free}(\bm{\Gamma} \vdash \bm{A} : \Ty)$ and $I^{\Tm} : \mathsf{Free}(\bm{\Gamma} \vdash \bm{A} : \Ty) \to \mathsf{Free}(\bm{\Gamma} \vdash \bm{a} : \Tm\ \bm{A})$.
Contextual isomorphisms are exactly the maps that are right orthogonal to $I$.
The maps that are left orthogonal to the contextual isomorphisms are called \defemph{contextual extensions}.

By the small object argument for orthogonal factorization systems~\cite{Kelly1980}, contextual extensions and contextual isomorphisms form an orthogonal factorization system.
Any morphism $F : \CCwf(\CC \to \CD)$ admits an unique (up to isomorphism) factorization $\CC \to \cxlim F \to \CD$ where $\CC \to \cxlim F$ is a contextual extension and $\cxlim F \to \CD$ is a contextual isomorphism.
The CwF $\cxlim F$ is called the \defemph{contextual image} of $F$.

In particular, given any $\CC : \CCwf$, the unique morphism $\Init \to \CC$ admits such a factorization.
Its contextual image is called the \defemph{contextual core} of $\CC$, and is denoted by $\cxl \CC$.
The map $\cxl \CC \to \CC$ is a contextual isomorphism by definition.
When the map $\cxl \CC \to \CC$ is also an isomorphism of CwFs, we say that $\CC$ is \defemph{contextual}.

This definition of contextuality is equivalent to the usual definition, as found for instance in \cite{CwFUSD}.
\begin{prop}
  A CwF $\CC$ is contextual if and only there exists a length function $l : \abs{\CC} \to \Nat$ such than for any $\Gamma : \abs{\CC}$, if $(l\ \Gamma) = 0$ then $\Gamma = \diamond$ and if $(l\ \Gamma) = n+1$, then there are unique $\Gamma' : \abs{\CC}$ and $A : \abs{\Ty}_{\Gamma'}$ such that $\Gamma = \Gamma' \rhd A$.
  \qed
\end{prop}

All type-theoretic structures can be transported along contextual isomorphisms.
Thus, given a morphism $F : \CMod_{\Th}(\CC \to \CD)$ of models of some theory $\Th$, the contextual image $\cxlim F$ carries a canonical structure of model of $\Th$, and the factors $\CC \to \cxlim F$ and $\cxlim F \to \CD$ are both morphisms of models of $\Th$.

If $\Th$ is a type theory signature, the category of contextual models of $\Th$ is written $\CMod_{\Th}^{\cxl}$.
We have an adjunction
\[ \begin{tikzcd}
    \CMod_{\Th}^{\cxl} \ar[r, bend right] \ar[r, phantom, "\bot"] & \CMod_{\Th} \rlap{\ .} \ar[l, bend right]
  \end{tikzcd} \]
The right adjoint $\CMod_{\Th}^{\cxl} \to \CMod_{\Th}$ is just the functor forgetting that a model is contextual.
The left adjoint $\CMod_{\Th} \to \CMod_{\Th}^{\cxl}$ takes the contextual core of a model.

\begin{prop}
  The contextual core of a model $\CC : \CMod_{\Th}$ is the initial model of $\Th$ equipped with a contextual isomorphism into $\CC$.
\end{prop}
\begin{proof}
  Given any other model $\CD$ equipped with a contextual isomorphism $F : \CD \to \CC$, we have, since $\Init_{\Th} \to \cxl \CC$ is a contextual extension, a unique lift in the following diagram.
  \[ \begin{tikzcd}
      \Init_{\Th} \ar[d] \ar[r] & \CD \ar[d, "F"] \\
      \cxl \CC \ar[ru, dashed] \ar[r] & \CC
    \end{tikzcd} \]
\end{proof}

\begin{prop}\label{prop:contextual_section}
  To check that a CwF $\CC$ is contextual, it suffices to check that the morphism $\cxl \CC \to \CC$ admits a section.
\end{prop}
\begin{proof}
  Assume that $r : \cxl \CC \to \CC$ admits a section $s : \CC \to \cxl \CC$.
  Then the following diagram commutes.
  \[
    \begin{tikzcd}
      \cxl \CC \ar[rd, "\id", shift left] \ar[rd, "r \cdot s"', shift right] \ar[rr, "r"] && \CC \\
      & \cxl \CC \ar[ru, "r"']&
    \end{tikzcd} \]
  Since $\cxl \CC$ is initial among the CwFs with a contextual isomorphism into $\CC$, we have that $r \cdot s = \id$, and $r : \cxl \CC \to \CC$ is therefore an isomorphism, as needed.
\end{proof}

\begin{prop}
  For any type theory signature $\Th$, the initial model $\Init_{\Th} : \CMod_{\Th}$ is contextual.
\end{prop}
\begin{proof}
  By initiality of $\Init_{\Th}$, the morphism $\cxl \Init_{\Th} \to \Init_{\Th}$ admits a section, which implies that $\Init_{\Th}$ is contextual by \cref{prop:contextual_section}.
\end{proof}

\begin{defi}
  Given a model $\CC$ of a theory $\Th$ and $\Gamma : \abs{\CC}$, the \defemph{contextual slice} $(\CC \sslash \Gamma)$ is defined to be the contextual core of the slice model $(\CC / \Gamma)$.
  \defiEnd
\end{defi}

\subsection{Join of families and telescopes}

We will not assume the presence of $\Sigma$-types in our type theories.
To circumvent their absence in some constructions, we will need to work with families of telescopes, whose types and terms are finite sequences of types and terms of the base family.
It is convenient to present them as the coproduct of length $n$ telescopes for all $n : \Nat$, and to generalize the notion of length $n$ telescope to a more heterogeneous notion, using the notion of join of families\footnote{The author learnt of this presentation from Christian Sattler.}.

We work internally to some presheaf category $\CPsh\ \CC$.

\begin{defi}
  Let $\MC = (\Ty_{\MC},\Tm_{\MC})$ and $\MD = (\Ty_{\MD},\Tm_{\MD})$ be two internal families (not necessarily representable).
  Their \defemph{join} $\MC \ast \MD$ is the internal family defined by:
  \begin{alignat*}{2}
    &\Ty_{\MC \ast \MD} && \triangleq (A : \Ty_{\MC}) \times (B : \Tm_{\MC}\ A \to \Ty_{\MD}) \\
    & \Tm_{\MC \ast \MD}\ (A, B) && \triangleq (a : \Tm_{\MC}\ A) \times (b : \Tm_{\MD}\ (B\ a))
  \end{alignat*}
  Whenever both $\MC$ and $\MD$ are representable, the family $\MC \ast \MD$ is also representable (since locally representable presheaves are closed under $\Sigma$-types).
  \defiEnd
\end{defi}

In other words, the join of $\MC$ and $\MD$ is the family of length $2$ telescopes, whose first and second components come respectively from $\MC$ and $\MD$.

\begin{defi}
  If $\MC$ is an internal family and $n : \Nat$, the family $\MC^{\ast n}$ of \defemph{length $n$ telescopes} is the $n$-fold iterated join of $\MC$.
  \defiEnd
\end{defi}

\begin{defi}
  If $\MC$ is an internal family, the family of \defemph{telescopes} of $\MC$ is the coproduct $\MC^{\star} \triangleq \coprod_{n : \Nat}\ \MC^{\ast n}$.
  We write $\Ty_{\MC}^{\star}$ and $\Tm_{\MC}^{\star}$ for the components of $\MC^{\star}$.
  If $\MC$ is representable, then $\MC^{\star}$ is also representable.
  \defiEnd
\end{defi}

If $\CC$ is a contextual CwF, we may identify its objects with the closed telescopes of types (i.e. the global elements of $\Ty^{\star}_{\CC}$) and its morphisms from $\Gamma$ to $\Delta$ with the natural transformations from $\Tm^{\star}_{\CC}\ \Gamma$ to $\Tm^{\star}_{\CC}\ \Delta$. When $\CC$ is an arbitrary CwF, this is an explicit description of the objects and morphisms of the contextual core of $\CC$.

\subsection{Cumulative categories with families}

We will actually work with type theories that extend the theory of cumulative categories with families, rather than the simpler theory of categories with families.
Cumulative categories with families were introduced by Coquand~\cite{CoquandNbE} to describe universe hierarchies.
Working with cumulative CwFs, instead of mere CwFs, ensures that every type admits a code in some universe.
In presence of identity types, this provides in turn a way to compare types up to internal equality of codes.

This is mainly for convenience: most of our results could also be formulated and proven for mere CwFs, comparing types up to equivalence.
However, using cumulative CwFs simplifies the proofs and the presentation.

\begin{defi}
  Internally to a presheaf model $\CPsh_{\CC}$, an internal cumulative family consists of a family
  \[\Ty : \Nat \to \SPsh_{\CC}\]
  of presheaves of types ($\Ty_{n}$ is the presheaf of types in the $n$-th universe of the hierarchy), a family
  \[\Tm : \{n : \Nat\} \to \Ty_{n} \to \SRepPsh_{\CC} \]
  of locally representable presheaves of terms, lifting functions
  \[ \mathsf{Lift}_{\Ty} : \Ty_{n} \to \Ty_{n+1}\]
  and isomorphisms
  \[ \mathsf{lift}_{\Tm} : \Tm_{n}\ A \simeq \Tm_{n+1}\ (\mathsf{Lift}_{\Ty}\ A) : \mathsf{lower}_{\Tm}. \]
  Given an internal cumulative family $\MC = (\Ty, \Tm)$, we write $\MC^{n}$ for the internal family $(\Ty_{n}, \Tm_{n})$.
 
  A cumulative CwF, or cCwF, is a category $\CC$ equipped with a global cumulative family $(\Ty, \Tm)$.
  \defiEnd
\end{defi}

\begin{defi}
  The structure of cumulative universes over a cumulative internal family $(\Ty,\Tm)$ consists of operations \[ \UU : (n : \Nat) \to \Ty_{n+1} \] along with isomorphisms \[ \Tm_{n+1}\ \UU_{n} \simeq \Ty_{n}, \]
  that we will leave implicit.
  \defiEnd
\end{defi}

This differs slightly from Coquand's definition of cumulative CwFs.
Coquand requires the natural transformations $\mathsf{Lift}_{\Ty} : \Ty_{n} \to \Ty_{n+1}$ to be injective, and uses equalities $\Tm\ \UU_{n} = \Ty_{n}$ instead of isomorphisms.

When talking about type theory signatures in this paper, we mean signatures over the theory of cumulative CwFs with universes, i.e. extensions of the signature of cumulative CwFs with universes by new operations and equations only (no new sorts).

The notion of contextuality generalizes straightforwardly to cumulative CwFs.
For instance, a contextual isomorphism between cumulative CwFs is a morphism that is bijective on types and terms for each universe level.

The notion of telescope can also be adapted to cumulative families.
Given an internal cumulative family $\MC$ and a list $w = (w_{1}, \cdots, w_{n})$ of natural numbers, the family of $w$-shaped telescopes is the join $\MC^{\ast w} \triangleq \MC^{w_{1}} \ast \dots \ast \MC^{w_{n}}$.
The family of all telescopes of $\MC$ is the coproduct $\MC^{\star} \triangleq \coprod_{w : \Nat^{\star}}\ \MC^{\ast w}$

For many of the properties of cumulative CwFs with universes and morphisms of cCwFs that are defined by conditions on types and terms, the condition on types is a consequence of the condition on terms of the corresponding universe.
For instance, to check that a morphism is a contextual isomorphism, it is sufficient to check that it is bijective on terms for each universe level.

To improve the readability, we will leave the universe levels implicit in most constructions and proofs.

\section{Weak and strong type structures}\label{sec:weak_id_types}

In this section we define the weak and strong variants of the basic type-theoretic structures: identity types and $\Pi$-types.
We work in the internal language of $\CPsh\ \CC$ for a fixed category $\CC$.

Generally, the computation rules of a weak type structure are expressed by internal equalities, whereas the computation rules of strong type structures are expressed by strict equalities.
We prefer to use the adjective \emph{strong} instead of \emph{strict} to qualify type structures and type theories with strict computation rules, in order to avoid ambiguity when talking about strict identity types.
Strong identity types will refer to identity types with a strict $\beta$-rule, whereas strict identity types will refer to identity types satisfying the UIP principle.

\subsection{Weak identity types}

In presence of strong $\Sigma$-types and strong $\Pi$-types, there are several equivalent ways to define the eliminator for identity types.
The Martin-Löf eliminator is given by the following rule.
\begin{mathpar}
  \inferrule[Martin-Löf eliminator]{A\ \type \\ [x : A, y : A, p : \Id\ x\ y]\ P(x,y,p)\ \type \\ [x : A]\ d(x) : P(x,x,\refl) \\ x : A \\ y : A \\ p : \Id\ x\ y}{\J\ P\ d\ x\ y\ p : P(x,y,p)}
\end{mathpar}
In absence of strong $\Pi$-types, it is known that the Martin-Löf eliminator is not strong enough to even define transport\footnote{For a countermodel, take the CwF with $\Id$ and $\refl$ freely generated by a type $A$, a type family $B$ over $A$, and terms $a_{1}, a_{2} : A$, $b : B\ a_{1}$ and $p : \Id\ \{A\}\ a_{1}\ a_{2}$.
  The only terms of that model are the variables, the weakenings of the generators, and their iterated reflexivity paths.
  Since there is no closed term of type $B\ a_{2}$, that model does not satisfy transport.
  But it can still be equipped with the Martin-Löf eliminator.}.
In \cite{IdTypeFS}, a variant of the Martin-Löf eliminator, now called the Frobenius variant of the Martin-Löf eliminator, is introduced.
The idea is to circumvent the absence of $\Pi$-types by allowing the target type family $P$ of the elimination to depend on any telescope $\Delta$ of parameters.
\begin{mathpar}
  \inferrule[Frobenius Martin-Löf eliminator]{A\ \type \\ [x : A, y : A, p : \Id\ x\ y]\ \Delta(x,y,p)\ \type^\star \\\\
    [x : A, y : A, p : \Id\ x\ y, \delta : \Delta(x,y,p)]\ P(x,y,p,\delta)\ \type \\
    [x : A, \delta : \Delta(x,x,\refl)]\ d(x,\delta) : P(x,x,\refl,\delta) \\\\
    x : A \\ y : A \\ p : \Id\ x\ y \\ \delta : \Delta(x,y,p)}{\J\ P\ d\ x\ y\ p\ \delta : P(x,y,p,\delta)}
\end{mathpar}
Another alternative is the Paulin-Mohring eliminator, also known as based path induction, or one-sided eliminator.
\begin{mathpar}
  \inferrule[Paulin-Mohring eliminator]{A\ \type \\ x : A \\ [y : A, p : \Id\ x\ y]\ P(y,p)\ \type \\\\ d(x) : P(x,\refl) \\ y : A \\ p : \Id\ x\ y}{\J\ x\ P\ d\ y\ p : P(y,p)}
\end{mathpar}
North~\cite{IdWFSCauchy} and Isaev~\cite{IsaevIxTT} have independently given proofs of the fact that the Paulin-Mohring eliminator is equivalent to the Frobenius variant of the Martin-Löf eliminator in the presence of strong $\Sigma$-types.

We use the weak variant of the Paulin-Mohring eliminator, with the computation rule weakened to a weak equality.
We will show that the other eliminators can also be derived, even in the absence of strong $\Sigma$-types.
In fact, we will prove that weak identity types can be lifted from a family $\MC$ to its telescope family $\MC^{\star}$, i.e. the family whose types are list of types of $\MC$; the derivation of the Frobenius eliminator can be seen as a consequence of this fact.
The main step of this derivation is originally due to András Kovács.
The proof has been simplified using the notion of join of internal families by Christian Sattler.

\begin{defi}\label{def:id_intro}
  Let $\MC = (\Ty,\Tm)$ be an internal family (not necessarily representable).
  An \defemph{identity type introduction structure} over $\MC$ is specified by the following signature:
  \begin{alignat*}{2}
    & \Id && : \{A : \Ty\} (x : \Tm\ A) (y : \Tm\ A) \to \Ty \\
    & \refl && : \{A : \Ty\} (x : \Tm\ A) \to \Tm\ (\Id\ x\ x) \tag*{\defiEnd}
  \end{alignat*}
\end{defi}

\begin{defi}\label{def:weak_id_elim}
  Let $\MC = (\Ty_{\MC}, \Tm_{\MC})$ and $\MD = (\Ty_{\MD}, \Tm_{\MD})$ be internal families with identity type introduction structures.
  A \defemph{weak identity type elimination structure} from $\MC$ to $\MD$ is specified by the following signature:
 \begin{alignat*}{3}
   & \J && :{ } && \{A : \Ty_{\MC}\} \{x : \Tm_{\MC}\ A\} (P : (y : \Tm_{\MC}\ A)(p : \Tm_{\MC}\ (\Id\ x\ y)) \to \Ty_{\MD}) \\
   &&&&& (d : \Tm_{\MD}\ (P\ x\ \refl)) \{y : \Tm_{\MC}\ A\} (p : \Tm_{\MC}\ (\Id\ x\ y)) \\
   &&&&& \to \Tm_{\MD}\ (P\ y\ p) \\
   & \J_{\beta} && :{ } && \{A : \Ty_{\MC}\} \{x : \Tm_{\MC}\ A\} (P : (y : \Tm_{\MC}\ A)(p : \Tm_{\MC}\ (\Id\ x\ y)) \to \Ty_{\MD}) \\
   &&&&& (d : \Tm_{\MD}\ (P\ x\ \refl)) \\
   &&&&& \to \Tm_{\MD}\ (\Id\ (\J\ x\ P\ d\ x\ \refl)\ d) \tag*{\defiEnd}
   \end{alignat*}
 \end{defi}

 A \defemph{weak identity type structure} over an family $\MC$ consists of an identity type introduction structure over $\MC$ along with a weak identity type elimination structure from $\MC$ to $\MC$.


\begin{defi}
  An internal family $\MC = (\Ty,\Tm)$ with weak identity types is said to have \defemph{representable singletons} if for every $A : \Ty$ and $x : \Tm\ A$, the dependent presheaf of singletons $\Singl\ x \triangleq (y : \Tm\ A) \times \Tm\ (\Id\ x\ y)$ is representable.
  \defiEnd
\end{defi}
The higher-order parameters occurring in the signature for elimination structures are isomorphic to $\Singl\ x$ for some $x$, and therefore assuming that singletons are representable is sufficient to give a signature for weak identity types.

Note that if $\MC$ is a representable family, then it automatically has representable singletons, since representable presheaf families are closed under dependent sums.
But we will also consider type theories with representable singletons but a non-representable family.
Syntactically, this means that the context extension is restricted to the contractible context extension $\Gamma,(y:A,p:\Id\ x\ y)$.
This restricted context extension is inspired by Brunerie's type-theoretic definition of weak $\omega$-groupoid \cite[Appendix A]{BrunerieThesis}.

\begin{defi}\label{def:strong_id}
  A weak identity type structure is said to be \defemph{strong}, or to have a \defemph{strict $\beta$-rule}, if the following equations hold, for all relevant arguments:
  \begin{mathpar}
    \J\ P\ d\ x\ \refl = d

    \J_{\beta}\ P\ d = \refl
  \end{mathpar}
\end{defi}

We say that an internal cumulative family $\MC$ has weak identity types if each family $\MC^{n}$ has an identity type introduction structure, together weak elimination structures from $\MC^{n}$ to $\MC^{m}$ for all $n, m : \Nat$.

Given an internal family $\MC = (\Ty,\Tm)$ equipped with weak identity types, we can derive the transport operation:
\[ \mathsf{transport} : \{A, x, y\} (P : \Tm\ A \to \Ty) (p : \Tm\ (\Id\ x\ y)) \to \Tm\ (P\ x) \to \Tm\ (P\ y). \]
We will often write $p^{\star}\ d$ instead of $\mathsf{transport}\ P\ p\ d$, leaving the family $P$ implicit.

We will write $(p \cdot q)$ for the composition of two internal equalities $p : \Id\ x\ y$ and $q : \Id\ y\ z$, and $p^{-1}$ for the inverse of an internal equality $p : \Id\ x\ y$.

The standard notions of homotopy type theory, such as contractible types, propositional types, equivalences, etc, can be defined.
However, since we may not have $\Sigma$-types or $\Pi$-types, they are not encoded by types of the theory.

\subsection{Weak \texorpdfstring{$\Pi$}{Pi}-types}

We now define the weak variant of $\Pi$-types.
We defined weak identity types for families that are only required to have representable singletons.
Similarly, it will be useful to have a definition of $\Pi$-types that is as general as possible with respect to the representability of the families.
In our case, we will consider $\Pi$-types in a family $\MD$ with domains, or arities, in another family $\MC$ (and codomains in $\MD$).
The family $\MC$ is required to be representable, whereas the family $\MD$ is only required to have weak identity types with representable singletons.

We only consider $\Pi$-types with function extensionality.
We use one of the definitions of function extensionality from \cite{StrPi}.

\begin{defi}\label{def:weak_pi}
  Let $\MC$ be a representable internal family and $\MD$ be an family equipped with weak identity types.
  An introduction structure for $\Pi$-types in $\MD$ with arities in $\MC$ is presented by:
  \begin{alignat*}{2}
    & \Pi && : (A : \Ty_{\MC})\ (B : \Tm_{\MC}\ A \to \Ty_{\MD}) \to \Ty_{\MD} \\
    & \lam && : \{A,B\}\ (b : (a : \Tm_{\MC}\ A) \to \Tm_{\MD}\ (B\ a)) \to \Tm_{\MD}\ (\Pi\ A\ B)
  \end{alignat*}
  An application structure consists of:
  \begin{alignat*}{3}
    & \app && :{ } && \{A,B\}\ (f : \Tm_{\MD}\ (\Pi\ A\ B))\ (a : \Tm_{\MC}\ A) \to \Tm_{\MD}\ (B\ a) \\
    & \app_{\beta} && :{ } && \{A,B\} (b : (a : \Tm_{\MC}\ A) \to \Tm_{\MD}\ (B\ a)) (a : \Tm_{\MC}\ A) \to \\
    &&&&& \Tm_{\MD}\ (\Id\ (\app\ (\lam\ b)\ a)\ (b\ a))
  \end{alignat*}
  Given an application structure, we can derive:
   \begin{alignat*}{3}
     & \happly && :{} && \{A,B\}\ \{f,g : \Tm_{\MD}\ (\Pi\ A\ B)\} (p : \Tm_{\MD}\ (\Id\ f\ g)) (a : \Tm_{\MC}\ A) \to \\
     &&&&& \Tm_{\MD}\ (\Id\ (\app\ f\ a)\ (\app\ g\ a)).
  \end{alignat*}
  An extensionality structure consists of:
  \begin{alignat*}{3}
    & \funext && :{} && \{A,B\} \{f,g : \Tm_{\MD}\ (\Pi\ A\ B)\} (h : \Tm_{\MD}\ (\Pi\ A\ (a \mapsto \Id\ (\app\ f\ a)\ (\app\ g\ a)))) \to \\
    &&&&& \Tm_{\MD}\ (\Id\ f\ g) \\
    & \funext_{\beta} && :{ } && \{A,B\} \{f : \Tm_{\MD}\ (\Pi\ A\ B)\} \to \\
    &&&&& \Tm_{\MD}\ (\Id\ (\funext\ (\lam\ (a \mapsto \refl)))\ \refl) \\
    & \funext\text{-}\app && :{ } && \{A,B,f,g\} (h : \Tm_{\MD}\ (\Pi\ A\ (a \mapsto \Id\ (\app\ f\ a)\ (\app\ g\ a)))) (a : \Tm_{\MC}\ A) \to \\
    &&&&& \Tm_{\MD}\ (\Id\ (\happly\ (\funext\ h)\ a)\ (\app\ h\ a))  \\
    & \funext\text{-}\app_{\beta} && :{ } && \{A,B,f\} (a : \Tm_{\MC}\ A) \to \\
    &&&&& \Tm_{\MD}\ (\Id\ (\funext\text{-}\app\ (\lam\ (a \mapsto \refl)))\ p),
  \end{alignat*}
  where $p : \Tm_{\MD}\ (\Id\ (\lam\ (x \mapsto \happly\ (\funext\ (\lam\ (a \mapsto \refl)))))\ (\lam\ (x \mapsto \refl)))$ can be derived from $\app_{\beta}$, $\happly_{\beta}$, $\funext$ and $\funext_{\beta}$.

  The structure of $\Pi$-types in $\MD$ with arities in $\MC$ consists of an introduction structure, an application structure and an extensionality structure.
  \defiEnd
\end{defi}

Remark that in the signature of $\funext$, the homotopy $h$ between $f$ and $g$ is encoded as an inhabitant of the $\Pi$-type $(\Pi\ A\ (a \mapsto \Id\ (\app\ f\ a)\ (\app\ g\ a)))$, rather than as a family $(a : \Tm\ A) \to \Tm\ (\Id\ (\app\ f\ a)\ (\app\ g\ a))$.
This is important in the absence of a strict $\beta$-rule for $\Pi$-types, as we would not be able to prove the congruence law for $\funext$ otherwise (the fact that whenever two homotopies $h, h'$ between $f$ and $g$ are themselves homotopic, then $\funext\ f\ g\ h$ and $\funext\ f\ g\ h'$ are internally equal).

\begin{defi}\label{def:strong_pi}
  If $\MD$ has weak identity types, a weak $\Pi$-type structure in $\MD$ with arities in $\MC$ is said to have a \defemph{strict $\beta$-rule} if it satisfies the equations:
  \begin{mathpar}
    \app\ (\lam\ b)\ a = b\ a

    \app_{\beta} = \refl
  \end{mathpar}
\end{defi}

We say that an internal cumulative family or a cumulative CwF has weak $\Pi$-types if its families at each universe level have weak $\Pi$-types.

\subsection{Lifting type structures to telescopes}\label{ssec:lift_tele}

In this subsection, we show that both weak and strong $\Id$ and $\Pi$-type structures on a family $\MC$ can be lifted to the telescope family $\MC^{\star}$.
Similar results have been proven and used before in the literature \cite{2DTT, HoInvCwA}.
We generalize them to weak type structures and non-representable families.
The constructions have been formalized in Agda.
We refer the reader to the formalization for the detailed constructions.

\begin{con}\label{con:lift_wid_intro}
  Let $\MC$ and $\MD$ be families with identity type introduction structures, along with an weak identity type elimination structure from $\MC$ to $\MD$.

  Then the family $\MC \ast \MD$ is equipped with the following identity type introduction structure:
  \begin{alignat*}{2}
    & \Id\ \{(A,B)\}\ (x_{c}, x_{d})\ (y_{c}, y_{d}) && \triangleq (\Id\ \{A\}\ x_{c}\ y_{c} , p_{c} \mapsto \Id\ \{B\ y_{c}\}\ (p_{c}^{\star}\ y_{c})\ y_{d}) \\
    & \refl\ \{(A,B)\}\ (x_{c}, x_{d}) && \triangleq (\refl\ \{x_{c}\}, p),
  \end{alignat*}
  where $p$ is some term of type $\Id\ (\refl^{\star}\ x_{c})\ x_{c}$, definable using $\J_{\beta}$.
  \defiEnd
\end{con}

\begin{con}\label{con:lift_wid_join}
  Let $\MC$, $\MD$ and $\ME$ be families with identity type introduction structures, along with identity type elimination structures from $\MC$ to $\MD$ and $\ME$, from $\MD$ to $\MD$ and $\ME$ and from $\ME$ to $\ME$.

  Then there exist weak identity type elimination structure from $(\MC \ast \MD)$ to $\ME$ and from $\MC$ to $(\MD \ast \ME)$.
  \defiEnd
\end{con}

\begin{con}
  Let $\MC$ be an internal family with a weak identity type structure.

  Then for every $n : \Nat$, the family $\MC^{\ast n}$ of length $n$ telescopes has a canonical identity type introduction structure, and for every $n, m : \Nat$ there is a weak identity type elimination structure from $\MC^{\ast n}$ to $\MC^{\ast m}$.
\end{con}
\begin{proof}
  By iterating \cref{con:lift_wid_intro} and \cref{con:lift_wid_join}.
\end{proof}

\begin{con}\label{con:pi_lift_r}
  Let $\MC$, $\MD$ and $\ME$ be internal families, such that $\ME$ has weak identity types and weak $\Pi$-types with arities in $\MC$ and $\MD$.

  Then $\ME$ has weak $\Pi$-types with arities in $\MC \ast \MD$.
  \defiEnd
\end{con}

\begin{con}\label{con:pi_lift_l}
  Let $\MC$, $\MD$ and $\ME$ be internal families, such that $\MD$ and $\ME$ have weak identity types and weak $\Pi$-types with arities in $\MC$, along with a weak identity type elimination structure from $\MD$ to $\ME$.
  Note that by \cref{con:lift_wid_join}, $\MD \ast \ME$ has weak identity types.

  Then $\MD \ast \ME$ has weak $\Pi$-types with arities in $\MC$.
  \defiEnd
\end{con}

\begin{con}
  If $\MD$ has weak $\Pi$-types with arities in $\MC$, then for any $n, m : \Nat$, the family $\MD^{\ast m}$ of length $m$ telescopes of $\MD$ has weak $\Pi$-types with arities in $\MC^{\ast n}$.
\end{con}
\begin{proof}
  By iterating \cref{con:pi_lift_r} and \cref{con:pi_lift_l}.
\end{proof}


\subsection{Parametrized elimination structures}\label{ssec:param_id_elim}

\begin{defi}
  Let $\MC$, $\MD$ and $\ME$ be internal families, together with identity type introduction structures over $\MC$ and $\ME$. \\
  A parametrized identity type elimination structure from $\MC$ to $\ME$ with parameters in $\MD$ consists of operations
  \begin{alignat*}{3}
   & \J && :{ } && \{A : \Ty_{\MC}\} \{x : \Tm_{\MC}\ A\} (Q : \Singl\ x \to \Ty_{\MD}) (P : (y : \Singl\ x) (q : \Tm_{\MD}\ (Q\ y)) \to \Ty_{\ME}) \\
   &&&&& (d : (q : \Tm_{\MD}\ (Q\ (x,\refl))) \to \Tm_{\ME}\ (P\ (x,\refl)\ q)) (y : \Singl\ x) (q : \Tm_{\MD}\ (Q\ y)) \\
   &&&&& \to \Tm_{\ME}\ (P\ y\ q) \\
   & \J_{\beta} && :{ } && \{A : \Ty_{\MC}\} \{x : \Tm_{\MC}\ A\} (Q : \Singl\ x \to \Ty_{\MD}) (P : (y : \Singl\ x) (q : \Tm_{\MD}\ (Q\ y)) \to \Ty_{\ME}) \\
   &&&&& (d : (q : \Tm_{\MD}\ (Q\ (x,\refl))) \to \Tm_{\ME}\ (P\ (x,\refl)\ q)) (q : \Tm_{\MD}\ (Q\ (x,\refl))) \\
   &&&&& \to \Tm_{\ME}\ (\Id\ (\J\ Q\ P\ d\ y\ q)\ (d\ q)) \tag*{\defiEnd}
  \end{alignat*}
\end{defi}

\begin{con}
  Assume that $\MC$, $\MD$ and $\ME$ are internal families equipped with identity type introduction structures and identity type elimination structures from $\MC$ to $\MC$, $\MD$ and $\ME$, from $\MD$ to $\MD$ and $\ME$ and from $\ME$ to $\ME$.

  Then we can construct a parametrized identity type elimination structure from $\MC$ to $\ME$ with parameters in $\MD$.
  \defiEnd
\end{con}

The Frobenius variant of the Paulin-Mohring identity type eliminator is exactly a parametrized identity type elimination structure with parameters in the family of telescopes.

\section{The homotopy theory of cCwFs with weak identity types}\label{sec:homotopy_theory}

\subsection{Contextual equivalences}
We recall the classes of local weak equivalences, local trivial fibrations and local fibrations introduced by \cite{HoThTT}.
We will use the adjective \defemph{contextual} instead of local: a property of CwFs (or cumulative CwFs) is said to be contextual when it holds for a CwF $\CC$ if and only if it holds for all contextual slices $(\CC \sslash \Gamma)$, similarly, a property of CwF morphisms is said to be contextual when it holds for a morphism $F : \CC \to \CD$ if and only if it holds for all restrictions $(F \sslash \Gamma) : (\CC \sslash \Gamma) \to (\CD \sslash F\ \Gamma)$ to contextual slices.

\begin{defi}
  A morphism $F : \CC \to \CD$ of cumulative CwFs is said to be a \defemph{contextual trivial fibration} or a \defemph{strong contextual equivalence} if its actions on types and terms are surjective, i.e. if it satisfies the following type and term lifting properties.
  \begin{description}
    \item[strong type lifting] For every $\Gamma : \CC$ and type $A : \abs{\Ty_{\CD}}_{F\ \Gamma}$, there exists a lift $A_{0} : \abs{\Ty_{\CC}}_{\Gamma}$ such that $F\ A_{0} = A$.
    \item[strong term lifting] For every $\Gamma : \CC$, type $A : \abs{\Ty_{\CC}}_{\Gamma}$ and term $a : \abs{\Ty_{\CD}}_{F\ \Gamma}\ (F\ A)$, there exists a lift $a_{0} : \abs{\Tm_{\CC}}_{\Gamma}\ A$ such that $F\ a_{0} = a$.
      \defiEnd
  \end{description}
\end{defi}

\begin{defi}
  A morphism $F : \CC \to \CD$ of cumulative CwFs, where $\CD$ is equipped with weak identity types, is said to be a \defemph{weak contextual equivalence} if its actions on types and terms are surjective up to weak equality, i.e. if it satisfies the following weak type and term lifting properties.
  \begin{description}
    \item[weak type lifting] For every $\Gamma : \CC$ and type $A : \abs{\Ty_{\CD}}_{F\ \Gamma}$, there exists a lift $A_{0} : \abs{\Ty_{\CC}}_{\Gamma}$ and a weak equality in $\Id\ \{\UU\}\ (F\ A_{0})\ A$.
    \item[weak term lifting] For every $\Gamma : \CC$, type $A : \abs{\Ty_{\CC}}_{\Gamma}$ and term $a : \abs{\Ty_{\CD}}_{F\ \Gamma}\ (F\ A)$, there exists a lift $a_{0} : \abs{\Tm_{\CC}}_{\Gamma}\ A$ and a weak equality in $\Id\ \{F\ A\}\ (F\ a_{0})\ a$.
      \defiEnd
  \end{description}
\end{defi}

\begin{rem}\label{rem:split_equiv}
  A (weak or strong) contextual equivalence is said to be \defemph{split} if it comes equipped with a choice of (weak or strong) lifts.
  Classically, any (weak or strong) contextual equivalence can be split.
  Thus we won't distinguish split contextual equivalences from general contextual equivalences in this paper.

  Hofmann's conservativity theorem states that the morphism $\Init_{\mathsf{ITT}} \to \Init_{\mathsf{ETT}}$ from the initial model of Intensional Type Theory to the initial model of Extensional Type Theory is a strong contextual equivalence.
  Constructively, the morphism $\Init_{\mathsf{ITT}} \to \Init_{\mathsf{ETT}}$ is not a split strong contextual equivalence, since that would provide a way to decide equality of terms of ETT.
\end{rem}

Denote by $I$ the set containing the cCwF morphisms $I^{\Ty_{n}} : \mathsf{Free}(\bm{\Gamma} \vdash) \to \mathsf{Free}(\bm{\Gamma} \vdash \bm{A} : \Ty_{n})$ and $I^{\Tm_{n}} : \mathsf{Free}(\bm{\Gamma} \vdash \bm{A} : \Ty_{n}) \to \mathsf{Free}(\bm{\Gamma} \vdash \bm{a} : \Tm_{n}\ \bm{A})$.
A split strong contextual equivalence is exactly a map with the right lifting property with respect to $I$.
A map with the left lifting property with respect to all strong contextual equivalences is called a \defemph{cofibration}.
The small object argument ensures that any map can be factored as a cofibration followed by a strong contextual equivalence.

If $\CC$ is any cCwF and $A$ is a type of $\CC$ in a context $\Gamma$,  we will write $\CC \to \CC[\Gamma \vdash \bm{a} : A]$ for the extension of $\CC$ by a new term $\bm{a}$ of type $A$ in context $\Gamma : \abs{\CC}$, i.e. the following pushout of $I^{\Tm_{n}}$ along the map $(\Gamma,A) : \mathsf{Free}(\bm{\Gamma} \vdash \bm{A} : \Ty) \to \CC$ that sends $\bm{\Gamma}$ to $\Gamma$ and $\bm{A}$ to $A$.
\[ \begin{tikzcd}
    \mathsf{Free}(\bm{\Gamma} \vdash \bm{A} : \Ty_{n}) \ar[d, "I^{\Tm_{n}}"] \ar[r, "{(\Gamma, A)}"] \ar[rd, very near end, phantom, "\ulcorner"] & \CC \ar[d] \\
    \mathsf{Free}(\bm{\Gamma} \vdash \bm{a} : \Tm_{n}\ \bm{A}) \ar[r, "\bm{a}"] & \CC[\Gamma \vdash \bm{a} : A]
  \end{tikzcd} \]
The universal property of $\CC[\Gamma \vdash \bm{a} : A]$ says that a morphism $F : \CC[\Gamma \vdash \bm{a} : A] \to \CD$ is determined by a morphism $F' : \CC \to \CD$ along with a term $a : \abs{\Tm_{\CD}}_{F'\ \Gamma}\ (F'\ A)$.
Extensions of the form $\CC \to \CC[\Gamma \vdash \bm{a} : A]$ are called \defemph{basic $I$-cellular extensions}.

Remark that thanks to the presence of universes, a pushout of $I^{\Ty_{n}}$ is also a basic $I$-cellular extension.
For the same reason, the type lifting properties are redundant in the definitions of strong and weak contextual equivalences.

We recall that use bold symbols ($\bm{\Gamma}$, $\bm{A}$, $\bm{a}$, ...) to indicate the generators of a free model or of a free extension of a model.
Thus, when we write $\CC[\Gamma \vdash \bm{a} : A]$, $\bm{a}$ is a new term of $\CC[\Gamma \vdash \bm{a} : A]$, whereas $\Gamma$ and $A$ already exist in $\CC$.

Denote by $J$ the set consisting of the morphisms $J^{\Ty_{n}} : \mathsf{Free}(\bm{\Gamma} \vdash \bm{A} : \Ty_{n}) \to \mathsf{Free}(\bm{\Gamma} \vdash \bm{e} : \Id\ \{\UU_{n}\}\ \bm{A}\ \bm{B})$ and $J^{\Tm_{n}} : \mathsf{Free}(\bm{\Gamma} \vdash \bm{a} : \Tm_{n}\ \bm{A}) \to \mathsf{Free}(\bm{\Gamma} \vdash \bm{p} : \Tm\ (\Id\ \{\bm{A}\}\ \bm{a}\ \bm{b}))$.
A map with the right lifting property with respect to $J$ is called a \defemph{contextual fibration}, and maps with the left lifting property with respect to all contextual fibrations are called \defemph{trivial cofibrations}.
By the small object argument, any map can also be factored functorially as a trivial cofibration followed by a contextual fibration.

Given a cCwF $\CC$ and a term $a : \abs{\Tm_{\CC}}_{\Gamma}\ A$, we write $\CC \to \CC[\Gamma \vdash \bm{p} : \Id\ \{A\}\ a\ \bm{b}]$ for the extension of $\CC$ by a new term $\bm{b}$ of type $A$ and a new path $\bm{p}$ of type $\Id\ \{A\}\ a\ b$.
Extensions of this kind are called \defemph{basic $J$-cellular extensions}.

Note that while the classes of weak contextual equivalences, strong contextual equivalences and contextual fibrations are independent of the additional type-theoretic structure that we consider, this is not the case for the classes of cofibrations and trivial cofibrations.
Thus, we have to be careful, when working with models of a theory $\Th$, to use the correct notions in $\CMod_{\Th}$, namely the classes of morphisms with the left lifting property with respect to the strong contextual equivalences and contextual fibrations of $\CMod_{\Th}$.

All classes coincide however in the categories $\CMod_{\Th}$ of general models and $\CMod^{\cxl}_{\Th}$ of contextual models.
This is clear for the classes of weak contextual equivalences, strong contextual equivalences and contextual fibrations, since they are defined as contextual properties on morphisms.
We also show that it holds for trivial cofibrations and cofibrations.
\begin{prop}
  Let $F : \CMod_{\Th}^{\cxl}(\CC \to \CD)$ be a morphism between contextual models of a theory $\Th$.
  The morphism $F$ has the left lifting property with respect to all strong contextual equivalences (resp. trivial fibrations) if and only if it has the left lifting property with respect to the strong contextual equivalences (resp. contextual fibrations) between contextual models.
\end{prop}
\begin{proof}
  The forward implications are straightforward.
  For the reverse implications, assume that $F$ has the left lifting property with respect to all strong contextual equivalences (resp. contextual fibrations) and take a lifting problem
  \[ \begin{tikzcd}
      \CC \ar[d, "F"'] \ar[r, "\alpha"] & \CA \ar[d, "G"] \\
      \CD \ar[r, "\beta"'] & \CB \rlap{\ ,}
    \end{tikzcd} \]
  where $G$ is a strong contextual equivalence (resp. contextual fibration).

  We can consider the contextual image factorizations $\CC \to \cxlim \alpha \to \CA$ and $\CD \to \cxlim \beta \to \CB$ of the horizontal maps $\alpha$ and $\beta$, and the induced map $\cxlim \alpha \to \cxlim \beta$.
  Since $\cxlim \alpha \to \CA$ and $\cxlim \beta \to \CB$ are contextual isomorphisms, the induced map $\cxlim \alpha \to \cxlim \beta$ is also a strong contextual equivalence (resp. contextual fibration).
   \[ \begin{tikzcd}
       \CC \ar[ddd, "F"'] \ar[rr, "\alpha"] \ar[rd] && \CA \ar[ddd, "G"] \\
       & \cxlim \alpha \ar[d] \ar[ru] &\\
       & \cxlim \beta \ar[rd] &\\
      \CD \ar[rr, "\beta"'] \ar[ru] \ar[ruu, dashed] && \CB \rlap{\ ,}
    \end{tikzcd} \]
  Since $\CC$ and $\CD$ are contextual, their contextual extensions $\cxlim \alpha$ and $\cxlim \beta$ are also contextual.
  We can thus find a lift $\CD \to \cxlim \alpha$ in the above diagram.
  The composition $\CD \to \cxlim \alpha \to \CA$ is a solution to the original lifting problem.
\end{proof}

\begin{prop}\label{prop:sweq_props}
  Strong and weak contextual equivalences satisfy the following properties:
  \begin{enumerate}
    \item\label{itm:sweq_props_1} Isomorphisms are contextual isomorphisms.
    \item\label{itm:sweq_props_2} Contextual isomorphisms are strong contextual equivalences.
    \item\label{itm:sweq_props_3} Strong contextual equivalences are weak contextual equivalences.
    \item\label{itm:sweq_props_4} Strong contextual equivalences are closed under composition.
    \item\label{itm:sweq_props_5} Weak contextual equivalences are closed under composition.
    \item\label{itm:sweq_props_6} Given morphisms $F : \CC \to \CD$ and $G : \CD \to \CE$, if $G$ and $F \cdot G$ are weak contextual equivalence, then $F$ is also a weak contextual equivalence.
    \item\label{itm:sweq_props_7} When $\CD$ is contextual, then the remaining 2-out-of-3 condition also holds: if $F$ and $F \cdot G$ are weak contextual equivalences, then $G$ is a weak contextual equivalence.
    \item\label{itm:sweq_props_8} The class of weak contextual equivalences is closed under retracts.
    \item\label{itm:sweq_props_9} A morphism is a strong contextual equivalence if and only if it is a contextual fibration and a weak contextual equivalence.
  \end{enumerate}
\end{prop}
\begin{proof}\hfill
  \begin{description}
    \item[\ref{itm:sweq_props_1},\ref{itm:sweq_props_2},\ref{itm:sweq_props_3}] Straightforward.
    \item[\ref{itm:sweq_props_4}] It is sufficient to check the condition objectwise. It then reduces to the closure of surjective functions are under composition.
    \item[\ref{itm:sweq_props_5}] Assume that $F : \CC \to \CD$ and $G : \CD \to \CE$ are weak contextual equivalences.
      To prove that $F \cdot G : \CC \to \CE$ is also a weak contextual equivalence, we have to check the weak term lifting property (recall that the weak type lifting property is redundant in presence of universes).

      Take a context $\Gamma : \abs{\CC}$, a type $A : \abs{\Ty_{\CC}}_{\Gamma}$ and a term $a : \abs{\Tm_{\CE}}_{G\ (F\ \Gamma)}\ (G\ (F\ A))$.
      By the weak term lifting property of $G$, we have a lift $a' : \abs{\Ty_{\CD}}_{F\ \Gamma}\ (F\ A)$ and a path $p : \abs{\Tm_{\CE}}_{G\ (F\ \Gamma)}\ (\Id\ (G\ a')\ a)$.
      By the weak term lifting property of $F$, we have a term $a'' : \abs{\Tm_{\CC}}_{\Gamma}\ A$ and a path $q : \abs{\Tm_{\CD}}_{F\ \Gamma}\ (\Id\ (F\ a'')\ a')$.

      We then have $G\ q \cdot p : \abs{\Tm_{\CD}}_{G\ (F\ \Gamma)}\ (\Id\ (G\ (F\ a''))\ a)$, as needed.
    \item[\ref{itm:sweq_props_6}] Assume that $G : \CD \to \CE$ and $F \cdot G : \CC \to \CE$ are weak contextual equivalences.
      We check the weak term lifting properties for $F : \CC \to \CD$.

      Take $\Gamma : \abs{\CC}$, a type $A : \abs{\Ty_{\CC}}_{\Gamma}$ and a term $a : \abs{\Tm_{\CD}}_{F\ \Gamma}\ (F\ A)$.
      The weak term lifting property of $F \cdot G$ gives a term $a' : \abs{\Tm_{\CC}}_{\Gamma}\ A$ and a path $p : \abs{\Tm_{\CE}}_{G\ (F\ \alpha)}\ (\Id\ (G\ (F\ a'))\ (G\ A))$.
      The weak term lifting property of $G$ then gives a path $p' : \abs{\Tm_{\CD}}_{F\ \Gamma}\ (\Id\ (F\ A')\ A)$, as desired.
    \item[\ref{itm:sweq_props_7}] Since $\CD$ is contextual, we can identify its objects with closed telescopes of types.
      Given any closed telescope $\Gamma : \abs{\Ty^{\star}_{\CD}}$, we can lift it to a closed telescope $\Gamma_{0} : \abs{\Ty^{\star}_{\CC}}$ along with a path $\alpha : \abs{\Tm_{\CD}^{\star}}\ (\Id\ (F\ \Gamma_{0})\ \Gamma)$.

      Now take some type $A : \abs{\Ty_{\CD}}_{\Gamma}$ and term $a : \abs{\Tm_{\CE}}_{G\ \Gamma}\ (G\ A)$.
      We can transport $A$ over $\alpha$ to obtain $\alpha^{\star}\ A : \abs{\Ty_{\CD}}_{F\ \Gamma_{0}}$.
      Using the fact that $F$ is a weak equivalence again, we lift $(\alpha^{\star}\ A)$ to some $A_{0} : \abs{\Ty_{\CC}}_{\Gamma_{0}}$, along with an equality $p : \abs{\Tm_{\CD}}_{F\ \Gamma_{0}}\ (\Id\ (F\ A_{0})\ (\alpha^{\star}\ A))$.
      We can now transport $a$ over $G\ \alpha$ and $G\ p$ to obtain $a' : \abs{\Tm_{\CE}}_{G\ (F\ (\Gamma_{0}))}\ (G\ (F\ A_{0}))$.
      Now using the weak lifting property of $F \cdot G$, we get $a'' : \abs{\Tm_{\CC}}_{\Gamma_{0}}\ A_{0}$ and a path $q : \abs{\Tm_{\CE}}_{G\ (F\ \Gamma_{0})}\ (\Id\ (G\ (F\ a''))\ a')$.
      Transporting $F\ a' : \abs{\Tm_{\CD}}_{F\ \Gamma_{0}}\ (F\ A_{0})$ over $\alpha$ and $p$, we obtain our desired lift $a_{0} : \abs{\Tm_{\CD}}_{\Gamma}\ A$.
      The fact that it is indeed a lift of $a$ can be derived from $q$ and the fact that the transports over $G\ \alpha$, $G\ p$, $\alpha$ and $p$ cancel each others.
    \item[\ref{itm:sweq_props_8}]
      Consider the following commutative diagram, where $F$ is a weak contextual equivalence, $s_{1} \cdot r_{1} = \id$ and $s_{2} \cdot r_{2} = \id$.
      \[ \begin{tikzcd}
          \CA \ar[d, "G"] \ar[r, "s_{1}"] & \CC \ar[d, "F", "{\rotatebox{90}{$\sim$}}"'] \ar[r, "r_{1}"] & \CA \ar[d ,"G"] \\
          \CB \ar[r, "s_{2}"] & \CD \ar[r, "r_{2}"] & \CB
        \end{tikzcd} \]
      We have to prove that $G$ is also a weak contextual equivalence.

      Take an object $\Gamma : \abs{\CA}$, a type $A : \abs{\Ty_{\CA}}_{\Gamma}$ and a term $a : \abs{\Tm_{\CB}}_{G\ \Gamma}\ A$.
      We can find a weak lift $a' : \abs{\Tm_{\CC}}_{s_{1}\ \Gamma}\ (s_{1}\ A)$ along with a path between $F\ a'$ and $s_{2}\ a$.
      We then have $r_{1}\ a' : \abs{\Tm_{\CA}}_{r_{1}\ (s_{1}\ \Gamma)}\ (r_{1}\ (s_{1}\ A))$ and a path between $r_{2}\ (F\ a')$ and $r_{2}\ (s_{2}\ a)$.
      We can see that this is our desired weak lift: we have $r_{1}\ A' : \abs{\Tm_{\CA}}_{\Gamma}\ A$ and the path is between $G\ (r_{1}\ A')$ and $A$.
    \item[\ref{itm:sweq_props_9}]
      A strong contextual equivalence is clearly both a contextual fibration and a weak contextual equivalence.

      For the reverse inclusion, take a morphism $F : \CC \to \CD$ that is both a contextual fibration and a weak contextual equivalence.
      We show that $F$ satisfies the strong term lifting property.
      Take an object $\Gamma : \abs{\CC}$, a type $A : \abs{\Ty_{\CC}}_{\Gamma}$ and a term $a : \abs{\Tm_{\CD}}_{F\ \Gamma}\ (F\ A)$.
      From the weak term lifting property, we obtain a term $a' : \abs{\Tm_{\CC}}_{\Gamma}\ A$ along with a path $p$ between $a$ and $F\ a'$.
      Since $F$ is a fibration, we can lift this path to a path $p'$ in $\CC$ between some term $a''$ and the term $a'$, such that $F\ a'' = a$ and $F\ p' = p$.
      The term $a''$ is then a strong lift of $a$.
  \end{description}
\end{proof}

The classes of cofibrations and strong contextual equivalences are completely determined by the dependent sorts of the presentation of a type theory by a signature.
In this paper, we mainly consider theories with two families of sorts: the types and terms for each universe level.
This is why the set $I$ of generating cofibrations contains exactly the families of maps $I^{\Ty_{n}}$ and $I^{\Tm_{n}}$.
For type theories with richer contextual structures, such as two-level type theories (with additional sorts for outer types and terms) or cubical type theories (with an additional sort for the interval), the set $I$ would contain an additional element for each additional sort.

The classes of weak contextual equivalences, trivial cofibrations and contextual fibrations are however not completely determined by the presentation of the type theory.
Indeed they require choosing a suitable notion of equivalence or weak equality for each sort of the theory.
At the level of model structures, this would correspond to the choice of a relative cylinder object for each map in $I$.
In this paper we only consider the notion of weak equality provided by the identity types, but other choices may be possible.
For example, we could compare types up to equivalence, instead of comparing their codes up to equality.
The fact that a chosen notion of equivalence/weak equality is good for some theory can be tested by the fact that the classes of weak contextual equivalences, trivial fibrations and fibrations define some left semi-model structure on the category of contextual models of that type theory.

\begin{defi}\label{def:wb_tcof}
  We say that a type theory $\Th$ extending the theory of weak identity types is \defemph{semi-model} if the classes of weak contextual equivalences, trivial fibrations and fibrations constitute a left-semi model structure on the category $\CMod^{\cxl}_{\Th}$ of contextual models of $\Th$.
  \defiEnd
\end{defi}

We don't recall the definition of left semi-model structure in this paper, the only consequence of its definition that we use is the following proposition (altough we conjecture that in this setting, this consequence is sufficient to ensure that $\Th$ is semi-model).
\begin{prop}
  If a type theory $\Th$ is semi-model, then given any cofibrant and contextual model $\CC$, any trivial cofibration $j : \CMod_{\Th}(\CC \to \CD)$ under $\CC$ is a weak equivalence.
  \qed
\end{prop}

The main result of \cite{HoThTT} is the following theorem.
\begin{thm*}[{\cite[Theorem 6.9]{HoThTT}}]
  The classes of weak contextual equivalences, trivial fibrations and fibrations define a left-semi model structure on the category of contextual CwFs equipped with strong $\Id$-, $\Unit$-, $\Sigma$- (and optionally $\Pi$-) type structures.
\end{thm*}
Note that \cite[Theorem 6.9]{HoThTT} applies to CwFs without universe hierarchies, for which types are compared up to equivalence rather than equality of codes.

\subsection{Some useful weak contextual equivalences}\label{ssec:useful_weq}

\begin{prop}\label{prop:contr_weq}
  Let $\CC$ be any contextual model, and let $A : \abs{\Ty^{\star}_{\CC}}_{\diamond}$ be a closed telescope of types.
 
  The following conditions are equivalent:
  \begin{enumerate}
    \item\label{itm:contr_weq_1} The type $A$ is contractible.
    \item\label{itm:contr_weq_2} The map $j : \CC \to \CC[\bm{a} : A]$ is a weak contextual equivalence.
    \item\label{itm:contr_weq_3} The map $j : \CC \to \CC[\bm{a} : A]$ admits a retraction $r : \CC[\bm{a} : A] \to \CC$ that is a weak contextual equivalence.
  \end{enumerate}
\end{prop}
\begin{proof}
  By $2$-out-of-$3$, (\ref{itm:contr_weq_2}) and (\ref{itm:contr_weq_3}) are equivalent.
  It is easy to see that (\ref{itm:contr_weq_2}) implies that $A$ is contractible.

  We now assume that $A$ is contractible, and check that the map $r : \CC[\bm{a} : A] \to \CC$ that maps $\bm{a}$ to the center of contraction $a_{0}$ of $A$ satisfies the weak term lifting property.
  Let $\Gamma$ be a context of $\CC[\bm{a} : A]$, $B$ be a type over $\Gamma$ and $b_{0}$ be a term of type $r\ B$.
  Since $\CC[\bm{a} : A]$ is contextual, we can view $\Gamma$ as a telescope of types.

  We can now bring everything into the internal language of $\CPsh\ \CC$.
  We have $A : \Ty^{\star}_{\CC}$, its center of contraction $a_{0} : \Tm^{\star}_{\CC}\ A$, $\Gamma : \Tm^{\star}_{\CC}\ A \to \Ty^{\star}_{\CC}$, $B : (a : \Tm^{\star}\ A) \to \Tm^{\star}_{\CC}\ (\Gamma\ a) \to \Ty_{\CC}$ and $b_{0} : (\gamma : \Tm^{\star}\ (\Gamma\ a_{0})) \to \Tm_{\CC}\ (B\ a_{0}\ \gamma)$, and we need to construct some $b : (a : \Tm^{\star}_{\CC}\ A) (\gamma : \Tm^{\star}_{\Gamma}\ a) \to \Tm\ (B\ a\ \gamma)$ such that for any $\gamma : \Tm^{\star}_{\CC}\ (\Gamma\ a_{0})$, we have an element of $\Id\ (b\ a_{0}\ \gamma)\ (b_{0}\ \gamma)$.

  This weak lift $b$ can be obtained from the parametrized identity type eliminator derived in \cref{ssec:param_id_elim}, by transporting over paths obtained from the contractibility of $A$.
\end{proof}

\begin{prop}\label{prop:weq_singl}
  Given any contextual model $\CC$, closed type $A : \abs{\Ty_{\CC}}_{\diamond}$ and closed term $a : \abs{\Tm_{\CC}}_{\diamond}\ A$, the basic $J$-cellular extension
  \[ \CC \to \CC[\bm{b} : A, \bm{p} : \Id\ a\ \bm{b}] \] is a weak contextual equivalence.
\end{prop}
\begin{proof}
  Direct application of \cref{prop:contr_weq}, using the fact that the type $(b : A) \times (p : \Id\ a\ b)$ is contractible.
\end{proof}

\begin{prop}\label{prop:weq_pi_singl}
  If a type theory $\Th$ has $\Pi$-types with a strict $\beta$-rule, then given any contextual model $\CC$, any basic $J$-cellular extension
  \[ \CC \to \CC[(\gamma : \Gamma) \vdash \bm{b}(\gamma) : A, (\gamma : \Gamma) \vdash \bm{p}(\gamma) : \Id\ a(\gamma)\ \bm{b}(\gamma)] \] is a weak contextual equivalence.
\end{prop}
\begin{proof}
  We pose $\CC_{0} \triangleq \CC[(\gamma : \Gamma) \vdash \bm{b}(\gamma) : A, (\gamma : \Gamma) \vdash \bm{p}(\gamma) : \Id\ a(\gamma)\ \bm{b}(\gamma)]$ and $j_{0} : \CC \to \CC_{0}$.
 
  Since $\CC$ is contextual, we can see the context $\Gamma$ as a telescope of types.
  We have shown in \cref{ssec:lift_tele} that the $\Pi$-types of $\CC$ can be lifted to the families of telescopes.
  In particular, we can form $\Pi$-types with $\Gamma$ as the domain.

  We consider the model $\CC_{1} \triangleq \CC[\bm{c} : \Pi\ \Gamma\ A, \bm{q} : \Pi\ \Gamma\ (\gamma \mapsto \Id\ a(\gamma)\ (\app\ \bm{c}\ \gamma)) ]$.
  The contractibility of $(c : \Pi\ \Gamma\ A) \times (q : \Pi\ \Gamma\ (\gamma \mapsto \Id\ a(\gamma)\ (\app\ c\ \gamma)))$ is one of the equivalent characterizations of function extensionality.
  By \cref{prop:contr_weq}, the map $j_{1} : \CC \to \CC_{1}$ is thus a weak equivalence.

  We have a map $F : \CC_{0} \to \CC_{1}$ that sends $\bm{b}$ to $\app\ \bm{c}$ and $\bm{p}$ to $\app\ \bm{q}$.
  We also have a map $G : \CC_{1} \to \CC_{0}$ that sends $\bm{c}$ to $\lam\ \bm{b}$ and $\bm{q}$ to $\lam\ \bm{p}$.
  The fact that $\bm{q}$ can be sent to $\lam\ \bm{p}$ relies on the strict $\beta$-rule.

  The strict $\beta$-rule also implies that $G$ is a retraction of $F$, i.e. that $F \cdot G = \id$.

  This implies that $j_{0}$ is a retraction of $j_{1}$: the following diagram commutes.
  \[ \begin{tikzcd}
      \CC \ar[d, "j_{0}"] \ar[r, equal] & \CC \ar[d, "j_{1}"] \ar[r, equal] & \CC \ar[d, "j_{0}"] \\
      \CC_{0} \ar[r, "F"] & \CC_{1} \ar[r, "G"] & \CC_{0}
    \end{tikzcd} \]
  Since weak equivalence are closed under retracts, $j_{0}$ is a weak equivalence.
\end{proof}

General recognition theorems for left semi-model structures and \cref{prop:weq_pi_singl} should imply that any type theory over the theory of cumulative CwFs with weak identity types and $\Pi$-types with a strict $\beta$-rule is semi-model.

\subsection{Cellular models}\label{ssec:cellular_models}

The (cofibration, strong contextual equivalence) weak factorization system constructed by the small object argument gives us a way to replace any model of a theory $\Th$ by an equivalent cellular model.
The cellular models are those that are freely generated by a collection of types and terms.
This is convenient, since many theorems that are traditionally established for the initial model (such as normalization, ...) can actually be expected to hold for all cellular models, which share the syntactic nature of the initial model.
In this subsection we introduce some notations and recall some of the basic properties of cellular models.
We work with a fixed type theory signature $\Th$ extending the theory of weak identity types.

We use a coinductive definition of cellular extensions, i.e. extensions of a model by a collection of new types and new terms.
Because we work with cumulative CwFs with universes, it is sufficient to consider extensions by a collection of new terms.
\begin{defi}\label{def:cellular_models}
  A \defemph{cellular extension} $X$ over a model $\CC$ of $\Th$ consists of a family $X_{0}^{\Tm_{n}} : (\Gamma : \abs{\CC}) \to \abs{\Ty_{n,\CC}}\ \Gamma \to \SSet$ over types of $\CC$ for each universe level $n$, and a further cellular extension ${\uparrow}X$ over the model $\CC[X_{0}]$, which is defined as the free extension \[ \CC[X_{0}] \triangleq \CC[\{ \Gamma \vdash \bm{a} : A \mid n \in \Nat, a \in X_{0}^{\Tm_{n}}\ \Gamma\ A \}], \]
  or equivalently as the following pushout
  \[
    \begin{tikzcd}
      {\underset{n : \Nat, x : X_{0}^{\Tm_{n}}\ \Gamma\ A}{\coprod} \mathsf{Free}(\bm{\Gamma} \vdash \bm{A} : \Ty_{n})} \ar[d] \ar[rd, very near end, phantom, "\ulcorner"] \ar[r] & \CC \ar[d] \\
      {\underset{n : \Nat, x : X_{0}^{\Tm_{n}}\ \Gamma\ A}{\coprod} \mathsf{Free}(\bm{\Gamma} \vdash \bm{a} : \Tm_{n}\ \bm{A})} \ar[r] & \CC[X_{0}] \rlap{\ .}
    \end{tikzcd} \]
  In other words, $\CC[X_{0}]$ is the free extension of $\CC$ by a family of terms indexed by $X_{0}^{\Tm_{n}}$ at each universe level $n$.
  There are no dependencies between the added types and terms; dependencies are instead encoded by iterating this construction, possibly a countably infinite number of times.

  A cellular extension $X$ generates a sequence
  \[ \CC \to \CC[X_{0}] \to \CC[X_{0}][X_{1}] \to \cdots \to \CC[X_{<m}] \to \cdots \] of models (where $X_{1} = ({\uparrow}X)_{0}$, etc).
  We write $\CC[X]$ for the colimit of this sequence.

  We write $X^{\Tm_{n}} : (\Gamma : \abs{\CC[X]}) \to \abs{\Ty_{n,\CC[X]}}\ \Gamma \to \SSet$ for the family of generating terms of $\CC[X]$, i.e. the coproduct of $(X_{m}^{\Tm_{n}})$ over $m : \Nat$.
  Given an element $a : X^{\Tm_{n}}\ \Gamma\ A$, we denote the corresponding term of $\CC[X]$ by $\bm{a} : \abs{\Tm_{\CC[X]}}\ \Gamma\ A$.

  A cellular model is a model obtained as cellular extension $\Init_{\Th}[X]$ of the initial model $\Init_{\Th}$.
  \defiEnd
\end{defi}

\begin{prop}
  For any cellular extension $\CC[X]$, the map $\CC \to \CC[X]$ is a contextual extension.
\end{prop}
\begin{proof}
  This follows from the fact that the same set $I$ of maps is used to generate the (contextual extension, contextual isomorphism) orthogonal factorization system and the (cofibration, strong contextual equivalence) weak factorization system.
\end{proof}

\begin{cor}
  Any cellular extension $\CC[X]$ of a contextual model $\CC$ is contextual.
  \qed
\end{cor}

There is also a relation between cellular models and the theory of type theory signatures defined in \cref{def:th_sig}.
Indeed the finite cellular models correspond exactly to the possible premises of the operations and equations of a signature.

For example, the premises of the identity type former can be encoded by the cellular model $\Init_{\Th}[\bm{A} : \UU, \bm{x} : A, \bm{y} : A]$.
The premises of the $\Pi$-type former would be encoded by the cellular model $\Init_{\Th}[\bm{A} : \UU, (a : \bm{A}) \vdash \bm{B} : \UU]$.
The $\Id$ and $\Pi$ type-theoretic operations can be seen as the types $\Id\ \{\bm{A}\}\ \bm{x}\ \bm{y}$ and $\Pi\ \bm{A}\ \bm{B}$ of these models.

Thus we can often perform constructions for all operations of the theory $\Th$ by looking at the types and terms of finite cellular models.

\subsection{Fibrant congruences and quotients}\label{ssec:fibrant_cong}

Since the categories of models of type theories are complete and cocomplete, there is a notion of internal equivalence relation on a model of type theory (where ``internal'' here means that the concept is defined using objects and arrows of the category), and moreover any internal equivalence relation has a quotient.
However, general quotients may be ill-behaved, and are hard to compute.
This is already the case for quotients and colimits of categories; originally distinct objects may be identified in the quotient, and originally non-composable morphisms may then become composable in the quotients, leading to new morphisms that do not correspond to any morphism of the base category.

In this subsection, we define a smaller class of congruences, which we call fibrant congruences, for which the quotients are better behaved and can be computed pointwise.
We show that the strong contextual equivalences are, up to contextual isomorphism, exactly the quotients by fibrant congruences.

\begin{defi}
  Internally to a presheaf model $\CPsh\ \CC$, an internal \defemph{fibrant contextual congruence} $\widetilde{\MC}$ over an internal cumulative family $\MC = (\Ty,\Tm)$ consists of:
  \begin{enumerate}
    \item an internal equivalence relation $\widetilde{\Ty}$ on $\Ty_{n}$ for each universe level $n$:
      \begin{alignat*}{2}
        & \widetilde{\Ty} && : \{n\} \to \Ty_{n} \to \Ty_{n} \to \SProp.
      \end{alignat*}
      We will often write $(A \sim B) \in \widetilde{\MC}$ or $(A \sim B)$ instead of $\widetilde{\Ty}\ A\ B$.
    \item internal equivalence relations $\widetilde{\Tm}$ on $\Tm_{n}$, displayed over $\widetilde{\Ty}$:
      \begin{alignat*}{2}
        & \widetilde{\Tm} && : \{n,A,B\} \to (A \sim B) \to \Tm_{n}\ A \to \Tm_{n}\ B \to \SProp.
      \end{alignat*}
      We will often write $(a \sim_{p} b) \in \widetilde{\MC}$, $(a \sim_{p} b)$, or just $(a \sim b)$, instead of $\widetilde{\Tm}\ p\ a\ b$.
    \item such that for every universe level $n$, $(\Tm_{n}, \widetilde{\Tm})$ is (internally) a fibrant setoid family over $(\Ty_{n}, \widetilde{\Ty})$, i.e. for every pair of congruent types $p : (A \sim B)$ and term $a : \Tm_{n}\ A$, there exists a transported term $(p^{\star}\ a) : \Tm\ B$ such that $(a \sim b) \in \widetilde{\MC}$.
    \item such that the operations $\mathsf{Lift}_{\Ty} : \Ty_{n} \to \Ty_{n+1}$ and the isomorphisms $\mathsf{lift}_{\Tm} : \Tm_{n}\ A \simeq \Tm_{n+1}\ (\mathsf{Lift}_{\Ty}\ A)$ and $\Tm\ \UU_{n} \simeq \Ty_{n}$ preserve the equivalence relations.
    \item\label{itm:cong_compat} such that dependent types and terms have actions on the relations:
      \begin{itemize}
        \item for every telescope $A : \Ty^{\star}$ and dependent type $B : \Tm^{\star}\ A \to \Ty$, whenever $(a_{1} \sim a_{2}) \in \widetilde{\MC}$ are congruent (the equivalence relations are extended pointwise to telescopes) telescopes of terms of type $A$, then $(B\ a_{1} \sim B\ a_{2}) \in \widetilde{\MC}$;
        \item for every telescope $A : \Ty^{\star}$, dependent type $B : \Tm^{\star}\ A \to \Ty$ and dependent term $b : (a : \Tm^{\star}\ A) \to \Tm\ (B\ a)$, whenever $(a_{1} \sim a_{2}) \in \widetilde{\MC}$ are congruent telescopes of terms of type $A$, then $(b\ a_{1} \sim b\ a_{2}) \in \widetilde{\MC}$.
      \defiEnd
      \end{itemize}
  \end{enumerate}
\end{defi}

Given two congruent types $p : A_{1} \sim A_{2}$ and two dependent types $B_{1} : \Tm\ A_{1} \to \Ty$ and $B_{2}  : \Tm\ A_{2} \to \Ty$, there are several way to define a relation between $B_{1}$ and $B_{2}$.
\begin{description}
  \item[Unbiased] $B_{1} \sim B_{2}$ when for every $a_{1} : \Tm\ A_{1}$ and $a_{2} : \Tm\ A_{2}$ such that $(a_{1} \sim a_{2})$, we have $(B_{1}\ a_{1} \sim B_{2}\ a_{2})$.
  \item[Left-biased] $B_{1} \sim B_{2}$ when for every $a_{1} : \Tm\ A_{1}$, we have $(B_{1}\ a_{1} \sim B_{2}\ (p^{\star}\ a_{1}))$.
  \item[Right-biased] $B_{1} \sim B_{2}$ when for every $a_{2} : \Tm\ A_{2}$, we have $(B_{1}\ ((p^{-1})^{\star}\ a_{2}) \sim B_{2}\ a_{2})$.
\end{description}
The last component (\ref{itm:cong_compat}) of the definition of fibrant congruence ensures that they are all equivalent.

We say that a fibrant congruence $\widetilde{\MC}$ over an internal cumulative family $\MC$ is compatible with a theory $\Th$ if the operations of $\Th$ all preserve the equivalence relations of $\widetilde{\MC}$.

\begin{prop}\label{prop:cong_internal_quotient}
  Let $\widetilde{\MC}$ be a fibrant contextual congruence over an internal cumulative family $\MC$.
  Then there is a quotient internal cumulative family $(\MC / \widetilde{\MC})$ along with a morphism $q : \MC \to (\MC / \widetilde{\MC})$ of internal cumulative families, such that for every pair $(A \sim B) \in \widetilde{\MC}$ of congruent types, $q\ A = q\ B$, and for every pair $(a \sim_{p} b) \in \widetilde{\MC}$ of congruent terms, $q\ a = q\ b$.

  Furthermore, $q : \MC \to (\MC / \widetilde{\MC})$ is surjective on both types and terms, and the quotient is effective: given any two types (or terms) $x, y$, we have $q\ x = q\ y$ if and only if $(x \sim y) \in \widetilde{\MC}$.
\end{prop}
\begin{proof}
  We don't look at universe levels in this proof; the quotient can be defined levelwise.

  Write $\widetilde{\Ty}$ and $\widetilde{\Tm}$ for the equivalence relations of $\widetilde{\MC}$.

  We define $\Ty_{(\MC / \widetilde{\MC})}$ as the quotient of $\Ty_{\MC}$ by the equivalence relation $\widetilde{\Ty}$.
  We have a quotienting map $q_{\Ty} : \Ty_{\MC} \to \Ty_{(\MC / \widetilde{\MC})}$.

  We would like to define $\Tm_{(\MC / \widetilde{\MC})}\ (q\ A)$ as a quotient of $\Tm_{\MC}\ A$ for every $A : \Ty_{\MC}$.
  This is however not possible in a non-univalent metatheory, as this would require an equality between the quotients $(\Tm_{\MC}\ A) / \widetilde{\Tm}$ and $(\Tm_{\MC}\ B) / \widetilde{\Tm}$ for every pair $(A \sim B) \in \widetilde{\MC}$ of congruent types.
  The fibrancy of the congruence $\widetilde{\MC}$ only provides an isomorphism $(\Tm_{\MC}\ A) / \widetilde{\Tm} \simeq (\Tm_{\MC}\ B) / \widetilde{\Tm}$.

  Instead, we define $\Tm_{(\MC / \widetilde{\MC})}\ A$ as the quotient of the presheaf $(B : \Ty_{\MC}) \times (b : \Tm_{\MC}\ B) \times (q_{\Ty}\ B = A)$ by the relation $(\sim)$ defined by $(B, b, -) \sim (C, c, -) \triangleq (b \sim c)$.
  Then for every $A : \Ty_{\MC}$, we have an isomorphism $\Tm_{(\MC / \widetilde{\MC})}\ (q\ A) \simeq (\Tm_{\MC}\ A) / \widetilde{\Tm}$.
\end{proof}

\begin{defi}
  Let $F : \CC \to \CD$ be a cCwF morphism.
  The \defemph{kernel} $\ker F$ of $F$ consists of equivalence relations on types and terms defined by:
  \begin{alignat*}{2}
    & (A \sim B) \in \ker F && \triangleq F\ A = F\ B \\
    & (a \sim b) \in \ker F && \triangleq F\ a = F\ b 
  \end{alignat*}
  Remark that $\ker F$ does not necessarily satisfy the fibrancy condition of the definition of fibrant contextual congruence, but it satisfies all of the other conditions.

  If $F$ is also a morphism of models of some type theory $\Th$, then $\ker F$ is compatible with the operations of $\Th$.
  \defiEnd
\end{defi}

\begin{prop}\label{prop:fib_cong_quotient}
  Let $\Th$ be a type theory extending the theory of cumulative CwFs with universes.
  If $\widetilde{\CC}$ is a fibrant contextual congruence on a contextual model $\CC : \CMod^{\cxl}_{\Th}$, then it has a quotient $\Quot_{\widetilde{\CC}} : \CMod^{\cxl}_{\Th}$ and a quotient inclusion $\quot_{\widetilde{\CC}} : \CC \to \Quot_{\widetilde{\CC}}$, with the following properties:
  \begin{enumerate}
    \item\label{itm:fib_cong_quotient_1} For every model $\CD : \CMod_{\Th}$ and morphism $F : \CC \to \CD$ such that $\widetilde{\CC} \subseteq \ker F$, there is a unique morphism $G : \Quot_{\widetilde{\CC}} \to \CD$ such that $\quot_{\widetilde{\CC}} \cdot G = F$.
    \item\label{itm:fib_cong_quotient_2}
      The quotient inclusion $\quot_{\widetilde{\CC}} : \CC \to \Quot_{\widetilde{\CC}}$ is a strong contextual equivalence. (Note that splitting $\quot_{\widetilde{\CC}} : \CC \to \Quot_{\widetilde{\CC}}$ required the axiom of choice).
    \item\label{itm:fib_cong_quotient_3}
      The quotient is effective: $\ker \quot_{\widetilde{\CC}} = \widetilde{\CC}$.
      This means that for every pair $a, b$ of terms (or types), $a$ and $b$ are congruent in $\widetilde{\CC}$ if and only if they are identified in $\Quot_{\widetilde{\CC}}$ by $\quot_{\widetilde{\CC}}$.
  \end{enumerate}
\end{prop}
\begin{proof}
  Since $\CMod_{\Th}$ is cocomplete, we can define the quotient $\Quot_{\widetilde{\CC}}$ as the coequalizer
  \[ \begin{tikzcd}
      \underset{\Gamma : \abs{\CC}, (A_{1} \sim A_{2}) \in \widetilde{\CC}, (a_{1} \sim a_{2}) \in \widetilde{\CC}}{\coprod}\ \mathsf{Free}(\bm{\Gamma} \vdash \bm{a} : \bm{A}) \ar[r, shift left, "\pi_{1}"] \ar[r, shift right, "\pi_{2}"'] & \CC \rlap{\ ,}
    \end{tikzcd} \]
    where the coproduct ranges over all pairs $(a_{1} \sim a_{2})$ of congruent terms in all contexts and $\pi_{1}$ and $\pi_{2}$ map $\bm{a}$ respectively to $a_{1}$ and $a_{2}$.
    It then satisfies the universal property (\ref{itm:fib_cong_quotient_1}) by definition.

    We now construct a model $\CQ$ of $\Th$.
    The base category of $\CQ$ is the presheaf category $\widehat{\CC}$.
    A type of $\CQ$ over a presheaf $X$ is a natural transformation $A : X \to \Ty_{\CC/\widetilde{\CC}}$.
    A term over $X$ of type $A$ is a dependent natural transformation $a : (x : X) \to \Tm_{\CC/\widetilde{\CC}}\ (A\ x)$.
    The extension of a context $X$ by a type $A$ is the presheaf $(x : X) \times \Tm_{\CC/\widetilde{\CC}}\ (A\ x)$.
The fact that $\CQ$ is a model of $\Th$ then follows from the compatibility of $\widetilde{\CC}$ with the operations of $\Th$.

    The Yoneda embedding $\yo : \CC \to \widetilde{\CC}$ is not a morphism of cCwFs from $\CC$ to $\CQ$, because it does not preserve the representing objects of context extensions.
    However, using the fact that $\CC$ is contextual, we can define a morphism $F : \CC \to \CQ$ (along with a natural transformation from $\yo$ to $F$).
    The actions of $F$ on contexts and morphisms are defined by induction on their length.
    The actions of $F$ on types and terms are given by the components of the quotienting map of the internal quotient.
    The compatibility of $F$ with substitution follows from the action of dependent types and terms on the relations of $\widetilde{\CC}$ (component (\ref{itm:cong_compat}) of the definition of fibrant congruence).

    By the properties of the internal quotient, $F$ is a strong contextual equivalence and $\ker F = \widetilde{\CC}$.

    By the universal property of $\Quot_{\widetilde{\CC}}$, $F$ factors through $\quot_{\widetilde{\CC}}$.
    \[ \begin{tikzcd}
        \CC \ar[rd, "\quot_{\widetilde{\CC}}"'] \ar[rr, "F"] && \CQ \\
        & \Quot_{\widetilde{\CC}} \ar[ru] &
      \end{tikzcd} \]
    Since $F$ is a strong contextual equivalence, this factorization implies that $\quot_{\widetilde{\CC}}$ is also a strong contextual equivalence.
    The factorization also implies that $\ker \quot_{\widetilde{\CC}} \subseteq \ker F$, and $\widetilde{\CC} \subseteq \ker \quot_{\widetilde{\CC}}$ by definition, so $\ker \quot_{\widetilde{\CC}} = \widetilde{\CC}$.
\end{proof}

\begin{defi}
  Let $\CC$ be a contextual model of a theory $\Th$.
  Given any morphism $F : \CMod_{\Th}(\CC \to \CD)$ whose kernel $\ker F$ is a fibrant contextual congruence, we define its coimage to be the quotient of its kernel: $\coim F \triangleq \Quot_{\ker F}$.
  There is a canonical comparison map $\coim F \to \cxlim F$, obtained by the universal property of the quotient $\coim F$.
  \defiEnd
  \[ \begin{tikzcd}
      \CC \ar[rrr, "F"] \ar[rd, "\quot"'] &&& \CD \\
      & \coim F \ar[r] & \cxlim F \ar[ru] &
    \end{tikzcd} \]
\end{defi}

\begin{prop}
  Given a contextual model $\CC : \CMod_{\Th}^{\cxl}$, a model morphism $F : \CMod_{\Th}(\CC \to \CD)$ is a strong contextual equivalence if and only if its kernel $\ker F$ is a fibrant congruence and the canonical map $\coim F \to \cxlim F$ is an isomorphism.
\end{prop}
\begin{proof} We prove both implications.
  \begin{description}
    \item[($\Ra$)] Assume that $F : \CC \to \CD$ is a strong contextual equivalence.

      To see that $\ker F$ is a fibrant congruence, take a context $\Gamma : \abs{\CC}$, two types $A,B : \abs{\Ty_{n,\CC}}_{\Gamma}$ such that $F\ A = F\ B$ and a term $a : \abs{\Tm_{\CC}}_{\Gamma}\ A$.
      Then $F\ a : \abs{\Tm_{\CD}}_{F\ \Gamma}\ (F\ B)$, and by the strong term lifting property of $F$, we obtain some lift $b : \abs{\Tm_{\CC}}_{\Gamma}\ B$ such that $F\ a = F\ b$.
      This proves that $\ker F$ is fibrant.

      To check that $\coim F \to \cxlim F$ is an isomorphism, it suffices to check that the map $G : \coim F \to \CD$ is a contextual isomorphism.
      Since $F : \CC \to \CD$ factors through $G : \coim\ F \to \CD$, and $F$ is a strong contextual equivalence, $G$ is also a strong contextual equivalence, i.e. its actions on types and terms are surjective.
      It remains to check that they are injective.
      Take two types $A,B$ of $\coim F$ over a same base object such that $G\ A = G\ B$.
      Since $\coim F$ is computed pointwise, we have two types $A_{0},B_{0}$ of $\CC$ such that $\quot\ A_{0} = A$ and $\quot\ B_{0} = B$, and we can assume that they lie over the same base object of $\CC$.
      As $F\ A_{0} = G\ (\quot\ A_{0})$ and $F\ B_{0} = G\ (\quot\ B_{0})$, we have $F\ A_{0} = F\ B_{0}$, i.e. $(A \sim B) \in \ker F$, and thus there is an equality $\quot\ A_{0} = \quot\ B_{0}$.
      This shows that $G$ is injective on types.
      The same argument also shows that $G$ is injective on terms.
      Thus $G : \coim F \to \CD$ is a contextual isomorphism, and $\coim F \to \cxlim F$ is an isomorphism.
    \item[($\La$)]
      Assume that $\ker F$ is a fibrant congruence and that $\coim F \to \cxlim F$ is an isomorphism. $F$ is the composition of $\CC \to \coim F$ which is a strong contextual equivalence by \cref{prop:fib_cong_quotient}, $\coim F \to \cxlim F$ which is an isomorphism and $\cxlim F \to \CD$ which is a contextual isomorphism by definition.
      Since strong contextual equivalences are closed under composition, $F$ is a strong contextual equivalence.
  \end{description}
\end{proof}

\section{Equivalences between type theories}\label{sec:equivalence_models}

In this section, we discuss the notion of Morita equivalence between a weak type theory $\Th_{w}$ and a strong type theory $\Th_{s}$.
They have been introduced as the weak equivalences of a model structure on a category of type theories in \cite{IsaevMorita}.
While \cite{IsaevMorita} considers arbitrary morphisms between type theories, we only consider extensions of type theories by additional strict equalities.

\subsection{Equational extensions}

We fix a type theory signature $\Th_{w}$ over the theory of cumulative CwFs with universes and weak identity types.
\begin{defi}
  A \defemph{marked equation} over $\Th_{w}$ consists of a finitely generated cellular model $\Init_{w}[X]$, along with a closed internal equality $p : \abs{\Tm_{\Init_{w}[X]}}_{\diamond}\ (\Id\ \{A\}\ a\ b)$ of $\Init_{w}[X]$.

  The equation is said to hold strictly in a model $\CC$ of $\Th_{w}$ if for every object $\Gamma : \abs{\CC}$ and morphism $F : \Init_{w}[X] \to (\CC \sslash \Gamma)$, $F$ maps $p$ to the reflexivity equality.

  An \defemph{equational extension} of $\Th_{w}$ is a family of marked equations over $\Th_{w}$.
  \defiEnd
\end{defi}

If we were to compare types up to equivalence instead of internal equality of codes, the definition of marked equation would need to be extended to also include marked type equivalences.

We give some examples of equational extensions.
\begin{exas} \hfill
  \begin{enumerate}
    \item For the extensions from weak computation rules to strict computation rules, we mark the computation rules that should be made strict.
      For example, in the case of identity types, we mark the family of internal equalities $\J_{\beta}\ P\ d : \Id\ (\J\ P\ d\ \refl)\ d$.
      In the case of $\Pi$-types, we mark the internal equalities $\app_{\beta} : \Id\ (\app\ (\lam\ b)\ a)\ (b\ a)$ and $\funext_{\beta} : \Id\ (\funext\ f\ f\ (\lam\ (a \mapsto \refl)))\ \refl$ (and perhaps
      $\funext\text{-}\app_{\beta}$ as well).
    \item When considering the extension from inductive natural numbers to natural numbers with a strictly associative addition, we proceed in in two steps.
      First we extend the base theory by adding
      \[ \mathsf{plus} : \Tm\ \Nat \to \Tm\ \Nat \to \Tm\ \Nat \]
      as a new primitive operation, along with some of the internal equalities that it satisfies, such as
      \[ \mathsf{plus}_{0} : \{x\} \to \Tm\ (\Id\ (\mathsf{plus}\ 0\ x)\ x), \]
      \[ \mathsf{plus}_{1} : \{x\} \to \Tm\ (\Id\ (\mathsf{plus}\ x\ 0)\ x), \]
      \[ \mathsf{plus}_{2} : \{x\} \to \Tm\ (\Id\ (\mathsf{plus}\ (\mathsf{plus}\ x\ y)\ z)\ (\mathsf{plus}\ x\ (\mathsf{plus}\ y\ z))), \] etc.
      The operation $\mathsf{plus}$ is homotopic to the usual inductively defined addition, but not strictly equal to it.
      This kind of extension is conservativive.
      The weak type theory $\Th_{w}$ is then this extended theory.

      As a second step, we consider the equational extension of that theory obtained by marking the equalities $\mathsf{plus}_{0}$, $\mathsf{plus}_{1}$, $\mathsf{plus}_{2}$, etc.
      Thus the strong type theory $\Th_{s}$ includes the strict equalities $\mathsf{plus}\ 0\ x = x$, $\mathsf{plus}\ x\ 0 = x$, $\mathsf{plus}\ (\mathsf{plus}\ x\ y)\ z = \mathsf{plus}\ x\ (\mathsf{plus}\ y\ z)$, etc.
      It also includes the strict equalities $\mathsf{plus}_{0} = \refl$, $\mathsf{plus}_{1} = \refl$, $\mathsf{plus}_{2} = \refl$, etc.

    \item To consider the extension of a type theory with a new universe of strict propositions, we would also perform two steps.
      As a first step, we introduce a new constant type $\mathsf{SProp}$, along with an equality in $\Id\ \mathsf{SProp}\ \mathsf{HProp}$ with the universe $\mathsf{HProp}$ of propositions.
      We write $F : \mathsf{SProp} \to \mathsf{HProp}$ for the associated transport function.

      Secondly, we mark the family of equations
      \[ (A : \Tm\ \mathsf{SProp}) (a , b : \Tm\ (F\ A)) \to \Tm\ (\Id\ a\ b). \]

      In the resulting strong type theory, the only way to obtain closed elements of $\mathsf{SProp}$ is to use the inverse of the equivalence $F$ to replace elements of $\mathsf{HProp}$ by elements in $\mathsf{SProp}$.

      Note that the equational extension that marks instead the family of equations
      \[ (A : \Tm\ \mathsf{HProp}) (a , b : \Tm\ A) \to \Tm\ (\Id\ a\ b) \]
      is not a conservative extension in the absence of UIP.
      Indeed, as remarked in \cite{SProp}, if all propositions are strict propositions, then UIP holds.

    \item As a last example, we can also mark the family of all equalities
      \[ (A : \Ty)\ (x, y : \Tm\ A)\ (p : \Tm\ (\Id\ \{A\}\ x\ y)) \mapsto p. \]
      The corresponding strong type theory then includes the equality reflection rule.
  \end{enumerate}
\end{exas}

\subsection{Equivalences of theories}

We now work with a fixed choice of weak type theory $\Th_{w}$ and equational extension $\Th_{e}$.
The strong type theory $\Th_{s}$ is then defined as the extension of $\Th_{w}$ by the strict equalities $x = y$ and $p = \refl$ for every internal equality $p : \Id\ x\ y$ marked in $\Th_e$.

We have an adjunction between the categories $\CMod_{w}$ of models of $\Th_{w}$ and $\CMod_{s}$ of models of $\Th_{s}$.
\begin{center}\(\begin{tikzcd}
    \CMod_{w} \ar[r, swap, bend right, "L_{s}"] \ar[r, phantom, "\top"] & \CMod_{s} \ar[l, swap, bend right, "R_{s}"]
  \end{tikzcd}\)\end{center}
As $\Th_{s}$ is an equational extension of $\Th_{w}$, the functor $R_{s} : \CMod_{s} \to \CMod_{w}$ is simply the fully faithful forgetful functor that forgets that a strong model satisfies the additional equations. We will often omit $R_{s}$, and simply see any object of $\CMod_{s}$ as an object of $\CMod_{w}$. The left adjoint $L_{s} : \CMod_{w} \to \CMod_{s}$ can be shown to exist by various methods. One possibility is to use the adjoint functor theorem, using the fact that $\CMod_{w}$ and $\CMod_{s}$ are locally finitely presentable and that $R_{s}$ preserves limits.

The left adjoint can also be computed from the presentation of a model $\CC : \CMod_{w}$ by generators and relations. Such a presentation can be obtained from the cellular replacement of the model $\CC$ by some cellular model $\Init_{w}[X]$. Since left adjoints preserve colimits and cellular models are built by iterated pushouts, $L_{s}\ \Init_{w}[X] \simeq \Init_{s}[X]$, where $\Init_{s}[X]$ is the cellular strong model with the same generators as $\Init_{w}[X]$. Since $\CC$ is a quotient of $\Init_{w}[X]$ by a fibrant congruence and left adjoints preserve quotients, $L_{s}\ \CC$ is also the quotient of $\Init_{s}[X]$ by some congruence, although that congruence may fail to be fibrant in general.

We write $\eta^{X} : \Init_{w}[X] \to \Init_{s}[X]$ for the unit of this adjunction at a cellular model $\Init_{w}[X]$.

\begin{defi}\label{def:th_weq}
  We say that $\Th_{w}$ and $\Th_{s}$ are Morita equivalent if for every cofibrant contextual model $\CC : \CMod_{w}^{\cof}$ of $\Th_{w}$, the unit $\eta_{\CC} : \CMod_{w}(\CC \to R_{s}\ (L_{s}\ \CC))$ is a
  weak contextual equivalence.
  \defiEnd
\end{defi}

It is shown in \cite{IsaevMorita} that whenever $\Th_{w}$ is semi-model, then is weakly equivalent to $\Th_{s}$ if and only if $\Th_{s}$ is also semi-model and the adjunction $(L_{s} \dashv R_{s})$ is a Quillen equivalence.

We now show that in order to prove that $\Th_{w}$ and $\Th_{s}$ are equivalent, it is sufficient to look at the cellular models of $\Th_{w}$.
Recall that the cellular models of $\Th_{w}$ are very similar to the initial model of $\Th_{w}$.
Thus, for most type theories, whenever we can prove that the initial models of $\Th_{w}$ and $\Th_{s}$ are equivalent, we can expect the same methods to work for arbitrary cellular models, implying that $\Th_{w}$ and $\Th_{s}$ are equivalent.
\begin{prop}\label{prop:weq_theories_char}
  The following conditions are equivalent:
  \begin{enumerate}
    \item\label{item:weq_theories_char_1} The theories $\Th_{w}$ and $\Th_{s}$ are Morita equivalent
    \item\label{item:weq_theories_char_2} The condition of \cref{def:th_weq} holds for every cellular model of $\Th_{w}$.
    \item\label{item:weq_theories_char_3} The condition of \cref{def:th_weq} holds for every finite cellular model of $\Th_{w}$.
  \end{enumerate}
\end{prop}
\begin{proof} The forward implications trivially hold. We show the reverse implications.
   \begin{description}
     \item[(\ref{item:weq_theories_char_2} $\Ra$ \ref{item:weq_theories_char_1})]
       Take a cofibrant contextual model $\CC$. It is the retract of some cellular model $\Init_{w}[X]$. Then $L_{s}\ \CC$ is also a retract of $L_{s}\ \Init_{w}[X]$, and furthermore $\eta_{\CC} : \CC \to L_{s}\ \CC$ is a retract of $\eta^{X} : \Init_{w}[X] \to L_{s}\ \Init_{w}[X]$.
       Since weak contextual equivalences are closed under retracts and $\eta^{X}$ is a weak contextual equivalence by assumption, $\eta_{\CC} : \CC \to L_{s}\ \CC$ is also a weak contextual equivalence.
     \item[(\ref{item:weq_theories_char_3} $\Ra$ \ref{item:weq_theories_char_2})]
       For this we rely on some well-known properties of locally finitely presentable categories and freely generated models that we do not prove in this paper, since the proofs are quite lengthy, and not required for the main results of this paper. The idea is that since we consider finitary type theories, any type or term of a freely generated model is supported by a finite subset of generators. \\
       Let $\Init_{w}[X]$ be a cellular model of $\Th_{w}$. We know that $\Init_{w}[X]$ is the filtered colimit $\colim\limits_{Y \rat X} \Init_{w}[Y]$ of its finite cellular subextensions. Since left adjoints preserve colimits, $\Init_{s}[X]$ can be computed as the filtered colimit $\colim\limits_{Y \rat X} \Init_{s}[Y]$. For every type $A$ or term $a$ of $\Init_{s}[X]$, there merely exists a finite cellular subextension $Y \rat X$ such that the type or term already exists in $\Init_{s}[Y]$. Using condition \ref{item:weq_theories_char_3}, we can then compute a lift of $A$ or $a$ in $\Init_{w}[Y]$.
   \end{description}
\end{proof}

\section{Coherence for strict type theories}\label{sec:conservativity_strict}

In this section we specialize the relationship between strong contextual equivalences and fibrant congruences to the setting of equational extensions of theories.
As a byproduct, we obtain a decomposition of Hofmann's proof of the conservativity of extensional type theories over type theories satisfying the UIP principle.
We assume given a weak type theory $\Th_{w}$ and a strong type theory $\Th_{s}$ extending $\Th_{w}$ by a family of equations $\Th_{e}$.

\begin{defi}
  We say that a contextual congruence $\widetilde{\CC}$ over a model $\CC$ of $\Th_{w}$ includes the marked equations of the equational extension $\Th_{e}$ if, for every finite cellular model $\Init_{w}[X]$, marked equation $p : \abs{\Tm_{\Init_{w}[X]}}\ (\Id\ a\ b)$, object $\Gamma : \abs{\CC}$ and morphism $F : \Init_{w}[X] \to (\CC \sslash \Gamma)$, we have $(F\ a \sim F\ b) \in \widetilde{\CC}$ and $(F\ p \sim \refl\ \{F\ a\}) \in \widetilde{\CC}$.
 
  For example, for the extension from weak identity types to strong identity types, this says that $(\J\ P\ d\ \refl \sim d) \in \widetilde{\CC}$ and $(\J_{\beta}\ P\ d \sim \refl\ \{d\}) \in \widetilde{\CC}$ for all relevant arguments.
  \defiEnd
\end{defi}

\begin{lem}\label{lem:congruence_s}
  Let $\CC$ be a contextual model of $\Th_{w}$.
  Assume given a fibrant contextual congruence $\widetilde{\CC}$ over $\CC$, that is compatible with $\Th_{w}$ and includes the marked equations of the equational extension $\Th_{e}$.
  Then the quotient $\Quot_{\widetilde{\CC}}$ is a model of $\Th_{s}$.
\end{lem}
\begin{proof}
  Since the congruence $\widetilde{\CC}$ is fibrant, we can form the quotient $\Quot_{\widetilde{\CC}}$, and we know that the quotienting map $\quot_{\widetilde{\CC}} : \CC \to \Quot_{\widetilde{\CC}}$ is a strong contextual equivalence. Because $\widetilde{\CC}$ is compatible with $\Th_{w}$, the quotient $\Quot_{\widetilde{\CC}}$ is a model of $\Th_{w}$, and $\quot_{\widetilde{\CC}} : \CC \to \Quot_{\widetilde{\CC}}$ is a morphism of models of $\Th_{w}$.

  To show that $\Quot_{\widetilde{\CC}}$ is a model of $\Th_{s}$, it suffices to check that it satisfies all of the necessary equations.

  Take a cellular model $\Init_{w}[X]$ and a marked equation $p : \abs{\Tm_{\Init_{w}[X]}}\ (\Id\ a\ b)$.
  We need to check that for every morphism $F : \Init_{w}[X] \to (\Quot_{\widetilde{\CC}} \sslash \Gamma)$, $F$ maps $p$ to the reflexivity equality.

  Take such a morphism $F$.
  Since $\CC$ is contextual and $\quot_{\widetilde{\CC}}$ is a strong contextual equivalence, $\quot_{\widetilde{\CC}} : \CC \to \Quot_{\widetilde{\CC}}$ is surjective on contexts. Therefore we have some $\Gamma_0 : \abs{\CC}$ such that $\quot_{\widetilde{\CC}}\  \Gamma_{0} = \Gamma$.
  Note that the morphism $\quot_{\widetilde{\CC}} : \CC \to \Quot_{\widetilde{\CC}}$ can be restricted to $\quot_{\widetilde{\CC}} : (\CC \sslash \Gamma_{0}) \to  (\Quot_{\widetilde{\CC}} \sslash \Gamma)$.

  We will now construct a morphism $F_{0} : \Init_{w}[X] \to (\CC \sslash \Gamma_{0})$ such that $F_{0} \cdot \quot_{\widetilde{\CC}} = F$.
  The universal property of $\Init_{w}[X]$ says that $F$ is determined by the images of the generating terms of $X$.
  To construct $F_{0}$, we just have to pick a lift along $\quot_{\widetilde{\CC}}$ of the images of these generating terms.
  This is possible since $\quot_{\widetilde{\CC}}$ is a strong contextual equivalence.

  By hypothesis, $(F_{0}\ a \sim F_{0}\ b) \in \widetilde{\CC}$ and $(F_{0}\ p \sim \refl) \in \widetilde{\CC}$. Therefore, $\quot_{\widetilde{\CC}}\ (F_{0}\ a) = \quot_{\widetilde{\CC}}\ (F_{0}\ b)$ and $\quot_{\widetilde{\CC}}\ (F_{0}\ p) = \refl$, as needed.
 
  Thus all marked equations hold strictly in $\Quot_{\widetilde{\CC}}$, which is therefore a model of $\Th_{s}$.
\end{proof}

\begin{lem}\label{lem:congruence_seq}
  Let $\Init_{w}[X]$ be a cellular model of $\Th_{w}$.
  Assume that there exists a congruence $\widetilde{\Init_{w}[X]}$ over $\Init_{w}[X]$ satisfying the conditions of \cref{lem:congruence_s} and that is additionally included in the kernel $\ker \eta^{X}$, i.e. any types or terms that are congruent in $\widetilde{\Init_{w}[X]}$ are identified by $\eta^{X}$.
  Then the morphism $\eta^{X} : \Init_{w}[X] \to \Init_{s}[X]$ is a strong contextual equivalence.
\end{lem}
\begin{proof}
  The inclusion $\widetilde{\Init_{w}[X]} \subseteq \ker \eta^{X}$ implies, by the universal property of the quotient $\Quot_{\widetilde{\Init_{w}[X]}}$, that $\eta^{X}$ factors through $\quot_{\widetilde{\Init_{w}[X]}}$; we have $r : \Quot_{\widetilde{\Init_{w}[X]}} \to \Init_{s}[X]$ such that $\eta^{X} = \quot_{\widetilde{\Init_{w}[X]}} \cdot r$.
  \Cref{lem:congruence_s} says that $\Quot_{\widetilde{\Init_{w}[X]}}$ is a model of $\Th_{s}$.
  The universality of the arrow $\eta^{X} : \Init_{w}[X] \to \Init_{s}[X]$ then provides a section $s : \Init_{s}[X] \to \Quot_{\widetilde{\Init_{w}[X]}}$ of $r$.
  By the universal property of $\Init_{w}[X]$, we also have that $\eta^{X} \cdot s = \quot_{\widetilde{\Init_{w}[X]}}$.

  \begin{center}\(\begin{tikzcd}
      & \Quot_{\widetilde{\Init_{w}[X]}} \ar[rd, shift right, swap, "r"] & \\
      \Init_{w}[X] \ar[ru, "\quot_{\widetilde{\Init_{w}[X]}}"] \ar[rr, swap, "\eta^{X}"] && \Init_{s}[X] \ar[lu, shift right, swap, "s"]
    \end{tikzcd}\)\end{center}
  We can now see that $\eta^{\CC}$ is a retract of $\quot_{\widetilde{\CC}}$: the following diagram commutes.
  \begin{center}\(\begin{tikzcd}
      \Init_{w}[X] \ar[r, equal] \ar[d, "\eta^{X}"] &
      \Init_{w}[X] \ar[r, equal] \ar[d, "\quot_{\widetilde{\Init_{w}[X]}}"] &
      \Init_{w}[X] \ar[d, "\eta^{X}"] \\
      \Init_{s}[X] \ar[r, "s"] &
      \Quot_{\widetilde{\Init_{w}[X]}} \ar[r, "r"] &
      \Init_{s}[X]
    \end{tikzcd}\)\end{center}
  The left square of that diagram commutes thanks to the universal property of $\Init_{w}[X]$.

  Since $\quot_{\widetilde{\CC}}$ is a strong contextual equivalence and strong contextual equivalences are closed under retracts, $\eta_{\CC}$ is also a strong contextual equivalence.
\end{proof}

\begin{thm}\label{thm:conservativity_strict}
  Let $\Th_{w}$ be a type theory over the theory of cumulative CwFs with universes and weak identity types that includes the UIP principle and let $\Th_{s}$ be the extension of $\Th_{w}$ with the equality reflection rule.

  If either of the following two conditions holds, then the theories $\Th_{w}$ and $\Th_{s}$ are Morita equivalent.
  \begin{enumerate}
    \item The theory $\Th_{w}$ includes $\Pi$-types with a strict $\beta$-rule (and function extensionality).
    \item The theory $\Th_{w}$ is semi-model.
  \end{enumerate}
\end{thm}
\begin{proof}
  We have to show that for every cellular model $\Init_{w}[X]$ of $\Th_{w}$, the morphism $\eta^{X} : \Init_{w}[X] \to \Init_{s}[X]$ is a strong contextual equivalence.
  We will do so using \cref{lem:congruence_seq}.

  We define a congruence $\widetilde{\Init_{w}[X]}$ over $\Init_{w}[X]$.
  \begin{itemize}
    \item Two types $A,B : \abs{\Ty_{w}}_{\Gamma}$ are congruent if there exists some equality $p : \abs{\Tm_{w}}\ (\Id\ \{\UU\}\ A\ B)$.
    \item Two terms $a : \abs{\Tm_{w}}_{\Gamma}\ A$ and $b : \abs{\Tm_{w}}_{\Gamma}\ B$ are congruent if there exists some equality $p : \abs{\Tm_{w}}\ (\Id\ \{\UU\}\ A\ B)$ along with some equality $q : \abs{\Tm_{w}}\ (\Id\ \{B\}\ (p^{\star}\ a)\ b)$.
      Since $\Th_{w}$ includes UIP, the choice of $p$ is irrelevant.
  \end{itemize}
  The reflexivity, symmetry and transitivity properties are easily seen to hold.

  We can also check that $\widetilde{\Init_{w}[X]}$ is fibrant.
  Indeed, take any two congruent types $(A \sim B) \in \widetilde{\Init_{w}[X]}$ and a term $a : \Tm\ A$ of type $A$.
  Then we have an equality $p$ between $A$ and $B$, and we obtain a term $p^{\star}\ a : \Tm\ B$ such that $(a \sim p^{\star}\ a) \in \widetilde{\Init_{w}[X]}$.

  We still have to check the actions of dependent types and terms on the relations as well as the compatibility with the operations of $\Th_{w}$.
  All operations can be dealt with uniformly; we will only look at the $\Id$ and $\Pi$ type formers.

  In the case of the identity type former $\Id$, we have types $A_{1},A_{2} : \abs{\Ty_{w}}_{\Gamma}$, and terms $x_{1},y_{1} : \abs{\Tm_{w}}_{\Gamma}\ A_{1}$ and $x_{2},y_{2} : \abs{\Tm_{w}}_{\Gamma}\ A_{2}$ such that $(A_{1} \sim A_{2})$, $(x_{1} \sim x_{2})$ and $(y_{1} \sim y_{2})$, and we need to prove that $(\Id\ \{A_{1}\}\ x_{1}\ y_{1} \sim \Id\ \{A_{2}\}\ x_{2}\ y_{2})$.

  By definition of the relations of $\widetilde{\Init_{w}[X]}$, we have internal equalities $p : \abs{\Tm_{w}}_{\Gamma}\ (\Id\ A_{1}\ A_{2})$,
  $q_{x} : \abs{\Tm_{w}}_{\Gamma}\ (\Id\ (p^{\star}\ x_{1})\ x_{2})$ and
  $q_{y} : \abs{\Tm_{w}}_{\Gamma}\ (\Id\ (p^{\star}\ y_{1})\ y_{2})$.
  Here we have to use UIP to ensure that $q_{x}$ and $q_{y}$ both lie over the same type equality $p$.

  We can now see $\Id$ as an operation from the $\Sigma$-type $(A : \UU) \times (x : A) \times (y : A)$ to $\UU$.
  We have not assumed that $\Th_{w}$ has $\Sigma$-types, but we can use telescopes and the results of \cref{ssec:lift_tele} instead.
  Then from $p$, $q_{x}$ and $q_{y}$ we obtain an equality between $(A_{1}, x_{1}, y_{1})$ and $(A_{2}, x_{2}, y_{2})$ in the $\Sigma$-type (or telescope) $(A : \UU) \times (x : A) \times (y : A)$.
  The action on equalities of $\Id$ then provides an equality between $\Id\ \{A_{1}\}\ x_{1}\ y_{1}$ and $\Id\ \{A_{2}\}\ x_{2}\ y_{2}$, as needed.

  Let's also look at an operation with a higher-order argument: the $\Pi$-type former.
  In that case, we have types $A_{1}, A_{2} : \abs{\Ty_{w}}_{\Gamma}$ and dependent types $B_{1} : \abs{\Ty_{w}}_{\Gamma, (x : A_{1})}$ and $B_{2} : \abs{\Ty_{w}}_{\Gamma, (x : A_{2})}$, such that $(A_{1} \sim A_{2})$ and $(B_{1}\ a_{1} \sim B_{2}\ a_{2})$ for every pair $(a_{1} \sim a_{2})$.
  This means that we can find a type equality $p : \Id\ A_{1}\ A_{2}$ and a dependent type equality $q : (a : A_{1}) \to \Id\ (B_{1}\ a)\ (B_{2}\ (p^{\star}\ a))$.
  We need to construct a type equality between $\Pi\ A_{1}\ B_{1}$ and $\Pi\ A_{2}\ B_{2}$.

  There are then two cases.
  \begin{enumerate}
    \item If $\Th_{w}$ has $\Pi$-types with a strict $\beta$-rule, then we can view the $\Pi$-type former as an operation from the type $(A : \UU) \times (B : A \to \UU)$ to $\UU$.
      We can then conclude as in the case of identity types above.
    \item If the theory $\Th_{w}$ is semi-model, then we view the operation $\Pi$ as the type $\Pi\ \bm{A}\ \bm{B}$ of the cellular model
      \[ \CC \triangleq \Init_{w}[ \bm{A} : \UU, (a : A) \vdash \bm{B}(a) : \UU]. \]
     
      We then consider the cellular model
      \[ \CD \triangleq \Init_{w}\lbrack
      \begin{array}[t]{l} \bm{A_{1}} : \UU, \bm{A_{2}} : \UU, \bm{p} : \Id\ \bm{A_{1}}\ \bm{A_{2}}, \\
        (a : \bm{A_{1}}) \vdash \bm{B_{1}}(a) : \UU, \\
        (a : \bm{A_{2}}) \vdash \bm{B_{2}}(a) : \UU, \\
        (a : \bm{A_{1}}) \vdash \bm{q}(a) : \Id\ \bm{B_{1}}(a)\ \bm{B_{2}}(\bm{p}^{\star}\ a)\ \rbrack.
      \end{array} \]
     
      Our assumptions imply that there is a morphism $F : \CD \to (\Init_{w}[X] \sslash \Gamma)$ that sends $\bm{A_{1}}$ to $A_{1}$, $\bm{A_2}$ to $A_2$, etc.

      Note that we also have two maps $i_{1}, i_{2} : \CC \to \CD$, sending $(\bm{A}, \bm{B})$ respectively to $(\bm{A_{1}}, \bm{B_{1}})$ and $(\bm{A_{2}}, \bm{B_{2}})$.
      The map $i_{1}$ is a composition of two basic $J$-cellular extensions. The first of these two extensions adds $\bm{A_{2}}$ and $\bm{p}$ while the second adds $\bm{B_{2}}$ and $\bm{q}$.
      Since $\Th_{w}$ is semi-model and $\CC$ is contextual and cofibrant, the map $i_{1}$ is a weak equivalence.

      The maps $i_{1}$ and $i_{2}$ also admit a common retraction $r : \CD \to \CC$, which sends $\bm{A_{1}}$ and $\bm{A_{2}}$ to $\bm{A}$, $\bm{B_{1}}$ and $\bm{B_{2}}$ to $\bm{B}$, $\bm{p}$ to the reflexivity equality, and $\bm{q}$ to some proof of $\Id\ \bm{B}(a)\ \bm{B}(\refl^{\star}\ a)$. By $2$-out-of-$3$, the map $r : \CD \to \CC$ is also a weak contextual equivalence.

      Therefore we can lift the reflexivity equality $\refl : \Id\ (\Pi\ \bm{A}\ \bm{B})\ (\Pi\ \bm{A}\ \bm{B})$ from $\CC$ to $\CD$ in order to obtain an equality $\alpha : \Id\ (\Pi\ \bm{A_{1}}\ \bm{B_{1}})\ (\Pi\ \bm{A_{2}}\ \bm{B_{2}})$ in $\CD$.

      Applying the morphism $F : \CD \to (\Init_{w}[X] \sslash \Gamma)$, we obtain an equality between $\Pi\ A_{1}\ B_{1}$ and $\Pi\ A_{2}\ B_{2}$, as needed.

      The reader familiar with the theory of model categories will have noticed that this proof uses the fact that $\CD$ is a cylinder object for $\CC$.
      This method generalizes to arbitrary type-theoretic operations, replacing $\CC$ by a cellular model encoding the premises of the operation and $\CD$ by a suitable cylinder object for $\CC$.
  \end{enumerate}

  We can now conclude the proof.
  The congruence $\widetilde{\Init_{w}[X]}$ is included in $\ker \eta^{X}$, essentially by definition.
  It satisfies all of the conditions of \cref{lem:congruence_s} and \cref{lem:congruence_seq}, and $\eta^{X} : \Init_{w}[X] \to \Init_{s}[X]$ is thus a strong contextual equivalence.
  As this holds for all cellular models of $\Th_{w}$, the theories $\Th_{w}$ and $\Th_{s}$ are Morita equivalent
\end{proof}

\begin{rem}
  \Cref{thm:conservativity_strict} is actually not a generalization of Hofmann's conservativity theorem.
  Indeed, the type theories considered by Hofmann did not include a hierarchy of universes.
  The presence of universes makes the proof simpler and more uniform, since we can use the same relations on types and terms: internal equality.

  In the absence of universes, we have to use another equivalence relation on types.
  One solution is to use local universes~\cite{LocalUniverses}.
  A local universe in a CwF $\CC$ is a pair $(V, E)$ where $V$ is a closed type of $\CC$ and $E$ is a dependent type over $V$.
  A type $A : \abs{\Ty_{\CC}}_{\Gamma}$ is classified by a local universe $(V, E)$ if there is a term $\chi : \abs{\Tm_{\CC}}_{\Gamma}\ V$ such that $A = E[\chi]$.
  For many type theories, including the type theories considered by Hofmann, it is possible to show (for the cellular models) that every type $A$ is classified by some local universe $(V_{A}, E_{A})$, as witnessed by a term $\chi_{A}$.
  For example, in presence of $\Pi$-types and $\Sigma$-types, the type $\Pi\ (x : \Nat)\ (\Id\ x\ x)$ is classified by the local universe $(V, E)$, with
  \begin{alignat*}{2}
    & V && \triangleq (\Nat \to \Nat) \times (\Nat \to \Nat) \\
    & E(f , g) && \triangleq \Pi\ (x : \Nat)\ (\Id\ (f\ x)\ (g\ x)).
  \end{alignat*}
  We can then define a suitable congruence by saying that two types $A, B$ are related if they have the same local universe $(V_{A}, E_{A}) = (V_{B}, E_{B})$, and their classifying terms $\chi_{A}$ and $\chi_{B}$ are internally equal in $V_{A}$.
\end{rem}

\section{Type-theoretic higher congruences and coherence for non-strict type theories}\label{sec:conservativity_nonstrict}

In the previous section, we have seen that for type theories with the UIP principle, conservativity and coherence theorems can be proven by constructing some fibrant congruences on the cellular models of the weak type theory.
We now generalize this to type theories without UIP.
The core idea is to use a suitable notion of higher congruence instead of fibrant congruences.
While ordinary congruences can be seen as models valued in setoids, higher congruences should use some notion of weak $\infty$-groupoid.
There could be many possible ways to define higher congruences, based on different definitions of weak $\infty$-groupoids.
What seems to work best is to use a definition that is as close to type theory as possible, inspired by Brunerie's type-theoretic definition of weak $\infty$-groupoid \cite[Appendix A]{BrunerieThesis}.

The higher congruences over models of $\Th_{w}$ are described by a type theory $\Th_{w,2}$ extending the weak type theory $\Th_{w}$. This new type theory is a two-level type theory, in the sense that it has an inner layer and and outer layer.
Two-level type theories~\cite{2LTT} have been introduced to have a setting in which an inner theory with a non-strict equality and perhaps univalent universes, and an outer theory with strict identity types, can interact.
The two-level type theory that we consider is more minimal, with less structures available in the outer layer. Crucially, we don't assume that the outer layer has UIP.
Instead the fact that a model of $\Th_{w,2}$ validates UIP will be our definition of acyclicicy for higher congruences.

\subsection{Brunerie weak \texorpdfstring{$\infty$}{infinity}-groupoids}

Before defining our notion of $\infty$-congruence, we give a type-theoretic definition of weak $\infty$-groupoids that should be more or less equivalent to Brunerie's definition~\cite[Appendix A]{BrunerieThesis}.
This definition won't be used outside of this subsection, but it should provide intuition for the definition of $\infty$-congruence.

In Brunerie's definition of $\infty$-groupoids, an $\infty$-groupoid consists of a globular set $G$, along with structure given by interpretations of all coherence laws definable in some type theory.
The additional structure on globular sets is \emph{structure} in the categorical sense, which means that the forgetful functor from Brunerie $\infty$-groupoids to globular sets is faithful.
The definition of this structure is however quite involved.

Instead, we define two notions. The first notion is a notion of generalized $\infty$-groupoid, which is easy to define, but does not consist of only structure over a globular set. Secondly, we define a notion of reduced $\infty$-groupoid by identifying a subcategory of the generalized $\infty$-groupoids for which the additional data only consists of structure over the globular set.

We consider the theory $\Th_{\Id}$ of weak identity types.
A model of $\Th_{\Id}$ is a category $\CC$ with a terminal object, along with a family $(\Ty, \Tm)$ of types and terms, equipped with weak identity types.
Only the singletons are required to be representable, which means that only the context extensions of the form $\Gamma, (y : A, p : \Id\ x\ y)$ are guaranteed to exist in $\CC$.

We say that a \defemph{generalized $\infty$-groupoid} consists of a model $\CC$ of $\Th_{\Id}$ along with a closed type $\star$ of $\CC$.
Given a generalized $\infty$-groupoid $(\CC, \star)$, we have a set $\abs{\Tm_{\CC}}_{\diamond}\ \star$ of points, sets $\abs{\Tm_{\CC}}_{\diamond}\ (\Id\ \{\star\}\ -\ -)$ of $1$-cells, etc, generating a globular set.
The eliminator for the identity type then provides all of the operations of an $\infty$-groupoid.

For example, consider the category $\mathbf{sSet}$ of simplicial sets.
It is a model of $\Th_{\Id}$, and a closed type of $\mathbf{sSet}$ is a Kan complex.
Thus, given any Kan complex $X$, we have a generalized $\infty$-groupoid $(\mathbf{sSet}, X)$.

The circle $\mathbb{S}^{1}$ can be defined as the model of $\Th_{\Id}$ freely generated by a closed type $\star$, a point $\mathsf{base} : \star$ and a path $\mathsf{loop} : \Id\ \mathsf{base}\ \mathsf{base}$.

We now define \defemph{reduced $\infty$-groupoids}. Take a generalized $\infty$-groupoid $(\CC, \star)$.
We can regenerate freely its contexts, types and non-closed terms, starting from the empty context and the closed type $\star$, to obtain another generalized $\infty$-groupoid $(\CC', \star)$ along with a morphism $(\CC', \star) \to (\CC, \star)$ that is bijective on closed terms.
We say that $(\CC, \star)$ is reduced if that map is an isomorphism.
This is analogous to our definition of contextuality in \cref{ssec:contextual_models}, and can also be formalized using an orthogonal factorization system on the category of generalized $\infty$-groupoids, whose right maps are the morphisms that are bijective on closed terms.

It seems that the notion of generalized $\infty$-groupoid provides extra generality that is often useful.
For example, given a model $\CC$ of $\Th_{\Id}$ and two closed types $A, B$ of $\CC$, it is usually easier to compare the generalized $\infty$-groupoids $(\CC, A)$ and $(\CC, A)$ internally to $\CC$ instead of comparing their reduced variants $(\CC', A)$ and $(\CC'', B)$.

\subsection{Type-theoretic higher congruences}

We assume given a weak type theory $\Th_{w}$ and a strong type theory $\Th_{s}$ extending $\Th_{w}$ by a family of equations $\Th_{e}$.
\begin{defi}
  We define a type theory $\Th_{w,2}$ (the indices ${-}_{w,2}$ stand for (weak, two-level)), extending the weak type theory $\Th_{w}$ by:
  \begin{enumerate}
    \item An outer family $(\Ty^{o},\Tm^{o})$.
      Its components are annotated by the superscript ${}^{o}$.

      The outer family is not required to be representable, i.e. models of $\Th_{w,2}$ do not have to support context extensions by variables of outer types.
      The family $(\Ty,\Tm)$ corresponding to the theory $\Th_{w}$ is called the inner family.

      The outer family classifies the terms of the inner family: we have for every universe level $n$ a family $\tm_{n} : \Ty_{n} \to \Ty^{o}$ of codes for inner terms, with isomorphisms $\Tm_{n}\ A \simeq \Tm^{o}\ (\tm_{n}\ A)$.
      We will leave these isomorphisms implicit.

      Thanks to the universes of the inner theory, the types of the inner family are also classified by the outer family: we can pose $\ty_{n} \triangleq \tm_{n+1}\ \UU_{n}$, and we then have isomorphisms $\Ty_{n} \simeq \Tm^{o}\ \ty_{n}$.
    \item The outer family has weak identity types with \emph{representable singletons}.
      This consists of operations $\Id^{o}$, $\refl^{o}$, $\J^{o}$ and $\J_{\beta}^{o}$, as specified in \cref{def:weak_id_elim}.
    \item The outer family has $\Pi$-types with arities in the inner family, with a strict $\beta$-rule.
      This consists of a type former
      \begin{alignat*}{2}
        & \Pi^{o} && : (A : \Ty) (B : \Tm\ A \to \Ty^{o}) \to \Ty^{o} \rlap{\ ,}
      \end{alignat*}
      along with operations $\lam^{o}$, $\app^{o}$, $\funext^{o}$, $\funext_{\beta}^{o}$, $\funext\text{-}\app^{o}$ and $\funext\text{-}\app_{\beta}^{o}$ as specified in \cref{def:weak_pi} and \cref{def:strong_pi}.
      \defiEnd
  \end{enumerate}
\end{defi}

A model of $\Th_{w,2}$ has two representable sorts: the terms of the inner family and the singletons of the outer family.
This means that a context can be extended as normal by variables of arbitrary types of the inner family, but can only be extended by contractible pairs $(y : A,p : \Id^{o}\ \{A\}\ x\ y)$ of the outer family.

In any model of $\Th_{w,2}$, it is possible to turn any outer equality $p : \Tm^{o}\ (\Id^{o}\ \{\tm\ A\}\ x\ y)$ into an inner equality $[p] : \Tm\ (\Id\ \{A\}\ x\ y)$, obtained by transporting $\refl : \Tm\ (\Id\ x\ x)$ over $p : \Tm^{o}\ (\Id^{o}\ x\ y)$ in the family $(\Id\ x\ -)$.

\begin{defi}
  A \defemph{type-theoretic higher congruence} over a model $\CC : \CMod_{w}$ of $\Th_{w}$ is a model $\CD : \CMod_{w,2}$ of $\Th_{w,2}$ along with a weak contextual equivalence $\iota : \CMod_{w}(\CC \to \CD)$.
  \defiEnd
\end{defi}

The notion of weak contextual equivalence may now be slightly ambiguous, since the models of $\Th_{w,2}$ have two kinds of identity types (inner and outer).
We only consider the weak equivalences for the underlying models of $\Th_{w}$.

\begin{defi}
  A model $\CC$ of $\Th_{w,2}$ is said to be \defemph{acyclic} if the family of outer types $\tm : \Ty_{\CC} \to \Ty_{\CC}^{o}$ is $0$-truncated in $\CC$ (with respect to the outer identity types of $\CC$).
  This means that for every term $a : \abs{\Tm_{\CC}}_{\Gamma}\ A$ and loop $p : \abs{\Tm^{o}_{\CC}}_{\Gamma}\ (\Id^{o}\ a\ a)$, we have an inhabitant of $\abs{\Tm^{o}_{\CC}}_{\Gamma}\ (\Id^{o}\ p\ \refl^{o})$.
  We don't require these inhabitants to be stable under substitution.
  \defiEnd
\end{defi}

\begin{defi}\label{def:th_e}
  We say that a model $\CC$ of $\Th_{w,2}$ includes the marked equations of $\Th_{e}$ if, for every finite cellular model $\Init_{w}[X]$, marked equality $p : \abs{\Tm_{\Init_{w}[X]}}\ (\Id\ \{A\}\ a\ b)$, object $\Gamma : \abs{\CC}$ and morphism $F : \Init_{w}[X] \to (\CC \sslash \Gamma)$, we have outer equalities
  \[ \widehat{p} : \abs{\Tm^{o}_{\CC}}_{\Gamma}\ (\Id^{o}\ (F\ a)\ (F\ b)) \]
  and
  \[ \widetilde{p} : \abs{\Tm^{o}_{\CC}}_{\Gamma}\ (\Id^{o}\ p\ [\widehat{p}]), \]
  lifting the marked equations from equalities of the inner layer to equalities of the outer layer.
 
  For example, for the extension from weak identity types to strong identity types, we require, for all relevant arguments, outer terms
  \[ \widehat{\J_{\beta}}\ P\ d : \Tm^{o}\ (\Id^{o}\ (\J\ P\ d\ \refl)\ d), \] and
  \[ \widetilde{\J_{\beta}}\ P\ d : \Tm^{o}\ (\Id^{o}\ (\J_{\beta}\ P\ d)\ [\widehat{\J_{\beta}}\ P\ d]). \]
  The two terms $\widehat{\J_{\beta}}\ P\ d$ and $\widetilde{\J_{\beta}}\ P\ d$ can be seen as non-truncated variants of the equations $\J\ P\ d\ \refl = d$ and $\J_{\beta}\ P\ d = \refl$ of strong identity types.

  We denote by $\Th_{w,2,e}$ the extension of $\Th_{w,2}$ by these lifted equations.
  \defiEnd
\end{defi}

\begin{con}\label{con:cong_from_higher_cong}
  Let $\CC$ be an acyclic model of $\Th_{w,2}$.
  We construct a fibrant congruence $\widetilde{\CC}$ on $\CC$.

  Furthermore, if $\CC$ includes the marked equations of some equational extension $\Th_{e}$, then so does $\widetilde{\CC}$.
\end{con}
\begin{proof}[Construction]
  We work internally to the presheaf category $\CPsh\ \CC$.
  Using the results of \cref{ssec:lift_tele}, we can pretend that we have $\Sigma$-types in this construction.
  \begin{itemize}
    \item Two inner types $A, B : \Ty_{\CC}$ are congruent in $\widetilde{\CC}$ if there exists an outer equality between $A$ and $B$ in $\ty$.
    \item Two inner terms $a : \Tm_{\CC}\ A$ and $b : \Tm_{\CC}\ B$ are congruent in $\widetilde{\CC}$ if there exists an outer equality between $(A, a)$ and $(B, b)$ in $(X : \ty) \times \tm\ X$.
      Because $\CC$ is acyclic, given any two congruent terms $a : \Tm_{\CC}\ A$ and $b : \Tm_{\CC}\ B$ and an outer equality $p : \Id^{o}\ A\ B$, we can find an outer equality $q : \Id^{o}\ (p^{\star}\ a)\ b$ between $a$ and $b$ lying over $p$.
  \end{itemize}
  The fact that these relations are equivalence relations follows from the existence of reflexivity outer equality and inverses and compositions of outer equalities.

  We can also check the fibrancy of $\widetilde{\CC}$.
  Take two congruent inner types $A$ and $B$ and a term $a : \Tm\ A$.
  We have a path $p : \Id^{o}\ A\ B$ and can transport $a$ to $p^{\star}\ a : \Tm\ B$. Furthermore, $(a \sim p^{\star}\ a) \in \widetilde{\CC}$, as needed.

  It remains the compatibility of $\widetilde{\CC}$ with dependent types and terms, as well as with the operations of $\Th_{w}$.
  Since $\widetilde{\CC}$ is defined by the outer identity types of $\CC$, this follows from the fact that all operations of $\CC$ preserve the outer identity types.

  For instance, for the $\Id$-type former, we have $(A_{1} \sim A_{2}) \in \widetilde{\CC}$, $(x_{1} \sim x_{2}) \in \widetilde{\CC}$ and $(y_{1} \sim y_{2}) \in \widetilde{\CC}$.
  Thus we have outer equalities $p : \Id^{o}\ A_{1}\ A_{2}$, $q_{x} : \Id^{o}\ (\alpha^{\star}\ x_{1})\ x_{2}$, and $q_{y} : \Id^{o}\ (\alpha^{\star}\ y_{1})\ y_{2}$.
  Note that the acyclicity of $\CC$ is used to ensure that $q_{x}$ and $q_{y}$ lie over $p$ and not some other outer equality between $A_{1}$ and $A_{2}$.
  We can view $\Id$ in $\CC$ as an operation from $(A : \ty) \times (x : \tm\ A) \times (y : \tm\ A)$ to $\ty$.
  From $p$, $q_{x}$ and $q_{y}$, we obtain some outer equality between $(A_{1}, x_{1}, y_{1})$ and $(A_{2}, x_{2}, y_{2})$, seen as elements of $(A : \ty) \times (x : \tm\ A) \times (y : \tm\ A)$.
  Using the action on outer equalities of $\Id$, we derive an outer equality between $\Id\ \{A_{1}\}\ x_{1}\ y_{1}$ and $\Id\ \{A_{2}\}\ x_{2}\ y_{2}$, as needed.

  For higher-order type-theoretic operations, we need to the outer $\Pi$-types to encode the higher-order arguments.
  For instance, the $\Pi$-type former can be seen as an operation from $(A : \ty) \times (B : \tm\ A \to \ty)$ to $\ty$.
\end{proof}

\subsection{Coherence for non-strict type theories}

\begin{lem}\label{lem:higher_congruence_weq}
  Let $\Init_{w}[X]$ be a cellular model of $\Th_{w}$.
  Assume that there exists a higher congruence $\iota : \Init_{w}[X] \to \CC$ satisfying the following properties:
  \begin{enumerate}
    \item\label{itm:higher_congruence_weq_1} The morphism $\eta^{X} : \Init_{w}[X] \to \Init_{s}[X]$ factors through $\iota$: we have a morphism $G : \CMod_{w}(\CC \to \Init_{s}[X])$ such that $\eta^{X} = \iota \cdot G$.
    \item\label{itm:higher_congruence_weq_2} For every pair of inner types $A,B : \Ty_{\CC}$, if $\Id^{o}\ A\ B$ is inhabited in $\CC$, then $G\ A = G\ B$, and similarly, for every pair of inner terms $a,b : \Tm_{\CC}\ A$, if $\Id^{o}\ a\ b$ is inhabited in $\CC$, then $G\ a = G\ b$.
    \item\label{itm:higher_congruence_weq_0} The higher congruence $\CC$ includes the equations of $\Th_{e}$.
    \item\label{itm:higher_congruence_weq_3} The model $\CC$ of $\Th_{w,2}$ is acyclic.
  \end{enumerate}
  Then $\eta^{X} : \Init_{w}[X] \to \Init_{s}[X]$ is a weak contextual equivalence.
\end{lem}
\begin{proof}
  We consider the congruence $\widetilde{\CC}$ constructed in \cref{con:cong_from_higher_cong}.
  The congruence $\widetilde{\CC}$ does not meet the conditions of \cref{lem:congruence_s}, but only because $\CC$ is not a contextual model, and we have not shown that non-contextual models admit quotients.
  However, we can simply restrict our congruence to the contextual core $\cxl\ \CC$.
  We write $j$ for the inclusion $j : \cxl\ \CC \to \CC$.

  We obtain a quotient $\Quot_{\widetilde{\CC}}$ that is a model of $\Th_{s}$, and the quotient inclusion $\quot_{\widetilde{\CC}} : \cxl\ \CC \to \Quot_{\widetilde{\CC}}$ is a strong contextual equivalence. The assumption \ref{itm:higher_congruence_weq_2} implies that $\widetilde{\CC} \subseteq \ker (j \cdot G)$.
  We can then conclude similarly to the proof of \cref{lem:congruence_seq}: $(j \cdot G) : \cxl\ \CC \to \Init_{s}[X]$ factors through some map $r : \Quot_{\widetilde{\CC}} \to \Init_{s}[X]$, and since $\Quot_{\widetilde{\CC}}$ is a model of $\Th_{s}$, the universal property of $\Init_{s}[X]$ yields a section $s : \Init_{s}[X] \to \Quot_{\widetilde{\CC}}$ of $r$.

  \begin{center}\(\begin{tikzcd}
      & \cxl\ \CC \ar[d, "\rotatebox{90}{$\sim$}"', "j"] \ar[rrd, "\quot_{\widetilde{\CC}}", "\sim"'] &&& \\
      & \CC \ar[rd, "G"] & & \Quot_{\widetilde{\CC}} \ar[ld, shift right, swap, "r"] & \\
      \Init_{w}[X] \ar[ruu, bend left, "\iota'"] \ar[ru, "\iota", "\sim"'] \ar[rr, swap, "\eta^{X}"] && \Init_{s}[X] \ar[ru, shift right, swap, "s"] &
    \end{tikzcd}\)\end{center}
  Since $\Init_{w}[X]$ is contextual, $\iota$ factors through some map $\iota' : \Init_{w}[X] \to \cxl\ \CC$.
  By $2$-out-of-$3$, $\iota'$ is also a weak equivalence.
  By the universal property of $\Init_{w}[X]$, we also have an equality $\eta^{X} \cdot s = \iota' \cdot \quot_{\widetilde{\CC}}$.

  We can now check that $\eta^{X}$ is a retract of $\iota' \cdot \quot_{\widetilde{\CC}}$.
  Indeed, the following diagram commutes.
  \begin{center}\(\begin{tikzcd}
      \Init_{w}[X] \ar[r, equal] \ar[d, "\eta^{X}"] &
      \Init_{w}[X] \ar[r, equal] \ar[d, "\iota' \cdot \quot_{\widetilde{\CC}}"] &
      \Init_{w}[X] \ar[d, "\eta^{X}"] \\
      \Init_{s}[X] \ar[r, "s"] &
      \Quot_{\widetilde{\CC}} \ar[r, "r"] &
      \Init_{s}[X]
    \end{tikzcd}\)\end{center}

  Since $\iota' \cdot \quot_{\widetilde{\CC}}$ is a composition of two weak equivalences, and weak equivalences are closed under compositions and retracts, $\eta^{X} : \Init_{w}[X] \to \Init_{s}[X]$ is a weak equivalence.
\end{proof}

For every cellular model $\Init_{w}[X]$ of $\Th_{w}$, there is a corresponding cellular model $\Init_{w,2,e}[X]$ of $\Th_{w,2,e}$, defined as the image of $\Init_{w}[X]$ by the left adjoint $L_{w,2,e} : \CMod_{w} \to \CMod_{w,2,e}$ to the forgetful functor $R_{w,2,e} : \CMod_{w,2,e} \to \CMod_{w}$.
We denote the universal morphism in $\CMod_{w}(\Init_{w}[X] \to \Init_{w,2,e}[X])$ by $\eta_{w}^{X}$.
In practice, we will apply \cref{lem:higher_congruence_weq} to the models $\Init_{w,2,e}[X]$, which can be seen as the higher congruences freely generated by the equations of $\Th_{e}$.
First, we show that the morphism $\eta^{X} : \Init_{w}[X] \to \Init_{s}[X]$ factors through $\eta_{w}^{X} : \Init_{w}[X] \to \Init_{w,2,e}[X]$.

\begin{con}\label{con:init_s_two}
  Any model $\CC : \CMod_{w}$ is canonically extended to a model $\CC^{=}$ of $\Th_{w,2}$, whose outer identity types correspond to the definitional equality of $\CC$.
  We construct it in the internal language of $\CPsh\ \CC$.
  The outer family is given by the presheaf universe $\SPsh_{\CC}$.
  \begin{alignat*}{2}
    & \Ty_{\CC^{=}}^{o} && \triangleq \SPsh_{\CC} \\
    & \Tm_{\CC^{=}}^{o}\ A && \triangleq A
    \end{alignat*}
  The outer type families $\tm$ is defined by seeing the presheaf family $\Tm$ as a family of elements of the presheaf universe.
    \begin{alignat*}{2}
      & \tm\ A && \triangleq \Tm\ A
    \end{alignat*}
  The outer identity types and outer $\Pi$-types are given by the extensional equality types and dependent function types of the presheaf model $\CPsh\ \CC$.
  \begin{alignat*}{2}
    & \Id^{o}\ x\ y && \triangleq (x = y) \\
    & \Pi^{o}\ A\ B && \triangleq (x : \Tm\ A) \to B\ x
  \end{alignat*}
  As we have isomorphisms of presheaves $\Singl^{o}\ x \simeq \Unit$ and the terminal presheaf $\Unit$ is representable, the outer family has representable singletons, as required.

  Whenever $\CC$ is actually a model of $\Th_{s}$, then $\CC^{=}$ includes the equations of $\Th_{e}$.
\end{con}

The construction of \cref{con:init_s_two} can be seen as a variant of the Yoneda embedding $\yo : \CC \to \widehat{\CC}$.
The presheaf category $\widehat{\CC}$ is one of the intended models of two-level type theory~\cite{2LTT}, with $\Th_{w}$ or $\Th_{s}$ as the inner theory and extensional type theory as the outer theory.
The Yoneda embedding is a morphism of models of $\Th_{w}$ or $\Th_{s}$ (and a contextual isomorphism).
Because the sort $\Tm^{o}$ is not representable in the theory $\Th_{w,2}$, the construction of $\CC^{=}$ can stay over the category $\CC$ instead of moving to the presheaf category $\widehat{\CC}$.

For any cellular model $\Init_{w}[X]$, by the universality of the arrow $\eta^{X}_{w} : \Init_{w}[X] \to \Init_{w,2,e}[X]$ and \cref{con:init_s_two}, we have a morphism $\eta_{s}^{X} : \Init_{w,2,e}[X] \to (\Init_{s}[X])^{=}$ of models of $\Th_{w,2,e}$ such that $\eta_{w}^{X} \cdot \eta_{s}^{X} = \eta^{X}$ in $\CMod_{w}(\Init_{w}[X] \to \Init_{s}[X])$.

\begin{center}\(\begin{tikzcd}
    \Init_{w}[X] \ar[rr, "\eta^{X}"] \ar[rd, swap, "\eta_{w}^{X}"] && \Init_{s}[X] \\
    & \Init_{w,2,e}[X] \ar[ru, "\eta_{s}^{X}", swap] &
  \end{tikzcd}\)\end{center}

\begin{thm}\label{thm:higher_congruence_weq_cell}
  Let $\Init_{w}[X]$ be a cellular model of $\Th_{w}$.
  If the map $\eta_{w}^{X} : \Init_{w}[X] \to \Init_{w,2,e}[X]$ is a weak contextual equivalence and the model $\Init_{w,2,e}[X]$ is acyclic, then the map $\eta^{X} : \Init_{w}[X] \to \Init_{s}[X]$ is a weak contextual equivalence.
\end{thm}
\begin{proof}
  Since $\eta_{w}^{X}$ is a weak contextual equivalence, $\Init_{w,2,e}[X]$ is indeed a higher congruence on $\Init_{w}[X]$.
  We check that it satisfies the conditions of \cref{lem:higher_congruence_weq}.
  We have just proven that condition \ref{itm:higher_congruence_weq_1} holds: $\eta^{X} : \Init_{w}[X] \to \Init_{s}[X]$ factors through $\eta_{w}^{X} : \Init_{w}[X] \to \Init_{w,2,e}[X]$, as $\eta^{X} = \eta_{w}^{X} \cdot \eta_{s}^{X}$.

  The morphism $\eta_{s}^{X} : \Init_{w,2,e}^{X} \to (\Init_{s}[X])^{=}$ preserves outer equalities, which implies that condition \ref{itm:higher_congruence_weq_2} of \cref{lem:higher_congruence_weq} is satisfied.

  The last two conditions of \cref{lem:higher_congruence_weq}, namely the facts that $\Init_{w,2,e}[X]$ includes the equations of $\Th_{e}$ and is acyclic, hold respectively by definition of $\Init_{w,2,e}[X]$ and by assumption of the present proposition.

  Therefore, by \cref{lem:higher_congruence_weq}, the map $\eta^{X}$ is a weak contextual equivalence.
\end{proof}

\begin{thm}\label{thm:conservativity_nonstrict}
  Let $\Th_{w}$ be a type theory with a cumulative hierarchy of universes and weak identity types, and let $\Th_{s}$ be an extension of $\Th_{w}$ by a family of equations $\Th_{e}$.

  If, for every cellular model $\Init_{w}[X]$ of $\Th_{w}$, the map $\eta_{w}^{X} : \Init_{w}[X] \to \Init_{w,2,e}[X]$ is a weak equivalence and the model $\Init_{w,2,e}[X]$ of $\Th_{w,2}$ is acyclic, then the type theories $\Th_{w}$ and $\Th_{s}$ are Morita equivalent.
\end{thm}
\begin{proof}
  By \cref{prop:weq_theories_char} and \cref{thm:higher_congruence_weq_cell}.
\end{proof}

\section{Existence of freely generated higher congruences}\label{sec:first_weq}

In the previous section, we have established that in order to prove the conservativity of the extension of a theory $\Th_{w}$ by a family of equations $\Th_{e}$, it suffices to check two conditions for each cellular model $\Init_{w}[X]$ of $\Th_{w}$: that the map $\eta_{w}^{X} : \Init_{w}[X] \to \Init_{w,2,e}[X]$ is a weak contextual equivalence and that the model $\Init_{w,2,e}[X]$ is acyclic.

In this section, we investigate the first of these two conditions, which we view as a way to state that the higher congruences freely generated by the equations of $\Th_{e}$ exist.
Our claim is that this condition does not really depend on the equational extension $\Th_{e}$, but rather on the fact that the internal equalities of the weak theory $\Th_{w}$ are well-behaved.
We have already said in \cref{sec:homotopy_theory} that the well-behavedness of the internal equalities of $\Th_{w}$ can be tested by the fact that $\Th_{w}$ is semi-model.
We conjecture that these two ways of expressing the well-behavedness of the internal equalities of $\Th_{w}$ are actually equivalent.

\begin{conj}\label{conj:model_to_weq}
  Let $\Th_{w}$ be a theory over the theory of cumulative CwFs with universes and weak identity types.
  The theory $\Th_{w}$ is semi-model if and only if for every cellular model $\Init_{w}[X]$ and equational extension $\Th_{e}$ of $\Th_{w}$, the canonical morphism $\eta^{X}_{w} : \Init_{w}[X] \to \Init_{w,2,e}[X]$ is a weak contextual equivalence.
  \defiEnd
\end{conj}

Instead of proving the full conjecture, we will only prove that the maps $\eta_{w}^{X} : \Init_{w}[X] \to \Init_{w,2,e}[X]$ are weak equivalences under some additional assumptions:
\begin{itemize}
  \item We assume that the type theory $\Th_{w}$ includes $\Pi$-types with a strict $\beta$-rule.
  \item We assume that $\Init_{w,2,e}[X]$ is acyclic.
    This fact is required anyway to apply \cref{thm:conservativity_nonstrict}, and it simplifies the proof quite a bit.
  \item We assume that the computation rules of identity types are marked in $\Th_{e}$.
\end{itemize}

\begin{con}\label{con:model_simeq}
  Let $\Th_{w}$ be a theory over the theory of cumulative CwFs with universes, weak identity types and $\Pi$-types with a strict $\beta$-rule.

  Let $\CC$ be a model of $\Th_{w}$.
  We extend $\CC$ to a model $\CC^{\simeq}$ of $\Th_{w,2,e}$.
\end{con}
\begin{proof}[Construction]
  We work internally to $\CPsh\ \CC$. \\
  The outer types are the coproduct of the presheaves of inner types at each universe level.
  \begin{alignat*}{2}
    & \Ty_{\CC^{\simeq}}^{o} && \triangleq (n : \Nat) \times \Ty_{\CC,n} \\
    & \Tm_{\CC^{\simeq}}^{o}\ (n, A) && \triangleq \Tm_{\CC,n}\ A
  \end{alignat*}

  The outer types of codes for the inner terms are:
  \begin{alignat*}{2}
    & \tm\ A && \triangleq (n, A)
  \end{alignat*}
  The outer identity types are interpreted by the inner identity types.
  \begin{alignat*}{2}
    & \Id^{o}\ \{(n, A)\}\ x\ y && \triangleq (n, \Id\ \{A\}\ x\ y)
  \end{alignat*}
  And the outer $\Pi$-types are interpreted by the inner $\Pi$-types.
  \begin{alignat*}{2}
    & \Pi^{o} && : \{n : \Nat\} (A : \Ty_{\CC,n}) (B : \Tm_{\CC,n}\ A \to \Ty_{\CC}^{o}) \to \Ty_{\CC}^{o} \\
    & \Pi^{o}\ A\ (m, B) && \triangleq (\mathsf{max}\ n\ m, \Pi\ A\ B)
  \end{alignat*}
  Above, $A : \Ty_{\CC,n}$ is an inner type of $\CC$ at level $n$, $B$ is a dependent type over $A$ at level $m$, and the outer $\Pi$-type $\Pi^{o}\ A\ (m, B)$ is given by the inner $\Pi$-type at level $\mathsf{max}\ n\ m$.

  Finally, we have to provide an interpretation for every marked equation of $\Th_{e}$.
  Take a cellular model $\Init_{w}[X]$ and a marked equation $p : \abs{\Tm_{\Init_{w}[X]}}\ (\Id\ \{A\}\ a\ b)$.
  For every morphism $F : \Init_{w}[X] \to (\CC \sslash \Gamma)$, we have to construct
  $\widehat{p} : \abs{\Tm^{o}_{\CC^{\simeq}}}_{\Gamma}\ (\Id^{o}\ (F\ a)\ (F\ b))$ and
  $\widetilde{p} : \abs{\Tm^{o}_{\CC^{\simeq}}}_{\Gamma}\ (\Id^{o}\ (F\ p)\ [\widehat{p}])$.
  Since the outer identity types of $\CC^{\simeq}$ are interpreted by the inner identity types, we can simply define $\widehat{p} \triangleq F\ p$.

  The term $\widetilde{p}$ should then be an inhabitant of the outer type $\Id^{o}\ (F\ p)\ [F\ p]$.
  Recall that $[F\ p] : \Tm\ (\Id\ (F\ a)\ (F\ b))$ is defined as the outer transport of $\refl\ \{F\ a\}$ over $F\ p$.
  Now that outer transport coincides with inner transport, this is just the composition $\refl\ \{F\ a\} \cdot F\ p$.
  Thus $\widetilde{p}$ is simply an instance of the left identity law for path composition.
\end{proof}

\begin{thm}\label{thm:pi_first_weq}
  Let $\Th_{w}$ be a theory over the theory of cumulative CwFs with universes, weak identity types and $\Pi$-types with a strict $\beta$-rule.
  Let $\Th_{s}$ be the equational extension of $\Th_{w}$ by a family of equations $\Th_{e}$, such that the computation rule of identity types is marked in $\Th_{e}$.
 
  Given a cellular model $\Init_{w}[X]$ of $\Th_{w}$, if the model $\Init_{w,2,e}[X]$ of $\Th_{w,2,e}$ is acyclic, then the map $\eta^{X}_{w} : \Init_{w}[X] \to \Init_{w,2,e}[X]$ is a weak contextual equivalence
\end{thm}
\begin{proof}
  We will denote the components of $\Init_{w}[X]$ and $\Init_{w,2,e}[X]$ by $\Ty_{w}$, $\Ty_{w,2,e}$, $\Tm_{w}$, $\Tm_{w,2,e}$, etc.

  From the model $(\Init_{w}[X])^{\simeq}$ of $\Th_{w,2,e}$ constructed in \cref{con:model_simeq}, we obtain a retraction $r : \Init_{w,2,e}[X] \to (\Init_{w}[X])^{\simeq}$ of the map $\eta_{w}^{X} : \Init_{w}[X] \to \Init_{w,2,e}[X]$.
  From this data we know that given any type $A$ of $\Init_{w}[X]$ and term $a$ of type $\eta^{X}_{w}\ A$, there is some term $a_{0}$ of type $A$, namely $a_{0} \triangleq r\ a$.
  In order to show the weak term lifting property, it remains to show that $\eta^{X}_{w}\ a_{0}$ and $a$ are always equal up to inner equality.
  In fact we will show that they are even equal up to outer equality.

  We construct a model $\Init_{w,2,e}[X]^{\bullet}$ of $\Th_{w,2,e}$ for this purpose.
  We present it as a displayed model over $\Init_{w,2,e}[X]$, which means that all of its components depend on the corresponding components of $\Init_{w,2,e}[X]$ (we say that they are displayed over the base components).

  We will present $\Init_{w,2,e}[X]^{\bullet}$ using the internal language of $\CPsh\ \Init_{w,2,e}[X]$.
  For this purpose, we make use of the Yoneda embedding $\yo : \Init_{w,2,e}[X] \to \CPsh\ \Init_{w,2,e}[X]$.
  An inner type $A : \abs{\Ty_{w,2,e}}_{\Gamma}$ can be represented internally by a global element $A : \yo_{\Gamma} \to \Ty_{w,2,e}$, and inner terms and outer types and terms can be represented similarly.

  We also need to internalize the actions of the map $(r \cdot \eta^{X}_{w}) : \Init_{w,2,e}[X] \to \Init_{w,2,e}[X]$.
  In order to simplify the notations, we will just write $r$ for the internalized actions of this map.
  Note that any context in the image of $\eta^{X}_{w}$ can be seen as a telescope of inner types.
  Thus given any context $\Gamma$ of $\Init_{w,2,e}$, we have a global element $(r\ \Gamma) : \Ty_{w,2,e}^{\star}$.
  Given a morphism $f : \Init_{w,2,e}(\Delta \to \Gamma)$, we have a global natural transformation $(r\ f) : \Tm_{w,2,e}^{\star}\ (r\ \Delta) \to \Tm_{w,2,e}^{\star}\ (r\ \Delta)$.
  Given an inner type $A : \yo_{\Gamma} \to \Ty_{w,2,e}$, we have $(r\ A) : \Tm_{w,2,e}^{\star}\ (r\ \Gamma) \to \Ty_{w,2,e}$.
  Given an inner term $a : (\gamma : \yo_{\Gamma}) \to \Tm_{w,2,e}\ (A\ \gamma)$, we have $(r\ a) : (\gamma : \Tm_{w,2,e}^{\star}\ (r\ \Gamma)) \to \Tm_{w,2,e}\ (r\ A\ \gamma)$.
  Given an outer type $A : \yo_{\Gamma} \to \Ty^{o}_{w,2,e}$, we have $(r\ A) : \Tm_{w,2,e}^{\star}\ (r\ \Gamma) \to \Ty_{w,2,e}$.
  Given an outer term $a : (\gamma : \yo_{\Gamma}) \to \Tm^{o}_{w,2,e}\ (A\ \gamma)$, we have $(r\ a) : (\gamma : \Tm_{w,2,e}^{\star}\ (r\ \Gamma)) \to \Tm_{w,2,e}\ (r\ A\ \gamma)$.

  Since we assume that $\Init_{w,2,e}[X]$ is acyclic, we can make use of the internal fibrant congruence $\widetilde{\Init_{w,2,e}[X]}$ defined in \cref{con:cong_from_higher_cong}.
  We will just write $(x \sim y)$ when two types (or terms) $x$ and $y$ are congruent in $\widetilde{\Init_{w,2,e}[X]}$.
  We will write $(\Ty_{w,2,e} / {\sim})$ and $(\Tm_{w,2,e} / {\sim})$ for the components of the internal quotient of $\widetilde{\Init_{w,2,e}[X]}$, and we will implicitly coerce from the inner family to that quotient.

  We now give all of the components of $\Init_{w,2,e}[X]^{\bullet}$.
  \begin{itemize}
    \item A displayed context of $\Init_{w,2,e}[X]^{\bullet}$ over $\Gamma : \abs{\Init_{w,2,e}[X]}$ is given by a global element
      \[ \alpha : \yo_{\Gamma} \to (\Tm_{w,2,e}^{\star}/{\sim})\ (r\ \Gamma). \]
    \item A displayed morphism from the displayed context $\alpha : \yo_{\Gamma} \to (\Tm_{w,2,e}^{\star}/{\sim})\ (r\ \Gamma)$ to $\beta : \yo_{\Delta} \to (\Tm_{w,2,e}^{\star}/{\sim})\ (r\ \Delta)$ over a morphism $f : \Init_{w,2,e}[X](\Gamma \to \Delta)$ is given by a family of equalities
      \[ f^{\bullet} : (\delta : \yo_{\Delta}) \to r\ f\ (\beta\ \delta) = \alpha\ (\yo_{f}\ \delta). \]
    \item A displayed inner type $A^{\bullet}$ in a displayed context $\alpha : \yo_{\Gamma} \to (\Tm_{w,2,e}^{\star}/{\sim})\ (r\ \Gamma)$ over an inner type $A : \yo_{\Gamma} \to \Ty_{w,2,e}$ is a quotiented outer equality
      \[ A^{\bullet} : (\gamma : \yo_{\Gamma}) \to A\ \gamma \sim r\ A\ (\alpha\ \gamma). \]
      Similarly, a displayed outer term $a^{\bullet}$ over an inner term $a : (\gamma : \yo_{\Gamma}) \to \Tm_{w,2,e}\ (A\ \gamma)$ is a quotiented outer equality
      \[ a^{\bullet} : (\gamma : \yo_{\Gamma}) \to a\ \gamma \sim r\ a\ (\alpha\ \gamma). \]
    \item The extension of a displayed context $\alpha : \yo_{\Gamma} \to (\Tm_{w,2,e}^{\star}/{\sim})\ (r\ \Gamma)$ by a displayed inner type $A^{\bullet}$ should be a global element
      \[ \alpha' : \yo_{\Gamma \rhd A} \to (\Tm^{\star}_{w,2,e}/{\sim})\ (r\ (\Gamma \rhd A)). \]
      By the definition of context extensions in $\Init_{w,2,e}[X]$, we have an isomorphism
      \[ \yo_{\Gamma \rhd A} \simeq (\gamma : \yo_{\Gamma}) \times \Tm_{w,2,e}\ (A\ \gamma) \]
      and by definition of telescopes, we have an isomorphism
      \[ (\Tm^{\star}_{w,2,e}/{\sim})\ (r\ (\Gamma \rhd A)) \simeq (\gamma : (\Tm^{\star}_{w,2,e}/{\sim})\ (r\ \Gamma)) \times (\Tm_{w,2,e}/{\sim})\ (r\ A\ \gamma). \]
      Thus, up to these isomorphisms, $\alpha'$ can be defined by
      \[ \alpha'\ (\gamma, a) \triangleq (\alpha\ \gamma, a), \]
      where $a : \Tm_{w,2,e}\ (A\ \gamma)$ is transported to an element of $(\Tm_{w,2,e}/{\sim})\ (r\ A\ (\alpha\ \gamma))$ thanks to $A^{\bullet}\ \gamma : (A\ \gamma) \sim r\ A\ (\alpha\ \gamma)$.
    \item A displayed outer type in a displayed context $\alpha : \yo_{\Gamma} \to (\Tm_{w,2,e}^{\star}/{\sim})\ (r\ \Gamma)$ over an outer type $A : \yo_{\Gamma} \to \Ty^{o}_{w,2,e}$ is a natural transformation
      \[ A^{\bullet} : (\gamma : \yo_{\Gamma}) \to \Tm^{o}_{w,2,e}\ (A\ \gamma) \to (\Tm_{w,2,e}^{\star}/{\sim})\ (r\ A\ (\alpha\ \gamma)). \]
      A displayed outer term of type $A^{\bullet}$ over an outer term $a : (\gamma : \yo_{\Gamma}) \to \Tm^{o}_{w,2,e}\ (A\ \gamma)$ is a quotiented outer equality
      \[ a^{\bullet} : (\gamma : \yo_{\Gamma}) \to A^{\bullet}\ \gamma\ (a\ \gamma) \sim r\ a\ (\alpha\ \gamma). \]
    \item We now define the displayed outer identity type.
      Take a displayed context $\alpha : \yo_{\Gamma} \to (\Tm_{w,2,e}^{\star}/{\sim})\ (r\ \Gamma)$, a displayed type
      \[ A^{\bullet} : (\gamma : \yo_{\Gamma}) \to \Tm^{o}_{w,2,e}\ (A\ \gamma) \to (\Tm_{w,2,e}^{\star}/{\sim})\ (r\ A\ (\alpha\ \gamma)). \]
      over an outer type $A : \yo_{\Gamma} \to \Ty^{o}_{w,2,e}$ and displayed outer terms $x^{\bullet}$ and $y^{\bullet}$.
      The displayed outer identity type $(\Id^{o\bullet}\ \{A^{\bullet}\}\ x^{\bullet}\ y^{\bullet})$ is defined by:
      \begin{alignat*}{3}
        & (\Id^{o\bullet}\ \{A^{\bullet}\}\ x^{\bullet}\ y^{\bullet}) && :{ } && (\gamma : \yo_{\Gamma}) \to \Tm^{o}_{w,2,e}\ (\Id^{o}\ \{A\ \gamma\}\ (x\ \gamma)\ (y\ \gamma)) \to \\
        &&&&& (\Tm_{w,2,e}^{\star}/{\sim})\ (\Id\ \{r\ A\ (\alpha\ \gamma)\}\ (r\ x\ (\alpha\ \gamma))\ (r\ y\ (\alpha\ \gamma))) \\
        & (\Id^{o\bullet}\ \{A^{\bullet}\}\ x^{\bullet}\ y^{\bullet})\ \gamma\ p && \triangleq{ } && \refl\ (r\ x\ (\alpha\ \gamma)),
      \end{alignat*}
      where the above line is well-typed thanks to the fact that $(r\ x\ (\alpha\ \gamma)) \sim (r\ y\ (\alpha\ \gamma))$, which is derived from $x^{\bullet}\ \gamma : A^{\bullet}\ \gamma\ (x\ \gamma) \sim r\ x\ (\alpha\ \gamma)$, $y^{\bullet}\ \gamma : A^{\bullet}\ \gamma\ (y\ \gamma) \sim r\ y\ (\alpha\ \gamma)$ and $p : x\ \gamma \sim y\ \gamma$.

      The reflexivity outer term
      \begin{alignat*}{3}
        & (\refl^{o\bullet}\ \{A^{\bullet}\}\ x^{\bullet}) && :{ } && (\gamma : \yo_{\Gamma}) \to (\Id^{o\bullet}\ \{A^{\bullet}\}\ x^{\bullet}\ x^{\bullet})\ \gamma\ (\refl^{o}\ (x\ \gamma)) \sim \refl\ (r\ x\ (\alpha\ \gamma)),
      \end{alignat*}
      follows from the reflexivity of $(\sim)$.
    \item The extension of contexts by outer singletons is defined similarly to the extension by inner terms.
      Given a displayed context $\alpha : \yo_{\Gamma} \to (\Tm_{w,2,e}^{\star}/{\sim})\ (r\ \Gamma)$, a displayed outer type $A^{\bullet}$ and a displayed outer term $x^{\bullet}$, we need to define
      $\alpha' : \yo_{\Gamma\rhd \Singl^{o}\ x} \to (\Tm_{w,2,e}^{\star}/{\sim})\ (r\ (\Gamma \rhd \Singl^{o}\ x))$.
      We have isomorphisms
      \[ \yo_{\Gamma\rhd \Singl^{o}\ x} \simeq (\gamma : \yo_{\Gamma}) \times (y : \Tm_{w,2,e}^{o}\ (A\ \gamma)) \times (p : \Tm_{w,2,e}^{o}\ (\Id^{o}\ (x\ \gamma)\ y)), \]
      and
      \begin{alignat*}{3}
        & (\Tm_{w,2,e}^{\star}/{\sim})\ (r\ (\Gamma \rhd \Singl^{o}\ x)) && \simeq{ } && (\gamma : (\Tm_{w,2,e}^{\star}/{\sim})\ (r\ \Gamma)) \times{ } \\
        &&&&& (y : (\Tm_{w,2,e}/{\sim})\ (r\ A\ \gamma)) \times {} \\
        &&&&& (p : (\Tm_{w,2,e}/{\sim})\ (\Id\ (r\ x\ \gamma)\ y)).
      \end{alignat*}
      We can therefore define
      \[ \alpha'\ (\gamma, y, p) \triangleq (\alpha\ \gamma, A^{\bullet}\ \gamma\ y, \refl\ (r\ x\ (\alpha\ \gamma))), \]
      where the fact that $\refl\ (r\ x\ (\alpha\ \gamma))$ is in $(\Tm_{w,2,e}/{\sim})\ (\Id\ (r\ x\ (\alpha\ \gamma))\ (A^{\bullet}\ \gamma\ y))$ follows from $x^{\bullet}\ \gamma : A^{\bullet}\ \gamma\ (x\ \gamma) \sim (r\ x\ (\alpha\ \gamma))$ and $p : x\ \gamma \sim y$.

      To define the outer identity type elimination structure over a displayed context $\alpha : \yo_{\Gamma} \to (\Tm_{w,2,e}^{\star}/{\sim})\ (r\ \Gamma)$, take a displayed outer type $A^{\bullet}$, a displayed outer term $x^{\bullet}$, a displayed outer type
      \begin{alignat*}{3}
        & P^{\bullet} && :{ } && (\gamma : \yo_{\Gamma})\ (y : \Tm^{o}_{w,2,e}\ (A\ \gamma))\ (p : \Tm^{o}_{w,2,e}\ (\Id^{o}\ (x\ \gamma)\ y)) \to \\
        &&&&& \Tm^{o}_{w,2,e}\ (P\ \gamma\ y\ p) \to (\Tm_{w,2,e}^{\star}/{\sim})\ (r\ P\ (\alpha'\ (\gamma, y, p)))
      \end{alignat*}
      over the displayed context $\alpha'$ defined above, and displayed outer terms
      \[ d^{\bullet} : (\gamma : \yo_{\Gamma}) \to P^{\bullet}\ \gamma\ (x\ \gamma)\ (\refl^{o}\ x)\ (d\ \gamma) \sim r\ d\ (\alpha\ \gamma), \]
      $y^{\bullet}$ and $p^{\bullet}$.

      We need to construct
      \[ \J^{o\bullet} : (\gamma : \yo_{\Gamma}) \to P^{\bullet}\ \gamma\ (y\ \gamma)\ (p\ \gamma)\ (\J^{o}\ (P\ \gamma)\ (d\ \gamma)\ (y\ \gamma)\ (p\ \gamma)) \sim r\ (\J^{o}\ P\ d\ y\ p)\ (\alpha\ \gamma). \]
      Fix $\gamma : \yo_{\Gamma}$. By outer path induction, it suffices to show that
      \[ P^{\bullet}\ \gamma\ (x\ \gamma)\ (\refl^{o}\ x)\ (\J^{o}\ (P\ \gamma)\ (d\ \gamma)\ (x\ \gamma)\ (\refl^{o}\ x)) \sim r\ (\J^{o}\ P\ d\ y\ p)\ (\alpha\ \gamma). \]
      On the left hand side of this equation, we have
      \[ P^{\bullet}\ \gamma\ (x\ \gamma)\ (\refl^{o}\ x)\ (\J^{o}\ (P\ \gamma)\ (d\ \gamma)\ (x\ \gamma)\ (\refl^{o}\ x)) \sim P^{\bullet}\ \gamma\ (x\ \gamma)\ (\refl^{o}\ x)\ (d\ \gamma) \]
      by the weak computation rule for outer identity types and
      \[ d^{\bullet} : P^{\bullet}\ \gamma\ (x\ \gamma)\ (\refl^{o}\ x)\ (d\ \gamma) \sim r\ d\ (\alpha\ \gamma). \]
      On the right hand side, we can compute
      \[ r\ (\J^{o}\ P\ d\ y\ p)\ (\alpha\ \gamma) = \J\ (r\ P\ (\alpha\ \gamma))\ (r\ d\ (\alpha\ \gamma))\ (r\ y\ (\alpha\ \gamma))\ (r\ p\ (\alpha\ \gamma)). \]
      By $p^{\bullet}$, we have $r\ p\ (\alpha\ \gamma) \sim \refl\ (r\ x\ (\alpha\ \gamma))$, so we can deduce
      \[ r\ (\J^{o}\ P\ d\ y\ p)\ (\alpha\ \gamma) \sim \J\ (r\ P\ (\alpha\ \gamma))\ (r\ d\ (\alpha\ \gamma))\ (r\ x\ (\alpha\ \gamma))\ \refl, \]
      and since the computation rule of identity types is marked in $\Th_{e}$, we derive
      \[ r\ (\J^{o}\ P\ d\ y\ p)\ (\alpha\ \gamma) \sim r\ d\ (\alpha\ \gamma), \]
      completing the derivation of $\J^{o\bullet}$.

      In order to define $\J_{\beta}^{o\bullet}$, we need to prove, for every $\gamma : \yo_{\Gamma}$, that
      \[ r\ (\J^{o}_{\beta}\ P\ d\ y\ p)\ (\alpha\ \Gamma) \sim \refl\ d. \]
      This follows from the definition of $\J^{o}_{\beta}$ in the model $(\Init_{w}[X])^{\simeq}$ and the fact that the computation rule of weak identity types is marked in $\Th_{e}$.

    \item We also have to define the displayed outer $\Pi$-types. Take a displayed context $\alpha : \yo_{\Gamma} \to (\Tm_{w,2,e}^{\star}/{\sim})\ (r\ \Gamma)$, a displayed inner type $A^{\bullet}$ and a displayed outer type
      \begin{alignat*}{3}
        & B^{\bullet} && :{ } && (\gamma : \yo_{\Gamma})\ (a : \Tm_{w,2,e}\ (A\ \gamma)) \to \Tm_{w,2,e}^{o}\ (B\ \gamma\ a) \to \\
        &&&&& (\Tm_{w,2,e}/{\sim})\ (r\ B\ (\alpha\ \gamma)\ (r\ a\ (\alpha\ \gamma))).
      \end{alignat*}
      The displayed outer $\Pi$-types are then defined by
      \begin{alignat*}{3}
        & \Pi^{o\bullet}\ A^{\bullet}\ B^{\bullet} && :{ } && (\gamma : \yo_{\Gamma}) \to \Tm_{w,2,e}^{o}\ (\Pi^{o}\ (A\ \gamma)\ (B\ \gamma)) \to \\
        &&&&& (\Tm_{w,2,e}/{\sim})\ (r\ (\Pi^{o}\ A\ B)\ (\alpha\ \gamma)) \\
        & (\Pi^{o\bullet}\ A^{\bullet}\ B^{\bullet})\ \gamma\ f && \triangleq{ } && \lam\ (a \mapsto B^{\bullet}\ \gamma\ a\ (\app^{o}\ f\ a)),
      \end{alignat*}
      where $r\ (\Pi^{o}\ A\ B)\ (\alpha\ \gamma)$ computes to $\Pi\ (r\ A\ (\alpha\ \gamma))\ (r\ B\ (\alpha\ \gamma))$
      and $A^{\bullet}\ \gamma$ is used to coerce between terms of type $A\ \gamma$ and terms of type $r\ A\ (\alpha\ \gamma)$.

      To define the displayed outer lambda abstraction $\lam^{o\bullet}\ b^{\bullet}$ given a displayed outer term
      $b^{\bullet} : (\gamma : \yo_{\Gamma}) (a : \Tm_{w,2,e}\ (A\ \gamma)) \to B^{\bullet}\ \gamma\ a\ (b\ a) \sim r\ b\ (\alpha\ \gamma, a)$,
      we need to prove a quotiented outer equality
      \[ (\gamma : \yo_{\Gamma}) \to \lam\ (a \mapsto B^{\bullet}\ \gamma\ a\ (\app^{o}\ (\lam^{o}\ b)\ a)) \sim r\ (\lam^{o}\ b)\ (\alpha\ \gamma). \]
      Fix $\gamma : \yo_{\Gamma}$.
      We can compute $r\ (\lam^{o}\ b)\ (\alpha\ \gamma) = \lam\ (a \mapsto r\ b\ (\alpha\ \gamma, a))$, so it suffices to show, given $a : \Tm_{w,2,e}\ (A\ \gamma)$, that
      $B^{\bullet}\ \gamma\ a\ (\app^{o}\ (\lam^{o}\ b)\ a) \sim r\ b\ (\alpha\ \gamma, a)$.
      This follows from $b^{\bullet}\ \gamma$ and the fact that the $\beta$-rule for inner $\Pi$-types is strict.

      To define the displayed outer application $\app^{o\bullet}\ f^{\bullet}\ a^{\bullet}$ given displayed terms
      $f^{\bullet} : (\gamma : \yo_{\Gamma}) \to (\Pi^{o\bullet}\ A^{\bullet}\ B^{\bullet})\ \gamma\ f \sim r\ f\ (\alpha\ \gamma)$ and $a^{\bullet} : (\gamma : \yo_{\Gamma}) \to a\ \gamma \sim r\ a\ (\alpha\ \gamma)$,
      we need to prove a quotiented outer equality
      \[ (\gamma : \yo_{\Gamma}) \to B^{\bullet}\ \gamma\ a\ (\app^{o}\ f\ a) \sim r\ (\app^{o}\ f\ a)\ (\alpha\ \gamma). \]
      Fix $\gamma : \yo_{\Gamma}$. We can compute $r\ (\app^{o}\ f\ a)\ (\alpha\ \gamma) = \app\ (r\ f\ (\alpha\ \gamma))\ (r\ a\ (\alpha\ \gamma))$.
      By $f^{\bullet}$ and $a^{\bullet}$ we have $\app\ (r\ f\ (\alpha\ \gamma))\ (r\ a\ (\alpha\ \gamma)) \sim \app\ (\lam\ (a' \mapsto B^{\bullet}\ \gamma\ a'\ (\app^{o}\ f\ a')))\ (a\ \gamma)$.
      We can conclude using the fact that the $\beta$-rule for inner $\Pi$-types is strict.

      We omit the definition of the extensionality structure of the displayed $\Pi$-types.
    \item The definitions of the displayed operations of the inner layer follow from the fact that the congruence $(\sim)$ and the morphisms $r$ and $\eta^{X}_{w}$ preserves these operations.
    \item Because the displayed inner types and terms are propositional in this model, the strict equalities of the inner layer are automatically strict in $\Init_{w,2,e}[X]^{\bullet}$.
    \item Finally, we have to construct displayed outer terms $\widehat{p}$ and $\widetilde{p}$ for every marked equation $p$ of $\Th_{e}$ over a displayed context $\alpha : \yo_{\Gamma} \to (\Tm_{w,2,e}^{\star}/{\sim})\ (r\ \Gamma)$.
      This means that we have to show, for every $\gamma : \yo_{\Gamma}$, that $r\ \widehat{p}\ (\alpha\ \gamma) \sim \refl$ and $r\ \widetilde{p}\ (\alpha\ \gamma) \sim \refl$.

      We have $r\ \widehat{p}\ (\alpha\ \gamma) = r\ p\ (\alpha\ \gamma)$ by definition of $r$, $r\ p\ (\alpha\ \gamma) \sim p\ \gamma$ by induction hypothesis, $p\ \gamma \sim [\widehat{p}\ \gamma]$ is witnessed by the outer equality $\widetilde{p}$, and $[\widehat{p}\ \gamma] \sim \refl$ by definition of $[-]$, so the first equality holds.

      For the second equality, $r\ \widetilde{p}\ (\alpha\ \gamma)$ was defined by some instance $\J_{\beta}$ of the weak computation rule for the inner equality, and $\J_{\beta} \sim \refl$ because $\J_{\beta}$ is marked in $\Th_{e}$.
  \end{itemize}

  Thus we have constructed a displayed model $\Init_{w,2,e}[X]^{\bullet}$ over $\Init_{w,2,e}[X]$, and the universal property of $\Init_{w,2,e}[X]$ gives us a section of $\Init_{w,2,e}[X]^{\bullet}$.

  For every context $\Gamma$ of $\Init_{w,2,e}[X]^{\bullet}$, we obtain $\alpha_{\Gamma} : \yo_{\Gamma} \to (\Tm^{\star}/{\sim})\ (r\ \Gamma)$.
  A direct induction on contexts shows that for every context $\Gamma$ of $\Init_{w}[X]$, the map $\alpha_{(\eta^{X}_{w}\ \Gamma)} : \yo_{\eta^{X}_{w}\ \Gamma} \to (\Tm^{\star}/{\sim})\ (r\ (\eta^{X}_{w}\ \Gamma))$ is essentially the quotienting map $\Tm^{\star}\ (\eta^{X}_{w}\ \Gamma) \to (\Tm^{\star}/{\sim})\ (\eta^{X}_{w}\ \Gamma)$, up to the canonical isomorphism $\yo_{\eta^{X}_{w}\ \Gamma} \simeq \Tm^{\star}\ (\eta^{X}_{w}\ \Gamma)$ and the equality $r \circ \eta^{X}_{w} = \eta^{X}_{w}$.

  Now for every context $\Gamma$ of $\Init_{w}[X]$, inner type $A$ in $\Gamma$ and inner term $a$ of type $\eta^{X}_{w}\ A$, the section of $\Init_{w,2,e}[X]^{\bullet}$ provides a proof of $(\gamma : \yo_{\eta^{X}_{w}\ \Gamma}) \to a\ \gamma \sim r\ a\ (\alpha_{(\eta^{X}_{w}\ \Gamma)}\ \gamma)$.
  Since $\alpha_{(\eta^{X}_{w}\ \Gamma)}$ is essentially the identity map, this proves that $r\ a$ is a weak lift of $a$.

  Thus $\eta^{X}_{w}$ satisfies the weak term lifting property, and is indeed a weak contextual equivalence.
\end{proof}

\begin{thm}\label{thm:conservativity_nonstrict_pi}
  Let $\Th_{w}$ be a theory over the theory of cumulative CwFs with universes, weak identity types and $\Pi$-types with a strict $\beta$-rule.
  Let $\Th_{s}$ be the equational extension of $\Th_{w}$ by a family of equations $\Th_{e}$, such that the computation rule of identity types is marked $\Th_{e}$ .
 
  If every cellular model $\Init_{w}[X]$ is acyclic, then the theories $\Th_{w}$ and $\Th_{s}$ are Morita equivalent.
\end{thm}
\begin{proof}
  By \cref{thm:conservativity_nonstrict} and \cref{thm:pi_first_weq}.
\end{proof}

\subsection*{Towards a proof of conjecture~\ref{conj:model_to_weq}}

The reverse implication of \cref{conj:model_to_weq} should easily be provable as a consequence of the following lemma and general recognition theorems for left semi-model structures.
\begin{lem}
  Let $\Th_{w}$ be some type theory and $\Th_{e}$ be the equational extension consisting of the marked equality
  \[ (A : \UU)(x, y : A)(p : \Id\ x\ y) \to (p : \Id\ x\ y). \]
  Note that this adds the equality reflection rule to the strong type theory $\Th_{s}$ defined by $\Th_{e}$, but we only consider the two-level type theory $\Th_{w,2,e}$ here.

  If for every cellular model $\Init_{w}[X]$, the morphism $\eta^{X}_{w} : \Init_{w}[X] \to \Init_{w,2,e}[X]$ is a weak contextual equivalence, then every basic $J$-cellular extension $j : \Init_{w}[Y] \to \Init_{w}[Y][\Gamma\ \vdash \bm{p} : \Id\ a\ \bm{b}]$ with a cellular source is a weak contextual equivalence.
\end{lem}
\begin{proof}
  Consider the following square.
  \[ \begin{tikzcd}[column sep=60pt]
      \Init_{w}[Y] \ar[r, "\eta^{Y}_{w}"] \ar[d, "j"] & \Init_{w,2,e}[Y] \ar[d, "j'"] \\
      \Init_{w}[Y, \Gamma \vdash \bm{p} : \Id\ a\ \bm{b}] \ar[r, "\eta^{(Y,\bm{p}:\Id\ a\ \bm{b})}_{w}"] & \Init_{w,2,e}[Y, \Gamma \vdash \bm{p} : \Id\ a\ \bm{b}]
    \end{tikzcd} \]
  We want to prove that $j$ is a weak equivalence.
  The horizontal maps are weak equivalences by assumption, so it suffices to check that $j'$ is a weak equivalence.

  The marked equality of $\Th_{e}$ implies that inner and outer equalities are equivalent in models of $\Th_{w,2,e}$, and we can derive from this a weak equivalence
  \[ \Init_{w,2,e}[Y, \Gamma \vdash \bm{p} : \Id\ a\ \bm{b}] \to \Init_{w,2,e}[Y, \Gamma \vdash \bm{p} : \Id^{o}\ a\ \bm{b}]. \]

  Now \cref{prop:weq_pi_singl}, applied to the outer layer, implies that
  \[ j'' : \Init_{w,2,e}[Y] \to \Init_{w,2,e}[Y, \Gamma \vdash \bm{p} : \Id^{o}\ a\ \bm{b}] \]
  is a weak equivalence, and we can conclude by $2$-out-of-$3$.
\end{proof}

The forward implication of \cref{conj:model_to_weq} is significantly more complicated.
We believe that it can be proven by refining \cref{con:model_simeq} and \cref{thm:pi_first_weq} to work without our simplifying assumptions.
\begin{itemize}
  \item If the weak type theory $\Th_{w}$ does not include $\Pi$-types with strict $\beta$, we cannot equip $\Init_{w}[X]$ with the structure of a model of $\Th_{w,2,e}$ to obtain the morphism $r : \Init_{w,2,e}[X] \to \Init_{w}[X]$.
    However, when $\Th_{w}$ is semi-model, we should be able to construct another model $\CM$ of $\Th_{w,2,e}$, along with a contextual isomorphism $\Init_{w}[X] \to \CM$.
    We would then obtain the following diagram.
    \[ \begin{tikzcd}
        \Init_{w}[X] \ar[d, "\eta_{w}^{X}"'] \ar[r, "\sim"] & \CM \\
        \Init_{w,2,e}[X] \ar[ru, "r"']
      \end{tikzcd} \]
    It should then be possible to proceed similarly to \cref{thm:pi_first_weq}.

    There may be several possible constructions of the model $\CM$.
    If there are no marked equations in $\Th_{e}$, we can let $\CM$ be the presheaf category $\widehat{\Init_{w}[X]}$.
    This is used in the conservativity proof for two-level type theory of \cite[Proposition 2.17]{2LTT}.
    This choice does not work if $\Th_{e}$ is non-empty.

    Using something based on space-valued presheaves, instead of set-valued presheaves, could maybe work.
    Our approach, that almost works, is to choose $\CM \triangleq (\CMod_{w}^{\cxl})^{\op}$, the dual of the category of contextual models of $\Th_{w}$.

    Indeed, that category can be equipped with the structure of a model of $\Th_{w}$ as follows.
    An inner type (resp. term) over a context $\CC : \CMod_{w}^{\cxl}$ is a closed type (resp. term) of $\CC$.
    The extension of an context $\CC$ by a type $A$ is the free extension $\CC[\bm{a} : A]$.
    All type-theoretic operations over a context $\CC$ are interpreted by the corresponding type-theoretic operations of $\CC$ at the empty context.

    It can be shown that the unique morphism $\Init_{w} \to (\CMod_{w}^{\cxl})^{\op}$ is a contextual isomorphism.
    More generally, for any contextual model $\CC$, we have a contextual isomorphism $\CC \to (\CMod_{w}^{\cxl} \slash \CC)^{\op}$ defined by induction on the contexts of $\CC$.

    The model $(\CMod_{w}^{\cxl})^{\op}$ can almost be extended to the outer layer of $\Th_{w,2,e}$.
    An outer type over $\CC$ is a cellular extension $i : \CC \to \CC[X]$, and an outer term of $i : \CC \to \CC[X]$ is a retraction of $i$.
    The outer identity type over $i : \CC \to \CC[X]$ is interpreted by a relative cylinder object for $i$.
    The outer identity type eliminator is interpreted using the weak term lifting property of trivial cofibrations, which are weak equivalences when $\Th_{w}$ is semi-model.
    This construction is essentially the same as the homotopy theoretic model of identity types of \cite{HomotopyTheoreticModels}.
    It suffers from the same problem as the model of \cite{HomotopyTheoreticModels}: the eliminator for the outer identity types is not stable under substitution.

    The interpretation of the outer $\Pi$-types is however unproblematic, in particular thanks to the fact that outer types are cellular extensions, rather than arbitrary cofibrations: given a closed type $A$ of $\CC$ and a cellular extension $\CC[\bm{a} : A] \to \CC[\bm{a} : A][X]$, the outer $\Pi$-type is represented as the cellular extension $\CC \to \CC[\bm{a} : A \vdash X]$ (i.e. $(\bm{a} : A)$ is added as an argument of every generating term of $X$).

    To use this construction, we would thus need another coherence theorem, showing that $\Th_{w,2,e}$ is equivalent to its variant with an outer identity type eliminator that is not stable under substitution.
  \item If we don't assume that $\Init_{w,2,e}$ is acyclic, the definition of the displayed model $\Init_{w,2,e}[X]^{\bullet}$ used in \cref{thm:pi_first_weq} would have to be more complicated.
    Many of its components that are propositional in the current proof would become proof-relevant.

    The inclusion of the computation rules for identity types and $\Pi$-types in $\Th_{e}$ is only used to simplify the construction of $\Init_{w,2,e}[X]^{\bullet}$ using the acyclicity of $\Init_{w,2,e}$, and shouldn't be needed in the general case.
\end{itemize}


\section*{Acknowledgements}
The author would like to thank Thorsten Altenkirch, Martin Bidlingmaier, Paolo Capriotti, Thierry Coquand, Simon Huber, Ambrus Kaposi, András Kovács, Nicolai Kraus, Chaitanya Leena Subramaniam, Christian Sattler and Bas Spitters for expressing interest in this work and helpful discussions.

\bibliographystyle{alpha}
\bibliography{main}

\end{document}